\documentclass[acmsmall,authorversion]{acmart}
\usepackage[utf8]{inputenc} % allow utf-8 input
\usepackage[T1]{fontenc}    % use 8-bit T1 fonts
\usepackage{hyperref}       % hyperlinks
\usepackage{url}            % simple URL typesetting
\usepackage{booktabs}       % professional-quality tables
\usepackage{amsfonts}       % blackboard math symbols
\usepackage{nicefrac}       % compact symbols for 1/2, etc.
\usepackage{microtype}      % microtypography
\usepackage{xcolor}         % colors
\usepackage{amsmath}
\usepackage{algorithm}
\usepackage{algorithmicx}
\usepackage{algpseudocode}
\usepackage{bm}
\usepackage{subcaption}
\usepackage{graphicx}
\usepackage{wrapfig}
\usepackage{mdframed} 
\usepackage{comment}
\usepackage{enumitem}
\usepackage{wrapfig}
\usepackage[font=small]{caption}

\AtBeginDocument{%
  \providecommand\BibTeX{{%
    \normalfont B\kern-0.5em{\scshape i\kern-0.25em b}\kern-0.8em\TeX}}}

\setcopyright{rightsretained}
\acmJournal{POMACS}
\acmYear{2024} \acmVolume{8} \acmNumber{1} \acmArticle{17} \acmMonth{3} \acmDOI{10.1145/3639043}

\makeatletter
\gdef\@copyrightpermission{
	\begin{minipage}{0.2\columnwidth}
		\href{https://urldefense.com/v3/__https://creativecommons.org/licenses/by/4.0/}{\includegraphics[width=0.90\textwidth]{figures/4ACM-CC-by-88x31.eps}}
	\end{minipage}\hfill
	\begin{minipage}{0.8\columnwidth}
		\href{https://urldefense.com/v3/__https://creativecommons.org/licenses/by/4.0/}{This work is licensed under a Creative Commons Attribution International 4.0 License.}
	\end{minipage}
	\vspace{5pt}
}
\makeatother

\begin{document}

% \usetikzlibrary{positioning}
\newcommand{\n}{\mathbb{N}}
\newcommand{\z}{\mathbb{Z}}
\newcommand{\q}{\mathbb{Q}}
\newcommand{\cx}{\mathbb{C}}
\newcommand{\real}{\mathbb{R}}
\newcommand{\field}{\mathbb{F}}
\newcommand{\ita}[1]{\textit{#1}}
\newcommand{\com}[2]{#1\backslash#2}
\newcommand{\oneton}{\{1,2,3,...,n\}}
\newcommand\idea[1]{\begin{gather*}#1\end{gather*}}
\newcommand\ef{\ita{f} }
\newcommand\eff{\ita{f}}
\newcommand\proofs[1]{\begin{proof}#1\end{proof}}
\newcommand\inv[1]{#1^{-1}}
\newcommand\setb[1]{\{#1\}}
\newcommand\en{\ita{n }}
\newcommand{\vbrack}[1]{\langle #1\rangle}
\def\argmin{\arg\min}

\def\p{\tilde}
\def\c{\mathcal}
\def\wp{\widetilde}
\def\lam{\lambda}
\def\sumT{\sum_{t=1}^T}

\newcommand{\parentheses}[1]{\left(#1\right)}
\newcommand{\set}[1]{\left\{#1\right\}}
\newcommand{\brackets}[1]{\left[#1\right]}
\newcommand{\reals}{\mathbb{R}}

\theoremstyle{definition}
 
\newcounter{A}  
\newtheorem{ass}[A]{Assumption}

 \newcounter{N}   
\newtheorem{notation}[N]{Notation}

 \newcounter{P} 
\newtheorem{prop}[P]{Proposition}

\newcounter{R} 
\newtheorem{remark}[R]{Remark}

\title[Online Learning for Fair vRANs]{Fair Resource Allocation in Virtualized O-RAN Platforms}

\author{Fatih Aslan}
\email{f.aslan@tudelft.nl}
\orcid{0000-0002-0209-5993}
\author{George Iosifidis}
\email{g.iosifidis@tudelft.nl}
\orcid{0000-0003-1001-2323}
\affiliation{%
  \institution{TU Delft}
  \country{The Netherlands}
}

\author{Jose A. Ayala-Romero}
\email{jose.ayala@neclab.eu}
\orcid{0000-0001-7402-3174}
\author{Andres Garcia-Saavedra}
\email{andres.garcia.saavedra@neclab.eu}
\orcid{0000-0003-2005-2222}
\affiliation{%
  \institution{NEC Labs Europe}
  \country{Germany}}
\author{Xavier Costa-Perez}
\email{xavier.costa@i2cat.net}
\orcid{0000-0002-9654-6109}
\affiliation{%
  \institution{i2CAT, NEC Labs Europe and ICREA}
  \country{Spain}}

\renewcommand{\shortauthors}{Fatih Aslan et al.}


\begin{abstract}
O-RAN systems and their deployment in virtualized general-purpose computing platforms (O-Cloud) constitute a paradigm shift expected to bring unprecedented performance gains. However, these architectures raise new implementation challenges and threaten to worsen the already-high energy consumption of mobile networks. This paper presents first a series of experiments which assess the O-Cloud's energy costs and their dependency on the servers' hardware, capacity and data traffic properties which, typically, change over time. Next, it proposes a compute policy for assigning the base station data loads to O-Cloud servers in an energy-efficient fashion; and a radio policy that determines at near-real-time the minimum transmission block size for each user so as to avoid unnecessary energy costs. The policies balance energy savings with performance, and ensure that both of them are dispersed fairly across the servers and users, respectively. To cater for the unknown and time-varying parameters affecting the policies, we develop a novel online learning framework with fairness guarantees that apply to the entire operation horizon of the system (long-term fairness). The policies are evaluated using trace-driven simulations and are fully implemented in an O-RAN compatible system where we measure the energy costs and throughput in realistic scenarios.
\end{abstract}


\maketitle

\section{Introduction} \label{sec:intro}

\subsection{Background \& Motivation}
One of the most revolutionizing aspects of future mobile networks is the \emph{virtualization} of the Radio Access Network (vRAN), in particular of the base stations (vBS), and the execution of their software functions at general-purpose computing platforms \cite{bib:melodia-commag-21}. Driven by the Open RAN (O-RAN) Alliance, practically the entire Telco industry is currently investing in the development of vBSs, in anticipation of the eclipse of conventional RANs by 2028~\cite{mason-forecast}. Virtualized RANs promote the control of vBSs in (almost) real-time, using new knobs that tailor their operation to the environment, e.g., channel conditions, and to user needs for throughput, latency, and other KPIs. The proposed vRAN architectures typically include computing pools ({O-Cloud}) of heterogeneous processing units (PUs), with CPUs or ASIC/FPGA/GPU hardware accelerators (HAs), which execute dynamically-allocated compute workloads of one or more vBSs~\cite{bib:melodia-tutorial-23}. This native cloud-based architecture constitutes a paradigm shift for RAN and is anticipated to bring unprecedented performance gains \cite{bib:andres-magazine}.

Unfortunately, the virtualization of RAN is expected also to increase the Operating Expenditures (OpEx) of networks due to the high energy consumption of vBSs. Namely, unlike legacy base stations, the energy spent for executing the software vBS functions becomes very relevant and, in fact, can even surpass that of wireless transmissions \cite{bib:vbs-experiments, auer2011much}. Moreover, these costs are volatile and unpredictable, as they depend on a range of factors such as the radio characteristics of the transmitted data (e.g., the Signal-to-Noise Ratio, SNR), and the properties of the O-Cloud PUs. Coupled with the increasing RAN densification, this effect is bound to render the vRAN energy costs --- an already prevalent concern for operators\footnote{For instance, Verizon and Vodafone announced their target for net zero energy emissions by 2040 \cite{bib:energy-gsma}, and China Mobile has set to reduce energy consumption and carbon emissions by 20\% in the next few years \cite{bib:energy-china-mobile}.} --- prohibitively high for future mobile networks. Indeed, there is wide consensus that this is a key obstacle hampering the adoption of vRANs~\cite{mason-tco}, and hence is justifiably very high in the O-RAN agenda of industries \cite{bib:energy-erricson, bib:energy-telefonica}.

A promising method to tackle this issue is to leverage one of the key O-RAN architecture innovations: the RAN Intelligent Controller. The RIC, as commonly termed, provides a centralized abstraction of the network and is envisioned as a powerful enabler for control policies with different objectives, decision granularity, and time-scales \cite{bib:melodia-commag-dApps, bib:andres-magazine}. Interestingly, the RIC policies can shape the performance and energy cost of the vRAN in two ways: \emph{(i)} by assigning carefully the vBSs workloads to different PUs of O-Cloud; and \emph{(ii)} by affecting the characteristics of these workloads in almost real-time. For instance, the RIC could dictate the vBSs to route their most voluminous flows to servers equipped with HAs; to refrain from using energy-costly modulation schemes \cite{vrain_conf, bib:mike-globecom}; or to bound their transmission power \cite{ayala2021bayesian}. Such \emph{compute control} and \emph{radio control} policies can, in principle, be very effective in balancing the vRAN performance and energy costs, but require access to system and user parameters that are unknown and vary rapidly, and presume solving large-scale challenging optimization problems. 

\begin{figure}[t!]
\centering
\minipage{0.47\columnwidth}
\includegraphics[width=\columnwidth]{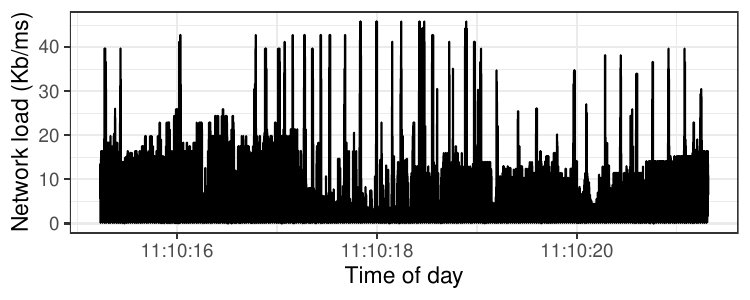}
\vspace{-5mm}
\caption{\small Cell load dynamics (msec granularity) over a few seconds, collected from an operational RAN in Frankfurt, Germany, May 2023.}
\label{fig:intro_01}
 \endminipage{}
 \hfill
 \minipage{0.47\columnwidth}
\includegraphics[width=\columnwidth]{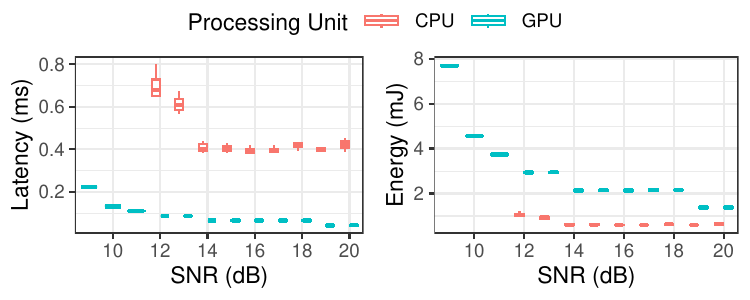}
\vspace{-6mm}
\caption{\footnotesize{Processing latency (left) and energy consumption (right) to process one TB under different SNRs; measured on an Intel Xeon CPU core and an NVIDIA V100 GPU.}}
\label{fig:intro_02}
\endminipage{}
\vspace{-1mm}
\end{figure}

To exemplify, the processing time for decoding/encoding the users' uplink/downlink streams depends on the amount of this traffic which, most often, is subject to rapid fluctuations; see, e.g., \cite{bib:foukas-sigcom21} and our measurements in Fig.~\ref{fig:intro_01}. Secondly, the processing time for each Transport Block (TB)\footnote{TB is the basic MAC-layer user data unit, and its length (bits) depends typically on the radio resource blocks and MCS.} depends on its SNR, which might change drastically as users move around. {To illustrate this, we show in Fig.~\ref{fig:intro_02} experimental results obtained using the testbed platform described in Sec.~\ref{sec:applications}.}
Fig.~\ref{fig:intro_02} (left) presents the processing time of a TB (of certain length) at a CPU or GPU, under different SNRs\footnote{5G FEC is implemented via LDPC iterative algorithms, which require more computations for lower SNR, see \cite{Blankenship2021}.}. {Leveraging its high degree of data parallelization (intrinsic to its architecture) the GPU speeds up the TB processing, in contrast to the CPU where the processing is sequential (see \cite{falcao2010massively} for more details).}
This delay variance is critical, as the vBS workloads are subject to stringent processing deadlines (1-3 ms) which, if violated, lead to data loss and energy waste \cite{bib:foukas-sigcom21, bib:nuberu}. Third, the energy cost of these computations depends on the PU technology and the data stream characteristics. For example, the experiments in Fig.~\ref{fig:intro_02} (right) show that a GPU decodes a TB (with SNR 14 dB) approximately $5\times$ faster than a CPU, but consumes $2.5\times$ more energy; and demonstrate how the (unknown and varying) SNR affects this comparison. Motivated by these observations, this work introduces and evaluates an algorithmic toolbox for the design of \emph{data-adaptive RIC policies}, in different time-scales, towards taming the energy consumption of vRANs while accounting for fairness criteria w.r.t. performance (across users) and energy costs (across servers).

\subsection{Methods \& Contributions}

We focus on two key resource management problems that affect the vRANs' performance and energy, and design a computing control and a radio control policy to tackle them. The first problem studies the assignment of vBS workloads to O-Cloud processing units. The workloads differ in their volume and SNR, and similarly the PUs are heterogeneous in terms of technology (CPU or HA), capacity and energy consumption. We find experimentally that CPU-based PUs can process small and \emph{cleaner} (high SNR) workloads with less energy; while GPU-based PUs can be used for voluminous and/or low-SNR data to ensure their timely processing. Therefore, we argue that a RIC at non-Real-Time (non-RT) can devise an intelligent workload assignment policy which dictates how the vBSs can leverage the PUs' diversity to balance energy costs and performance (successfully processed loads). Such a policy needs to adapt to the time-varying properties of workloads and PUs; and allocate fairly in the long-run the O-Cloud capacity across the vBSs, and the energy costs across the PUs. This latter property increases reliability (via load balancing), and is key for multi-vendor O-Clouds where the PUs are owned by different business entities. 

The second problem concerns a new radio control policy, similar in flavor to those in \cite{vrain_conf, ayala2021bayesian, bib:ayala-edgebol}. The starting point here is our experiments (Sec. \ref{sec:applications}) showing that HAs consume almost the same energy per TB independently of its length\footnote{GPUs have Streaming Multiprocessors that parallelize workloads effectively, and LDCP codes are amenable to parallelization.}. Thus, if users transmit larger TBs, the vRAN will consume less energy per bit. This creates an opportunity for the RIC to introduce a highly dynamic (i.e., at near-RT scale) minimum TB size policy (\texttt{minTB}) for each user, preventing transmission of small TBs. Nevertheless, such a policy will inevitably deteriorate the latency for users, as they might need to refrain from transmitting despite having non-empty (MAC-layer) buffers. It is therefore imperative to strike a balance between the energy savings and transmission delays; and further, to disperse fairly these delays across the users so as to avoid excessive service deterioration for some of them. Deciding the \texttt{minTB} thresholds requires access to the user traffic and HA energy costs, which change with time and are typically unknown when such near-RT policies are devised. Further, it involves solving large-scale optimization problems in almost real-time scales.

We design a novel optimization toolbox for the above compute and radio control policies based on online learning, cf. \cite{bib:hazan-monograph, bib:shai-monograph}. Our approach relies on the celebrated Follow The Regularized Leader (FTRL) framework \cite{bib:ftrl-shalev}, that has been particularly successful in the design of data-adaptive and robust decision policies \cite{bib:mcmahan2017survey}. We extend FTRL here to account for the specifics of these problems, namely we equip it with a \emph{two-sided long-term} fairness metric, so as to support fairness w.r.t. cost savings across the servers and fairness w.r.t. performance gains across users over its entire operation; and we include predictions for the unknown (system and user) parameters. Achieving fairness in such dynamic decision models is technically challenging, and previous works are confined to per-slot fairness (which impacts efficiency), with only few exceptions, e.g., \cite{bib:tareq_fairness, bib:hori-fair-regr-liao, bib:hori-fair-regr-gupta}. We overcome this barrier through a saddle-point transformation where in the dual space we track the two fairness metrics. The predictions, on the other hand, bring in the optimistic learning aspect \cite{bib:mohri2016accelerating, bib:rakhlin-nips13}, and, if used judiciously, can expedite the learning rate when they are (relatively) accurate, without deteriorating it otherwise. The proposed algorithms offer optimality guarantees (i.e., \emph{regret}) w.r.t. ideal benchmark policies one could only devise using oracles, and which hold for a wide range of perturbation models including adversarial ones. Our contributions are thus summarized as follows:

$\bullet$ We present new experimental results for the computing delay and energy costs in O-RAN, which motivate the design of dynamic and adaptive compute and radio control policies. 

$\bullet$ We develop a lightweight learning framework for designing control policies which: \emph{(i)} assign fairly the vBSs' compute workloads to O-Cloud processing units; and \emph{(ii)} decide TB size thresholds, towards reducing the vRAN energy costs while ensuring fair performance across users. 

$\bullet$ We prove the policies have sublinear regret under adversarial scenarios and assess their implementation overheads and dependency on system parameters, user demands and predictions accuracy. Along the road we develop technical results that improve these learning techniques. 
	
$\bullet$ We evaluate the policies using trace-driven and synthetic simulations, and we implement them at an O-RAN-compliant testbed to measure the actual vBS performance and PUs energy costs.

\textbf{Notation}. $\|\cdot\|$, $\|\cdot\|_{\infty}$, and $\|\cdot\|_1$ denote the $\ell_2$ (Euclidean), $\ell_\infty$ and $\ell_1$ norms. Vector transpose and the Hadamard product are denoted $\top$ and $\circ$, respectively. We denote vectors with small bold typeface letters and use subscripts to index them. We use $\bm x_{1:t}$ for the sum of vectors $\sum_{i=1}^t \bm x_i$, and do so also for scalars. $\bm x_t$ time-indexes vector $\bm x$, and  $\{\bm x_t\}_{t=1}^T$ denotes the sequence $\bm x_1, \bm x_2, \ldots, \bm x_T$. When the horizon is not relevant, we write $\{\bm x_t\}_t$. The $\ell_2$ diameter of a set $\mathcal X$ is denoted by $D_{\mathcal X}$.

\textbf{Outline of Paper}. Sec. \ref{sec:related-work} reviews the related work about O-RANs and online learning. Sec. \ref{sec:fairness_du_cu} presents the system model, the load assignment learning problem and its saddle-point reformulation. Sec. \ref{sec:fairness_du_cu_algorithms} provides a brief background on FTRL and introduces the optimistic FTRL algorithm for the load assignment problem; while Sec. \ref{sec:fairness-cost} presents the model and algorithm for the TB threshold policy. We provide motivating experiments, and extensive simulation and experimental evaluation of the two algorithms in Sec. \ref{sec:applications}. The remaining proofs and additional results can be found in the Appendix.

\section{Literature Review}\label{sec:related-work}

\subsection{Resource Management in vRANs}
   
Resource management solutions for mobile networks can be broadly classified into those using analytical functions that map control actions to performance metrics, e.g., \cite{rost-globecom15, bega2018cares, wang-compaware-WirNet19, halabian2019distributed}; solutions that employ offline-trained ML models, e.g., \cite{bega-NNSlicing-JSAC2020, cormac-NN-ComMag20}; and techniques that adapt to network conditions and user demands \cite{galanopoulos2020bayesian, ayala2019online, niyato-RL-TWC2019, alqerm2018sophisticated, xu2017online}. Unfortunately, function-based models rely on parameters that are most often unknown in vRANs; while the efficacy of ML models depends on the availability of \emph{representative} training data \cite{patras-DNN-Tutorial2019}. Examples of more adaptive solutions include Bayesian learning for optimizing video analytics \cite{galanopoulos2020bayesian} and BS energy costs \cite{vrain_journal, ayala2021bayesian, bib:ayala-edgebol}; and Reinforcement Learning for spectrum management and wireless scheduling \cite{niyato-RL-TWC2019, alcaraz2020online}, among many others. These approaches have high overhead, e.g., require expensive matrix inversions, and provide optimality guarantees only under stationary conditions; {therefore, they are typically employed for longer-term static resource control policies.} {Other practical solutions tailored to vRAN resource management include hybrid offline-trained and online-adapted vBS workload predictions \cite{bib:foukas-sigcom21}, or regression models~\cite{bib:nuberu}, so as to increase the utilization of the employed CPUs. These works do not provide optimality or fairness guarantees, and operate in real-time as opposed to the non-RT scale of our assignment policy; thus can be used concurrently.}

Here, we rely instead on the theory of online convex optimization (OCO) which: \emph{(i)} does not require access to performance/cost functions or system/user-related parameters; and \emph{(ii)} offers guarantees under a wide range of scenarios, including adversarial ones \cite{bib:zinkevich}.
The robustness of OCO is particularly useful for vRANs which, due to the virtualized vBS functions, exhibit volatile performance and, importantly, high energy costs \cite{bib:vbs-experiments, bib:nuberu,  bib:foukas-sigcom21, Blankenship2021}. These experimental findings motivated the design of policies about the vBSs' transmission power, modulation and coding schemes, spectrum usage, and others, which aim to curb the vRANs energy costs \cite{vrain_journal, ayala2021bayesian, bib:ayala-edgebol}. Such policies can be devised centrally by the RIC and applied to different vBSs concurrently \cite{bib:melodia-commag-dApps, bib:andres-magazine}. These solutions become more interesting due to HAs that are increasingly common in industry-grade vBS-hosting platforms \cite{bib:intel-2021, bib:dell-2022}. HAs are already used in cloud computing where, due to their high cost, it is imperative to utilize them effectively \cite{bib:gpus-darabi-sigm22, bib:gpu-tiwari2022}. When it comes to vBS functions, the performance and energy costs of HAs are substantially different from CPUs as our experiments find, hence paving the road for new energy-saving policies. Our proposal includes a new near-RT \emph{radio control} policy that decides the minimum TB size each user can employ; and a non-RT \emph{compute-control} policy which assigns vBSs' workloads to CPU and HAs. Both policies optimize performance and energy consumption, while being fair in terms of the service offered to users {and} the energy costs dispersed across the PUs.

\subsection{Fairness \& Online Learning}

Fairness is a key metric in resource management and has been extensively applied in cloud computing \cite{bib:wang-dominant-fairness, bib:bonald-sigmetrics15} and communication systems \cite{bib:altman2008, bib:kelly98}, among many others \cite{bib:fairness-nace, bib:fairness-leboudec, nguyen2019market}. More recently, \cite{bib:kelleler, bib:laporta} focused on max-min throughput fairness in RANs, e.g., via spectrum management; \cite{bib:fair-vran-andres} studied the fair allocation of computing capacity to vRAN functions and edge services; \cite{bib:fair-vran-ruffini} considered cost-fairness in multi-tenant O-RANs where operators lease computing for their vBS functions; {while \cite{halabian2019distributed, modina2022multi, fossati2020multi} focus on virtualization and slicing.} These interesting works, however, do not consider the inherent system and user dynamics in vRANs and/or do not provide fairness guarantees. Achieving fairness in such dynamic problems is indeed challenging, even from a theoretical point of view. Previous works have studied \emph{slot fairness} criteria \cite{bib:slot-fairness-sinclair, bib:slot-fairness-sadegh, bib:slot-fairness-jalota} where in each decision round the problem of fairness is tackled independently of the past or future decisions. By definition, the scope of these fairness metrics is limited {and it results in higher price of fairness~\cite{bib:tareq_fairness, bib:bertsimas2011price}}. More ambitious approaches attempt to achieve \emph{horizon fairness}, where the fairness metric is enforced across the entire system operation, not instantly. {We also refer the reader to the interesting work~\cite{bib:altman2012multiscale} that discusses fairness in wireless networks over multiple time scales.}

Recent studies on horizon fairness assume the utility functions to be either known or non-adversarial, e.g., study the stochastic version, \cite{bib:hori-fair-regr-baek, bib:hori-fair-regr-benade, bib:hori-fair-regr-cayci, bib:hori-fair-regr-gupta, bib:hori-fair-regr-liao, bib:tareq_fairness, bib:altman2012multiscale}. Instead, we target a framework that drops these assumptions. The closest to our work is \cite{bib:tareq_fairness}, which we extend here in many ways. First, we use a novel \emph{optimistic} learning fairness algorithm that leverages predictions for the performance and costs. Secondly, we consider a two-sided alpha-fairness criterion, i.e., w.r.t. cost savings across the servers and w.r.t. performance across users, where the two fairness parameters can be even different. And, finally, we employ a tailored learning algorithm with minimal computation and memory requirements. Namely, we rely on the FTRL framework \cite{bib:mcmahan2017survey} and draw ideas from optimistic learning \cite{bib:rakhlin-nips13, bib:mohri2016accelerating} that has been recently used, e.g., for caching \cite{bib:naram-caching} and network control \cite{bib:llp}. Here, we extend the optimistic learning algorithms to our problems, aiming for low computation and memory requirements and constant (not only sublinear) regret for perfect predictions. 

Finally, it is worth noting the connection of fairness with load-balancing techniques. The majority of studies in this latter area focus on the asymptotic regime and consider stochastic loads or servers with fixed capacities; see discussion of literature in \cite{bib:loadbalanc-srikant21}. Works that do drop these assumptions include \cite{bib:loadbalanc-shakkottai15} which assigns equal-length jobs with known deadlines; and \cite{bib:fair-online-naor13} which considers max-min fairness from the servers' (minimize maximum load) or jobs' (maximize minimum service) perspective. Our model and motivation is different. We aim at fairness (i.e., balance) w.r.t. time-varying and unknown job-utility and server-cost functions, we make no assumptions about the load arrivals, and our policies operate at time-scales where queue stability is of no concern. 

\section{System Model and Fairness Regret}\label{sec:fairness_du_cu}
We start with the system model for the assignment problem; define the learning problem and regret metrics; and propose a saddle-point reformulation which is used in the algorithm design in Sec. \ref{sec:fairness_du_cu_algorithms}.

\subsection{Model \& Problem Statement}\label{sec:fairness_du_cu_model}

We consider a vRAN with a set $\mathcal I$ of $I=|\mathcal I|$ vBSs, and a set $\mathcal{J}$ of $J=|\mathcal{J}|$ of PUs (or, servers) that comprise the O-Cloud. The operation of the system is time-slotted and the slot duration is considered $\sim(1-10)$ seconds, since this is a non-RT policy implemented by a RIC at the SMO, Fig. \ref{fig:model}(a). We study the system for a set $\c T$ of $T=|\c T|$ slots, and focus on the more computation-demanding uplink \cite{bib:ayala-edgebol, bib:vbs-experiments}. During every slot $t$, each vBS $i\in \mathcal I$ injects into the O-Cloud an amount of $\lam_{it}\!\geq\! 0$ data (bytes), stemming from its users, and we define $\bm \lam_t\!=\!(\lam_{it}, i\in \mathcal I)$. The required computations for these data, e.g., for FFT or FEC decoding, depend on their volume and wireless conditions that affect their SNR, see Fig. \ref{fig:intro_01}-\ref{fig:intro_02} and~\cite{bib:nuberu, bib:vbs-experiments}. Hence, in practice the value of $\bm \lambda_{t}$ and their computations are revealed at the end of each slot $t$. On the other hand, each server $j\! \in\! \mathcal J$ has computing capacity of $C_{jt}$ cycles during each slot $t$, and we define $\bm C_t\!=\!(C_{jt}, j\in \c J)$. We study the general case where the capacities might change over time. Similar to the loads, we assume $\bm C_t$ becomes known at the end of each slot.

A non-RT controller decides the O-Cloud \emph{assignment policy}, i.e., how much data (or, load) from each vBS will be routed to each server. We denote with $x_{ijt}\in [0,1]$ the load portion of vBS $i$ that is sent to server $j$ during slot $t$, and hence $x_{ijt}\lam_{it}$ is the assigned data from that vBS. We also define the assignment vector $\bm x_{it}\!=\!(x_{ijt}, \forall j\in \mathcal J)$ for each vBS $i\in \mathcal I$; the vector $\bm x_{jt}\!=\!(x_{ijt}, \forall i\in \mathcal I)$ for each server $j\in \mathcal J$; and the total assignment $\bm x_t\!=\!(x_{ijt}, \forall i \in \c I, j\in \c J)$. These decisions are subject to a simplex constraint for each vBS, thus each $\bm x_t, t\!\in \c T$, belongs to set:
\begin{align}
\c X=\left\{\bm{x}\in [0,1]^{I\cdot J}\ : \ \sum_{j\in \c J}x_{ij}=1, \forall i\in \c I \right\}. \label{eq:X-simplices}
\end{align}

\begin{figure*}[t!]
	\centering
	\begin{subfigure}[t]{0.48\textwidth}
		\centering
		\includegraphics[width=0.96\textwidth]{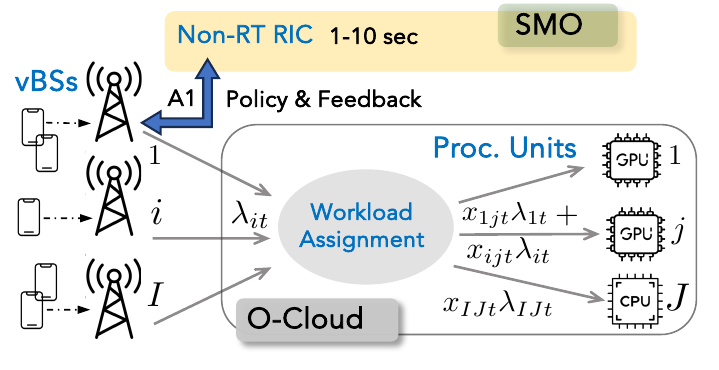}
        \caption{}
	\end{subfigure}%
	\hfill 
	\begin{subfigure}[t]{0.48\textwidth}
		\centering
		\includegraphics[width=0.98\textwidth, page=1]{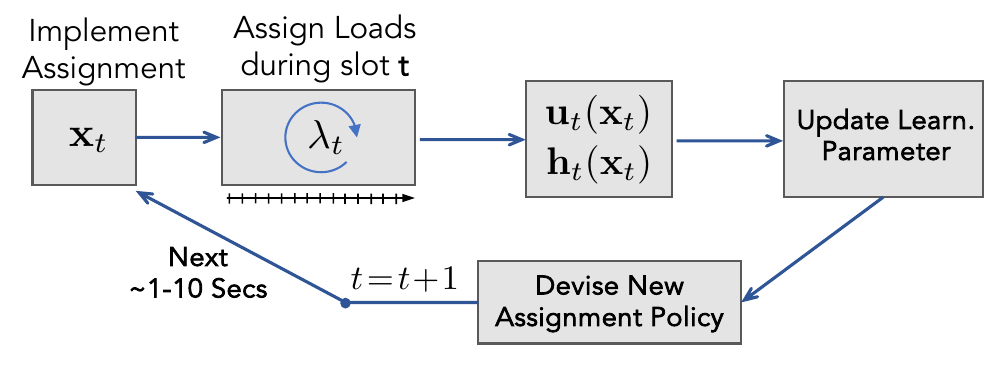}
        \caption{}
	\end{subfigure}
	\vspace{-1mm}
	\caption{\textbf{(a):} A non-RT controller at the Service \& Management Orchestration (SMO) framework devises the load assignment policy every $\sim \!1\!-\!10$ seconds and sends it to the vBSs via the A1 interface. \textbf{(b)}: Timing diagram of assignment implementation and learning policy.} 
  \label{fig:model}	
\end{figure*}

The assignment policy is updated at the beginning of each slot $t$ and shapes the system performance during that slot. If the controller assigns more load to a server than its capacity, then (part of) this data will not be processed before its deadline \cite{Blankenship2021}. This means that the associated vBSs will suffer reduced throughput \cite{bib:fluidran, bib:nuberu}. Thus, the benefit for a vBS when using a server decreases sharply when the total load approaches the server's capacity. We model this effect through a (possibly) time-varying utility vector function $\bm u_t(\bm x)\!=\!\big(u_{it}(\bm x), i\in \mathcal I\big)$, where $\bm u_t:\!\mathbb R^{I\times J} \mapsto \mathbb R^{I}_{+}$ is assumed non-negative and concave. Each element $u_{it}(\bm x)\!\in\![u_{min}, u_{max}]$ denotes the performance for vBS $i\in \c I$ under assignment $\bm x$, and captures the server heterogeneity, e.g., through $\bm C_t$.

Accordingly, we use functions $\bm h_t:\mathbb R^{I\times J} \mapsto \mathbb R_+^J$, to model the \emph{energy cost savings} for the servers at each slot $t\!\in\!\c T$, where $h_{jt}(\bm x)\in [h_{min}, h_{max}]$ is the cost reduction\footnote{Parameters $h_{min}$ and $h_{max}$, as well as $u_{min}$ and $u_{max}$, can be determined based on the vBS operation envelope.} of server $j\!\in \c J$ under assignment $\bm x$. This reduction is calculated with reference to the (unknown) energy cost the server would have paid, had it served the entire load in the network. Put it differently, these functions model the benefits from dispersing the load across multiple servers instead of using only one. In line with prior works, e.g., \cite{bib:fair-vran-andres}, and based on our measurements (Sec. \ref{sec:applications}) we consider these functions to be non-negative and concave on $\bm x$. Our analysis can be applied to any type of utility and cost functions satisfying these minimal requirements, and we study a specific example in Sec. \ref{sec:applications}. 

The goal of the controller is to devise a sequence of assignment policies $\{\bm x\}_{t=1}^T$ so as to achieve a two-sided fairness criterion: \emph{(i)} fairness w.r.t. the average utility perceived by the vBSs over the horizon $\c T$, i.e., w.r.t. $(1/T)\sum_{t\in\c T} \bm{u}_t(\bm{x}_t)$; and \emph{(ii)} fairness w.r.t. to the average energy cost savings of servers, i.e., $(1/T)\sum_{t\in\c T} \bm{h}_t(\bm{x}_t)$. To do so, the controller needs to overcome two challenges. First, to decide the per-slot assignment $\bm x_t$ in a way that optimizes the immediate performance and costs while tracking these two long-term (\emph{horizon}) fairness criteria. Secondly, it needs to achieve this balance without information about the system parameters $\{\bm \lam_t, \bm C_t\}_t$, and functions $\{\bm u_t, \bm h_t\}_t$, which are time-varying, unknown, and revealed after each $\bm x_t$ is decided, see Fig. \ref{fig:model}(b). In fact, we adopt the most general perturbation model where these parameters are assumed to be decided dynamically by an \emph{adversary} aiming to deteriorate the system operation \cite{bib:zinkevich}. Clearly, a policy that performs well under these conditions, can also perform under more benign static or stationary scenarios\footnote{A policy designed for adversarial environments might not be, in general, ideal (i.e., best-performing) for static environments and can be outperformed by algorithms tailored for such specific scenarios, when one has guarantees for their existence.}. For the fairness criteria, we employ the \emph{generalized} $\alpha$-fairness function \cite{bib:walrand-fairness, bib:altman2008}:
\begin{align}\label{eq:basic-fairness}
    F_{\alpha}(\bm{u}) \doteq \sum\limits_{i\in\mathcal{I}}f_\alpha(u_i) \quad \textrm{where} \quad f_\alpha(u_i) \doteq 
    \begin{cases} 
      \frac{u_i^{1-\alpha} - 1}{1-\alpha}, &\textrm{for } \alpha \in \mathbb{R}_{\geq0} \backslash \{1\}, \\
      \log(u_i), &\textrm{for } \alpha = 1.
   \end{cases}
\end{align}
Parameter $\alpha$ determines the type of fairness we wish to enforce; e.g., $\alpha\!=\!1$ yields the proportional fairness metric, while $\alpha\rightarrow \infty$ leads to max-min fairness. We define a similar fairness function for the cost savings, $F_{\beta}(\bm h)=\sum_{j\in\mathcal J} f_{\beta}(h_j)$, where in general it can be $\alpha\neq \beta$.

We evaluate the efficacy of the assignment policies using the metric of \emph{static regret} which is extended here to capture the two-sided horizon fairness as follows:
\begin{equation}
\begin{split}
	\bm{\mathcal{R}}_T(F_{\alpha}, F_\beta) \doteq \sup\limits_{\{\bm{u}_t, \bm{h}_t\}_{t=1}^T}\Bigg\{ F_\alpha&\left(\frac{1}{T}\sum\limits_{t\in\mathcal{T}}\bm{u}_t(\bm{x^\star})\right) + F_\beta\left(\frac{1}{T}\sum\limits_{t\in\mathcal{T}}\bm{h}_t(\bm{x^\star})\right) \\ 
	-&F_\alpha\left(\frac{1}{T}\sum\limits_{t\in\mathcal{T}}\bm{u}_t(\bm{x}_t)\right)-F_\beta\left(\frac{1}{T}\sum\limits_{t\in\mathcal{T}}\bm{h}_t(\bm{x}_t)\right)\Bigg\}. \label{eq:fairness_regret}
\end{split} 
\end{equation}
This metric evaluates the policy that decides $\{\bm x_t\}$ dynamically over $\c T$, by using a hypothetical benchmark $\bm {x}^\star$ that could be only devised with access at $t\!=\!0$ to all loads, capacities, and functions: 
\begin{align}
	\bm{x}^\star = \arg\max_{x\in \c X}\left\{F_\alpha\left(\frac{1}{T}\sum\limits_{t\in\mathcal{T}}\bm{u}_t(\bm{x})\right)+F_{\beta}\left(\frac{1}{T}\sum\limits_{t\in\mathcal{T}}\bm{h}_t(\bm{x})\right)\right\}. \notag
\end{align}
We aim to find $\{\bm{x}_t\}_{t=1}^T$ that ensures the loss compared to $\bm x^\star$ will diminish to zero, for \emph{any realization} of the unknown parameters, as is evidenced from the regret definition in \eqref{eq:fairness_regret}. 
\subsection{Reformulation \& Solution Approach }\label{sec:fairness_du_cu_theory}

Unfortunately, off-the-shelf (online) convex optimization algorithms cannot be applied directly on this problem, due to the time-averaging in the argument of functions $F_\alpha(\cdot)$ and $F_\beta(\cdot)$, which does not allow the necessary (for these techniques) decomposition over time; see also \cite{bib:tareq_fairness, bib:agrawal2014EC}. To tackle this issue, we introduce a \emph{proxy function} $\Psi_t:\Theta  \times \Phi  \times \c X \mapsto \mathbb R$  with two types of dual variables, $\bm \theta \in \Theta$ and $\bm \phi \in \Phi$, as follows:
\begin{align}
	\Psi_t(\bm \theta, \bm \phi, \bm x)=\Psi_t^\alpha(\bm \theta, \bm x)+&\Psi_t^\beta(\bm \phi, \bm x),  \label{eq:psi_t}
\end{align}
where functions $\Psi_t^\alpha:\Theta \times \c X \mapsto \mathbb R$ and $\Psi_t^\beta:\Phi  \times \c X \mapsto \mathbb R$, are defined as:
\begin{align}
	&\Psi_t^\alpha(\bm \theta, \bm x)=(-F_{\alpha})^\star(\bm \theta) - \bm \theta^\top \bm{u}_t(\bm x), \quad \Psi_t^\beta(\bm \phi, \bm x)=(-F_{\beta})^\star(\bm \phi) - \bm \phi^\top \bm{h}_t(\bm x),	
\end{align}
and the dual variables are bounded in  $\Theta\doteq[-1/u_{min}^{\alpha}, -1/u_{max}^{\alpha}]^I$ and $ \Phi\doteq[-1/h_{min}^{\beta}, -1/h_{max}^{\beta}]^J$. Function $(-F_{\alpha})^\star(\cdot)$ is the Fenchel convex conjugate of  $-F_{\alpha}\big(\bm{u}_t(\bm x)\big)$ \cite[Ch. 4]{bib:beck-book}, i.e., 
\begin{align}\label{eq:conjugate}
	(-F_{\alpha})^\star(\bm \theta)=\max_{\bm x \in \c X} \left\{\bm \theta^\top \bm{u}_t(\bm{x}) - \big(-F_{\alpha}(\bm{u}_t(\bm{x}))\big)	\right\}=\max_{\bm x \in \c X} \left\{\bm \theta^\top \bm{u}_t(\bm x) + F_\alpha\big(\bm u(\bm x)\big)\right\},
\end{align}
and similarly we define $(-F_{\beta})^\star(\bm \phi)$  for function $-F_\beta\big(\bm{u}_t(\bm x)\big)$. Interestingly, given the function structure in \eqref{eq:basic-fairness}, these proxy functions can be expressed analytically as
\begin{align}\label{eq:psi-analytical}
\Psi_t^\alpha(\bm \theta,\bm x)=\sum_{i=1}^I \frac{\alpha(-\theta_i)^{1-1/\alpha}-1}{1-\alpha}-\bm \theta^\top \bm u_t(\bm x), \quad 	\Psi_t^\beta(\bm \phi,\bm x)=\sum_{j=1}^J \frac{\beta(-\phi_j)^{1-1/\beta}-1}{1-\beta}-\bm \phi^\top \bm h_t(\bm x),
\end{align}
and when $\alpha=1$ we get $\Psi_t^\alpha(\bm \theta,\bm x)=-1-\log(-\bm \theta) - \bm \theta^\top \bm u_t(\bm x)$, and similarly for $\beta=1$ and $\Psi_t^\beta$.

These functions are suitable for our problem as we can recover the fairness objective with a minimization operation \cite[Th. 4.8]{bib:beck-book}. That is, leveraging their biconjugate equivalence we can write:
\begin{align}
F_{\alpha}\big(\bm u_t(\bm x)
\big)=\min_{\bm \theta \in  \Theta} \ \Psi_t^a(\bm \theta,\bm x) \quad \text{and} \quad 	F_{\beta}\big(\bm h_t(\bm x)\big)=\min_{\bm \phi \in  \Phi} \ \Psi_t^\beta(\bm \phi,\bm x). \label{conj1}
\end{align}
At the same time, $\Psi_t(\bm \theta, \bm \phi, \bm x)$ is linear on the utility and cost values, hence with this transformation we can maximize a (separable) sum of functions instead of a (non-separable) concave function of them. Putting these together, the problem we aim to solve has at its core the (per-slot) program:
\begin{align}\label{saddle-point}
	\max_{x\in \c X}\Big\{ F_\alpha\big( \bm{u}_t(\bm{x}_t) \big) + F_\beta\big( \bm{h}_t(\bm{x}_t) \big) \Big\}=\max_{x\in\c X} \Big\{\min_{\bm \theta \in \Theta} \Psi_t^\alpha(\bm \theta, \bm x) + \min_{\bm \phi \in \Phi} \Psi_t^\beta (\bm \phi, \bm x)	  \Big\},
\end{align}
which we will tackle with a saddle-point algorithm that updates the primal and dual variables successively, performing independent (but coordinated) learning in the primal and dual space. In particular, we will be running an OCO algorithm on $\bm x$ to bound the primal-space regret:
\begin{align} \label{primal-space-regret-xx}
\bm{\mathcal{R}}_{T}^x \doteq \sumT \Big( \Psi_{t}(\bm \theta_t, \bm \phi_t, \bm x) - \Psi_{t}(\bm \theta_t, \bm \phi_t, \bm x_t)\Big),\ \ \forall \bm x\in\c X,
\end{align}
and similarly, we will learn using the proxy function in the dual spaces, to bound:
\begin{align}\label{dual-space-regret-xx}
 \bm{\mathcal{R}}_{T}^\theta \doteq \sumT\Big( \Psi_{t}^{\alpha}(\bm \theta_t, \bm x_t) - \Psi_{t}^{\alpha}(\bm \theta, \bm x_t)\Big),\forall \bm \theta \in \Theta, \quad \bm{\mathcal{R}}_{T}^\phi \doteq \sumT\Big( \Psi_{t}^{\beta}(\bm \phi_t, \bm x_t) -  \Psi_{t}^{\beta}(\bm \phi, \bm x_t)\Big), \  \forall \bm \phi \in \Phi.
\end{align}
We will show in the next section that we can use these regret bounds to upper-bound the horizon-fair regret $\bm{\mathcal{R}}_T(F_\alpha, F_\beta)$, which is the goal of the RIC here.

\section{Learning Algorithms and Assignment Policy}\label{sec:fairness_du_cu_algorithms}
We now present the learning algorithms in the primal and dual space and characterize their regret bounds, which we then combine to build the controller's assignment policy and assess its regret. 

\subsection{OFTRL Algorithms}

We will perform the RIC policy learning using \emph{optimistic FTRL} algorithms \cite{bib:rakhlin-nips13, bib:mohri2016accelerating}. In the OFTRL template, the variables are updated at the beginning of each new slot $t+1$ using a time-varying regularizer $r_{1:t}(\bm x)$, all  past function gradients and a prediction for the function gradient at slot $t+1$ (the optimistic element). Such predictions, if incorporated carefully, can improve the learning rate when accurate, without sacrificing performance when they are inaccurate. For example, one can use past values of vBS loads as a prediction for the next loads without risking no-learning conditions if there is a distribution shift. Optimistic learning has been recently used, e.g., for content caching \cite{bib:naram-caching} or routing \cite{bib:llp}, but not for vRANs and not in conjunction with allocative fairness. 

The performance of such FTRL algorithms is shaped by the regularizers. A typical choice is the quadratic regularizer $r_{1:t}(\bm x)\!=\!\frac{\sigma_{1:t}}{2}\|\bm x\|^2$ and its proximal variant which uses instead $\|\bm x\!-\bm x_t\|^2$. Parameters $\{\sigma_t\}$ encode information about the system properties and the predictions' accuracy. When the constraint is a simplex, the entropic regularizer $r_{1:t}(\bm x)\!=\sigma_t\sum_{i\in\c I}x_i\log(x_i)$ allows a closed-form derivation of $\bm x_t$ and achieves lower dependency on the decision space diameter. Proximal regularizers require more memory and computations to optimize the variables, but achieve $\c O(1)$ regret when all predictions are accurate. On the other hand, non-proximal regularizers (as the entropic) are computationally-efficient but yield sublinear (not constant) regret $\c O(\sqrt T)$ even with perfect predictions, see \cite{bib:mohri2016accelerating}. Here, we use a quadratic regularizer in the dual space and an entropic one for the primal space. What is more, we tune these algorithms to achieve $\c O(1)$ regret for perfect predictions (despite being non-proximal), and provide closed-form derivations in both cases.

\subsubsection{OFTRL with Quadratic Regularizer}

We start with the analysis of the learning in the dual spaces $\Theta$ and $\Phi$. The proposed OFTRL dual updates for these minimization problems are:
\begin{align}
&\bm \theta_{t+1}=\arg\min_{\bm \theta \in \Theta}\left\{q_{1:t}(\bm \theta)+\bm \theta^\top(\bm \kappa_{1:t}+\bm{\p \kappa}_{t+1}) \right\},\label{dual-update-2a}\\
&\bm \phi_{t+1}=\arg\min_{\bm \phi \in \Phi}\left\{p_{1:t}(\bm \phi)+\bm \phi^\top(\bm \mu_{1:t}+\bm{\p \mu}_{t+1}) \right\},  \label{dual-update-2b}
\end{align}
where $q_{1:t}(\bm \theta)=\sum_{\tau=0}^t q_{\tau}(\bm \theta)$ and $p_{1:t}(\bm \phi)=\sum_{\tau=0}^t p_{\tau}(\bm \phi)$ are the aggregate dual regularizing functions imposed at slot $t$; vectors $\bm \kappa_{1:t}=\sum_{\tau=1}^t \nabla_{\bm \theta} \Psi_{\tau}^a(\bm \theta_{\tau}, \bm x_{\tau})$, $\bm \mu_{1:t}=\sum_{\tau=1}^t \nabla_{\bm \phi} \Psi_{\tau}^\beta(\bm \phi_{\tau}, \bm x_{\tau})$ are the aggregate dual gradients; and $\bm{\p \kappa}_{t+1}$, $\bm{\p \mu}_{t+1}$ denote the respective gradient predictions for $t+1$. Following the rationale in \cite{bib:mohri2016accelerating, bib:naram_sigm} and based on the  geometry of these spaces, we propose the regularizers: 
\begin{align}\label{eq:dual-regularizer}
&q_{1:t}(\bm \theta)=\frac{\sigma_{1:t}}{2}\|\bm \theta\|_{2}^2 \quad \text{where} \quad \sigma_{1:t}=\sigma \sqrt{\sum_{\tau=1}^{t}\|\bm{\kappa}_\tau - \Tilde{\bm{\kappa}}_\tau\|_{2}^2}, \quad \sigma=2\sqrt{2}/D_\Theta, \\
&p_{1:t}(\bm \phi)=\frac{\xi_{1:t}}{2}\|\bm \phi\|_{2}^2 \quad \text{where} \quad \xi_{1:t}=\xi \sqrt{\sum_{\tau=1}^{t}\|\bm{\mu}_\tau - \Tilde{\bm{\mu}}_\tau\|_{2}^2}, \quad \xi=2\sqrt{2}/D_{\Phi}, \label{eq:dual-regularizer-b} 	
\end{align}
which impose regularization commensurate to the prediction errors up to each slot $t\in\mathcal T$. It follows that $q_{1:t}(\bm \theta)$ is 1-strongly-convex w.r.t. the norm $\|\bm \theta \|_{(t)}=\sqrt{\sigma_{1:t}}\|\bm \theta\|$, which has dual norm $\|\bm \theta\|_{(t), \star}=\|\bm \theta\|/\sqrt{\sigma_{1:t}}$, and similarly for $p_{1:t}(\bm \phi)$, see \cite{bib:mcmahan2017survey}.

If we apply the OFTRL updates \eqref{dual-update-2a}-\eqref{dual-update-2b} with regularizers \eqref{eq:dual-regularizer}-\eqref{eq:dual-regularizer-b}, we can upper-bound the regret in the dual spaces as the next result states, which holds as is for $\Phi$ as well. 
\begin{lemma}\label{the:quadratic-regret}
For a compact convex set $\Theta$, update \eqref{dual-update-2a} with regularizer \eqref{eq:dual-regularizer} yields regret:
\begin{align}
	\bm{\mathcal{R}}_T^{\theta} \leq 4\sqrt{2}\textit{D}_{\Theta} \sqrt{\sum_{t=1}^{T}\|\bm{\kappa}_t - \Tilde{\bm{\kappa}}_t\|_{2}^2}.
\end{align}
\end{lemma}
This result improves the optimistic regret bound of quadratic regularizers by enabling constant regret $\c O(1)$, as opposed to sublinear (not constant) regret \cite{bib:mohri2016accelerating}, in the case of perfect predictions.

\subsubsection{OFTRL Algorithm with Entropic Regularizer}

For the primal update we employ entropic regularization due to the multi-simplex structure of $\c X$. The update for this problem is\footnote{Note that, as the primal-space problem is a minimization one the aggregate gradient here has a minus sign.}:
\begin{align}
\bm x_{t+1}=\arg\min_{\bm x\in \c X}\left\{  r_{1:t}(\bm x) - x^\top \big( \bm g_{1:t}+\bm{w}_{1:t}+ \bm{\p g}_{t+1}+\bm{\p w}_{t+1}\big)	 \right\} \label{primal-update-2}
\end{align}
where $r_{1:t}(\bm x)$ is the aggregate regularization at $t$; vectors $\bm g_{1:t}=\sum_{\tau=1}^t \nabla_{\bm x} \Psi_{\tau}^a(\bm \theta_{\tau}, \bm x_{\tau})$ and $\bm w_{1:t}=\sum_{\tau=1}^t \nabla_{\bm x} \Psi_{\tau}^\beta(\bm \phi_{\tau}, \bm x_{\tau})$ are the aggregate primal-space gradients; and $\bm{\p g}_{t+1}$ and $\bm{\p w}_{t+1}$ the gradient predictions for $t+1$. The proposed entropic regularizer for this multi-simplex constraint is:
\begin{align}\label{eq:primal-regularizer}
	r_{1:t}(\bm x)=\frac{\eta_{1:t}}{2}\left( I\log J + \sum_{i\in \c I}\sum_{j\in \c J} x_{ij}\log x_{ij}\right), \quad \text{where} \quad \eta_{1:t}=\eta\sqrt{\sum_{\tau=1}^{t}\|\bm g_\tau+\bm w_\tau - \bm{\p g}_\tau- \bm{\p w}_\tau\|_{\infty}^2 }.
\end{align}
Each $r_{1:t}(\bm x)$ is now 1-strongly-convex w.r.t. norm $\|\bm x\|_{(t)}=\|\bm x\|_1\big(\eta_{1:t}/I\big)^{1/2}$ (see Lemma \ref{lem:reg-conv-entropic} in Appendix). This update yields regret in the primal space that is upper bounded by the next Lemma. 
\begin{lemma}\label{the:entropic-regret}
For the convex set $\mathcal{X}$ defined in \eqref{eq:X-simplices}, the update \eqref{primal-update-2} with regularizer \eqref{eq:primal-regularizer} ensures:
\begin{align}
\bm{\mathcal{R}}_T^x \leq \left(\frac{\sqrt{2}I}{\eta}+\frac{\eta I\log J}{2} \right) \sqrt{\sum_{t=1}^{T}\|\bm{g}_t+\bm{w}_t-\Tilde{\bm{g}}_t -\Tilde{\bm{w}}_t\|_{\infty}^2}\quad \text{where} \quad \eta=\min\left\{ \frac{1}{2}, \sqrt{\frac{2\sqrt 2}{\log J}} \right\}.
\end{align}
\end{lemma}
Similarly to $\bm{\mathcal{R}}_T^\theta$ and $\bm{\mathcal{R}}_T^\phi$, this result ensure regret $\bm{\mathcal{R}}_T^x=\c O(1)$ when the predictions are perfect, while we still get $\bm{\mathcal{R}}_T^x=\c O\big(\sqrt T\big)$ even when the predictions are maximally inaccurate. 

\subsubsection{Implementation}
Given the tight deadlines of the vBS functions and the scale of vRANs, it is imperative the algorithm to be lightweight. To that end, we solve analytically its core optimization steps. Applying first-order optimality conditions on \eqref{eq:psi-analytical}, we can express the partial gradients as:
\begin{align}
	&\bm g_t = \Big(-\sum_{i=1}^I \theta_{it}\frac{\vartheta u_{it}(\bm x_t)}{\vartheta x_{ijt}},\ i\in\c I, j\in\c J\Big), &&\bm w_t = \Big(-\sum_{j=1}^J \phi_{jt}\frac{\vartheta h_{ht}(\bm x_t)}{\vartheta x_{ijt}},\ i\in\c I, j\in\c J\Big), \qquad \text{and} 	 \label{grad-xz-1}\\
	&\bm \kappa_t = \Big( -(-\theta_{it})^{-\frac{1}{\alpha}}- u_{it}(\bm x_t),\ i\in \c I\Big), 
	&&\bm \mu_t = \Big( -(-\phi_{jt})^{-\frac{1}{\beta}}	 - h_{jt}(\bm x_t),\ j\in \c J\Big).
	\label{grad-xz-2}
\end{align}
Furthermore, both the dual and primal variable updates can be performed with closed-form expressions leveraging the following formulas.
\begin{prop}\label{prop:closed-form1}
The closed-form solution to $\bm{\theta}_{t+1}$ in iteration \eqref{dual-update-2a} is given by
\begin{align}\label{eq:closed-form-quad}
		&\theta_{i,t+1} = \min\left\{\max\left\{-\frac{\kappa_{i,1:t} + \p \kappa_{i,t+1}}{\frac{2\sqrt{2}}{D_\Theta} \sqrt{\sum_{\tau=1}^{t}\|\bm \kappa_\tau - \bm{\p \kappa_\tau}\|^2}}, \frac{-1}{u_{min}^{\alpha}}\right\}, \frac{-1}{u_{max}^{\alpha}}\right\}, \forall i \in \c I.
	\end{align}
\end{prop}
A similar expression can be derived for variables $\{\phi_{jt}\}$, while for the primal update we can use:
\begin{prop}\label{prop:closed-form2}
The closed-form solution to $\bm{x}_t$ in iteration \eqref{primal-update-2} is given by:
\begin{align}\label{eq:closed-form}
x_{ij,t+1} = \frac{\exp\parentheses{ 2\omega_{ijt}/\eta_{1:t}    }}{\sum_{j\in\mathcal{J}}\exp\parentheses{2\omega_{ijt}/\eta_{1:t}}}, \ \ \ \forall i\in \c I, j\in \c J,
\end{align}
where $\omega_{ijt}\doteq g_{ij,1:t}+w_{ij,1:t}+\p g_{ij,t+1} + \p w_{ij,t+1}, \forall i,j,t$, and $\eta_{1:t}$ is given by \eqref{eq:primal-regularizer}.
\end{prop}

Such closed-form expressions are commonly used for entropic regularizers over \emph{one} simplex \cite{bib:shai-monograph}, and we extend this idea for the \emph{multi-simplex} set $\c X$. This allows to run the primal updates with $\mathcal O(1)$ memory since we maintain only the aggregate gradients, and with $\mathcal O(1)$ computation time.

\subsection{Horizon-Fair Assignment Policy}\label{sec:fairness_du_cu_policy}

We leverage the above results of regret and the expressions for the primal and dual updates to design the optimistic FTRL policy for the assignment problem; see Algorithm \ref{alg:main-alg-fair-balance}. The initialization (lines 1-2) requires minimal information, i.e., the dimension of the primal and dual-space constraint sets and the number of vBSs and servers. The first assignment is drawn randomly (line 3). After running the first slot with policy $\bm x_1$, we observe the utility and cost functions and the respective gradients (line 5), as these have been set by the adversary. Accordingly we calculate the primal and dual gradients (lines 6-7) using the provided closed-form expressions, and we obtain the predicted gradient vectors for the next slot (line 8). This information is used to build the primal and dual regularizers and calculate the assignment policy $\bm x_{t+1}$ and the dual variables, $\bm \theta_{t+1}$ and $\bm \phi_{t+1}$, that will be used during the next slot (line 9). These steps are repeated throughout the horizon $\c T$, which is not required as input to the algorithm nor has to be fixed in advance. The performance of Algorithm \ref{alg:main-alg-fair-balance} is characterized by the following theorem.
\begin{theorem}\label{the:regret-main}
	Algorithm~\ref{alg:main-alg-fair-balance} attains regret:
	\begin{align*}
		\bm{\mathcal{R}}_T(F_\alpha, F_\beta) &\leq  \frac{1}{T}\left(\frac{\sqrt{2}I}{\eta}+\frac{\eta I\log J}{2} \right) \sqrt{\sum_{t=1}^{T}\|\bm{g}_{t}+\bm{w}_{t}-\Tilde{\bm{g}}_{t}-\Tilde{\bm{w}}_{t} \|_{\infty}^2}+ \frac{4\sqrt{2}{D}_{\Theta}}{T}\! \sqrt{\!\sum_{t=1}^{T}\|\bm{\kappa}_{t} - \Tilde{\bm{\kappa}}_{t}\|_{2}^2} \\
		&+ \frac{4\sqrt{2}{D}_{\Phi}}{T}\! \sqrt{\!\sum_{t=1}^{T}\|\bm{\mu}_{t} \!-\! \Tilde{\bm{\mu}}_{t}\|_{2}^2} + \frac{1}{T}\sumT\left(\bm \theta_t \!- \bar{\bm \theta}_T\right)^\top \bm u_t(\bm x^\star)\!+\! \frac{1}{T}\sumT\left(\bm \phi_t \!- \bar{\bm \phi}_T\right)^\top \bm h_t(\bm x^\star)
	\end{align*}
	where $D_\Theta = (\frac{1}{u_{min}^\alpha} - \frac{1}{u_{{max}}^\alpha})\sqrt{I}$, $D_\Phi = (\frac{1}{h_{{min}}^\beta} - \frac{1}{h_{{max}}^\beta})\sqrt{J}$, $\bar{\bm \theta}_T\doteq\frac{1}{T}\sumT\bm \theta_t$,  $\bar{\bm \phi}_T\doteq\frac{1}{T}\sumT\bm \phi_t$.
\end{theorem}
\begin{proof}
	First, we observe that since we perform OFTRL on the dual variables $\bm \theta$, we get:
	\begin{align}
		\frac{1}{T}\sum_{t=1}^T \Psi_t^a(\bm \theta_t, \bm x_t) - \frac{1}{T}\sum_{t=1}^T \Psi_t^a(\bm \theta, \bm x_t) \leq \frac{R_T^\theta}{T} =  \frac{4\sqrt{2}{D}_{\Theta}}{T} \sqrt{\sum_{t=1}^{T}\|\bm{\kappa}_{t} - \Tilde{\bm{\kappa}}_{t}\|_{2}^2}. \notag
	\end{align}
	where the last step follows from Lemma \ref{the:quadratic-regret}. Similarly, for the dual variables $\bm \phi$, we get:
	\begin{align}
		\frac{1}{T}\sum_{t=1}^T \Psi_t^\beta(\bm \phi_t, \bm x_t) - \frac{1}{T}\sum_{t=1}^T \Psi_t^\beta(\bm \phi, \bm x_t) \leq \frac{R_T^\phi}{T} =  \frac{4\sqrt{2}{D}_{\Phi}}{T} \sqrt{\sum_{t=1}^{T}\|\bm{\mu}_{t} - \Tilde{\bm{\mu}}_{t}\|_{2}^2}. \notag
	\end{align}
	On the other hand, the OFTRL on the primal variables $\bm x$ with the above regularizer, yields: 
	\begin{align}
		&\frac{1}{T}\sum_{t=1}^T \Psi_t^a(\bm \theta_t, \bm{x}^\star) + \frac{1}{T}\sum_{t=1}^T \Psi_t^\beta(\bm{\theta}_t, \bm x^\star) - \frac{1}{T}\sum_{t=1}^T \Psi_t^a(\bm \theta_t, \bm x_t)- \frac{1}{T}\sum_{t=1}^T \Psi_t^\beta(\bm \theta_t, \bm x_t) \leq  \notag \\
		&\frac{R_T^x}{T}= \left(\frac{\sqrt{2}I}{\eta T}+\frac{\eta I\log J}{2T} \right) \sqrt{\sum_{t=1}^{T}\|\bm{g}_{t}+\bm{w}_{t}-\Tilde{\bm{g}}_{t}-\Tilde{\bm{w}}_{t} \|_{\infty}^2}, \label{eq:proof-load-1}
	\end{align}
where we applied the regret bound from Lemma \ref{the:entropic-regret}. Now, using \eqref{eq:proof-load-1} we can write:
\begin{align}
&\frac{1}{T}\sum_{t=1}^T \Psi_t^a(\bm \theta_t, \bm x_t) + \frac{1}{T}\sum_{t=1}^T \Psi_t^\beta(\bm \phi_t, \bm x_t) + \bm{\mathcal{R}}_T^x \geq \frac{1}{T}\sum_{t=1}^T \Psi_t^a(\bm \theta_t, \bm x^\star)+\frac{1}{T}\sum_{t=1}^T \Psi_t^\beta(\bm \phi_t, \bm{x}^\star) \notag \\
&=\frac{1}{T}\left[	\sum_{t=1}^T (-F_a)^\star(\bm \theta_t) - \bm \theta_t^\top \bm{u}_t(\bm x^\star)	\right] + \frac{1}{T}\left[	\sum_{t=1}^T (-F_\beta)^\star(\bm \phi_t) - \bm \phi_t^\top \bm{h}_t(\bm x^\star)	\right] \notag\\
&\geq (-F_a)^\star(\bm{\bar \theta}) - \bm{\bar \theta}^\top \left( \frac{1}{T}\sum_{t=1}^T \bm{u}_t(\bm x^\star)	\right) - \frac{1}{T}\sum_{t=1}^T (\bm \theta_t - \Bar{\bm \theta})^\top\bm u_t(\bm{x}^\star)\notag \\
&+(-F_\beta)^\star(\bm{\bar \phi}) - \bm{\bar \phi}^\top \left( \frac{1}{T}\sum_{t=1}^T \bm{h}_t(\bm x^\star)	\right) - \frac{1}{T}\sum_{t=1}^T (\bm \phi_t - \Bar{\bm \phi})^\top\bm h_t(\bm{x}^\star) \notag\\
&\geq \min_{\bm \theta\in\Theta}\left\{  (-F_\alpha)^\star(\bm{ \theta}) - \bm{\theta}^\top\left(\frac{1}{T}\sum_{t=1}^T\bm u_t(\bm{x}^\star)\right) \right\} - \frac{1}{T}\sum_{t=1}^T (\bm \theta_t - \Bar{\bm \theta})^\top\bm u_t(\bm{x}^\star) \notag \\
&+\min_{\bm \phi\in\Phi}\left\{ (-F_\beta)^\star(\bm\phi) - \bm\phi^\top\left(\frac{1}{T}\sum_{t=1}^T\bm h_t(\bm{x}^\star)\right) \right\} - \frac{1}{T}\sum_{t=1}^T (\bm \phi_t - \Bar{\bm \phi})^\top\bm h_t(\bm{x}^\star) \notag 
\end{align}
\begin{align}
&=F_a\left( \frac{1}{T}\sum_{t=1}^T\bm{u}_t(\bm x^\star)	\right) + F_\beta\left( \frac{1}{T}\sum_{t=1}^T\bm{h}_t(\bm x^\star)	\right)- \frac{1}{T}\sum_{t=1}^T (\bm \theta_t - \Bar{\bm \theta})^\top\bm u_t(\bm{x}^\star) - \frac{1}{T}\sum_{t=1}^T (\bm \phi_t - \Bar{\bm \phi})^\top\bm h_t(\bm{x}^\star) \notag
\end{align}
We conclude by rearranging and using the biconjugate equivalence \eqref{conj1} for $\Psi_t^a(\bm \theta_t, \bm x_t)$, $ \Psi_t^\beta(\bm \phi_t, \bm x_t)$.
\end{proof}

\begin{algorithm}[t]
\begin{small}
\caption{{Fair and Balanced Assignment Policy (non-RT RIC)}} \label{alg:main-alg-fair-balance}
\begin{algorithmic}[1]
\Require{$I$, $J$, Multi-simplex $\mathcal{X} \in \mathbb{R}^{\textit{I}\times J}$, $\alpha, \beta\geq 0$, $[u_{min}^{\alpha}, u_{max}^{\alpha}]$, $[h_{min}^{\beta}, h_{max}^{\beta}]$. }
\State $\Theta = \left[-1/u_{min}^{\alpha}, -1/u_{max}^{\alpha} \right]^I$, $\Phi=\left[-1/h_{min}^{\beta}, -1/h_{max}^{\beta}\right]^J$ \Comment{\textit{\color{blue}{\footnotesize{Initialize the dual spaces}}}}
\State $\eta=\min\left\{ {1}/{2}, \sqrt{{2\sqrt 2}/{\log J}} \right\}, \ \sigma={2\sqrt{2}}/{D_\Theta}, \ \xi={2\sqrt{2}}/{D_\Phi}$	\Comment{{\textit{\color{blue}{\footnotesize{Initialize the regul. parameters}}}}}	
\State $\bm{x}_1 \in \mathcal{X}, \bm{\theta}_1 \in \Theta, \, \bm{\phi}_1 \in \Phi$\Comment{\textit{\color{blue}{\footnotesize{Initialize primal and dual vars}}}}
\For{$t=1 \textbf{ to } T$}
\State Observe $\bm u_t(\bm x_t)$, $\bm h_t(x_t)$, $\nabla_{\bm x}\bm u_t(\bm x_t)$, $\nabla_{\bm x}\bm h_t(\bm x_t)$ 			\Comment{\textit{\color{blue}{\footnotesize{Adversary selects losses}}}}
\State Compute primal gradients $\bm g_t$, $\bm w_t$ with \eqref{grad-xz-1}
\State Compute dual gradients $\bm \kappa_t$, $\bm \mu_t$ with \eqref{grad-xz-2}
\State Obtain predictions ${\bm{\p g}}_{t+1}$, ${\bm{\p w}}_{t+1}$, ${\bm{\p k}}_{t+1}$ and ${\bm{\p \mu}}_{t+1}$
\State Compute $\bm x_{t+1}$ with~\eqref{eq:closed-form}, $\bm \theta_{t+1}$ and  $\bm \phi_{t+1}$ with ~\eqref{eq:closed-form-quad} \Comment{\textit{\color{blue}{\footnotesize{Update assignment and dual vars}}}}
\EndFor
\end{algorithmic}
\end{small}
\end{algorithm}

\textbf{Discussion}. There are some important notes in order here. First, observe the last two terms in the regret bound which quantify how much each dual vector deviates from its average (over $\mathcal T$). {These deviations depend on the type of the adversary, and remain sublinear under certain general conditions. Namely, the utility and cost functions can change in a non-i.i.d. fashion, even arbitrarily, as long as their perturbations remain within a sublinearly-growing perturbation budget. And there are two types of such budgets: budgeted severity, where we measure the severity of the adversary by summing the absolute value of (utility and cost) perturbations for the entire time-horizon; and partitioned severity, where we divide the time-horizon into contiguous partitions and calculate the absolute value of perturbations over each partition. As long as the perturbations satisfy at least one budget condition, the regret will remain sublinear. We refer the reader to \cite{bib:tareq_fairness} for further details, and stress that this condition is significantly milder than those in prior static or stochastic fairness frameworks \cite{bib:tassiulas-monograph, bib:neely-monograph, bib:altman2012multiscale}. We provide instances of such adversarial environments in Sec.~\ref{sec:applications}.} 

The theorem also highlights the effect of predictions. The first two terms of the regret bound are eliminated when the predictions are perfect, while the algorithm suffers additional regret which is commensurate to the prediction errors (measured with the $\ell_{\infty}$ norm). In any case, these terms remain below $\c O(\sqrt T)$. This reveals that predictions expedite the learning process while we retain the worst-case guarantees when they are inaccurate. Observe also that the bound depends on the numbers of servers only logarithmically, a known advantage of entropic regularizers, but has linear dependency on the number of vBSs. This is due to the structure of the constraint set $\mathcal X$ which consists of $I$ (not 1) simplices. Similarly, the diameters $D_\Theta$ and $D_\Phi$, which depend on the minimum and maximum utility and cost values, affect only linearly the regret bound. 

Finally, regarding its implementation, leveraging the closed-form expressions for the decision updates, Algorithm \ref{alg:main-alg-fair-balance} can be executed with $\c O(1)$ memory and $\c O(1)$ calculations, without the need to solve any optimization problem at runtime. 
At the same time, the algorithm is oblivious to user demands, system state (e.g., costs and available capacity), and channel conditions. These two features, along with its general convergence properties, make the proposed framework particularly useful from a practical point of view. As a last note, we wish to stress that our work advances the state-of-the-art by using closed-form expressions and predictions, and importantly by combining two different fairness metrics. An O-RAN operator will, of course, need to normalize carefully the utility and cost functions in order to achieve the desirable balance of these metrics, which is also affected by the values of $\alpha$ and $\beta$. For instance, one can divide each function with its maximum attainable value or simply scale them with a properly-selected parameter. We explore this aspect experimentally in Sec.~\ref{sec:applications-fair-balanced}.

\section{Fair Service of Users and vBS Cost Minimization}\label{sec:fairness-cost}

\begin{figure*}[t!]
	\centering
	\begin{subfigure}[t]{0.48\textwidth}
		\centering
  \includegraphics[width=0.92\textwidth]{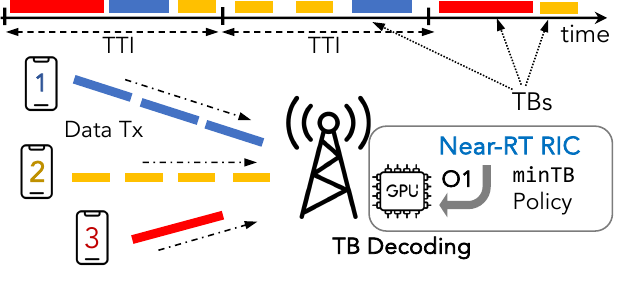}
        \caption{}
	\end{subfigure}%
	\hfill
	\begin{subfigure}[t]{0.48\textwidth}
		\centering
		\includegraphics[width=0.95\textwidth]{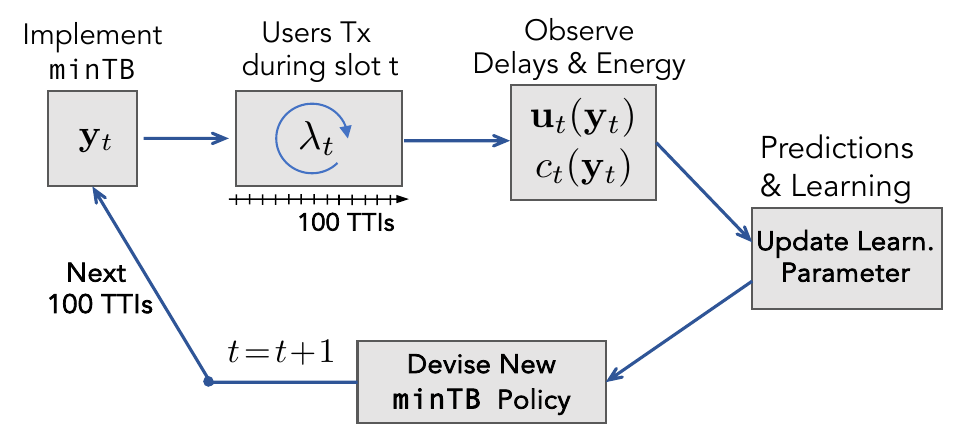}
        \caption{}
	\end{subfigure}
	\vspace{-2mm}
	\caption{\textbf{(a):} The near-RT controller decides the TB threshold (\texttt{minTB}) policy for each user at each slot; and the TBs are processed at a HA-equiped server. \textbf{(b):} Timing diagram of the learning algorithm for the~\texttt{minTB} policy.} 
  \label{fig:model-cost}	
\end{figure*}

Next, we study how a vBS can serve fairly its users in terms of latency by controlling the minimum size of their transmitted TBs, and minimize its own energy cost at the same time. Setting a threshold for the minimum TB size, the vBS prevents short TB transmissions that, as our experiments show (Sec. \ref{sec:applications-minTB}) increase the energy cost. On the other hand, such thresholds introduce waiting times for users that might be non-negligible, e.g., for latency-critical services. According to O-RAN specifications and previous feasibility studies, e.g., \cite{vrain_conf, ayala2021bayesian, bib:ayala-edgebol}, such radio control policies can be devised by a near-RT RIC and implemented with msec granularity and on per-user basis. We abuse slightly the notation here by redefining some parameters and variables.

\subsection{Model}
We consider a vBS that serves a set $\c I$ of $I=|\c I|$ users during a time period of $T$ slots, where each slot consists of $N$ TTIs (e.g., $N=100$), and we focus on the uplink again. During each slot $t$, each user $i\in \c I$ creates a certain amount of traffic (bytes) that needs to be transmitted to the vBS. We denote with $\mu_{itn}\geq 0$ the bytes created by user $i$ from the beginning of the slot up to TTI $n$, and define the vectors $\bm{\mu}_{it}=(\mu_{itn}, n\leq N)$ for each user $i$ and each slot $t$, and the vector $\bm{\mu}_{t}=(\bm{\mu}_{it}, i\in \c I)$ for the data of all users in slot $t$. The uplink transmission of a user is realized as soon as, and as many times as, its accumulated buffer load reaches the minimum TB size \texttt{minTB}. We denote with $\bm{y}_t=(y_{it}\geq 0, i\in\c I)$ the vector of \texttt{minTB} values for slot $t$ , which in the general case can be different for each user. These values are upper-bounded by the total number $K$ of transport blocks a vBS can support\footnote{Depending on the channel conditions, the actual number and size of transport blocks the vBS can support might fluctuate. Here, $K$ is the maximum possible number, and at each slot the exact bound is set by the vBS real-time scheduler.}. Hence, each $\bm{y}_t$ belongs to the set $\c Y = [0,K]^{I}$.

The \texttt{minTB} strategy $\bm{y}_t$ is decided by the vBS at the beginning of each slot in order to balance the service latency and its energy cost when processing the transmitted data. Our experiments show that large TB values improve the energy consumption per processed bit (J/b); yet they induce longer waiting times for the user traffic, see Sec. \ref{sec:applications}. Clearly, the more data is required before an uplink transmission is initiated, the more the user needs to wait to receive service. We consider a general model where the utility function $u_{it}:\mathbb R^{I}\mapsto \mathbb R_+$ denotes the (expected) performance perceived by user $i$ when the \texttt{minTB} strategy is $\bm y_t$. The vector $\bm u_t(\bm y)=(u_{it}(\bm y), i\in\c I)$ can measure directly the latency or a proxy metric such as time the user (MAC layer) buffer is empty\footnote{Recall that buffer queue length minimization is commonly used for reducing network delay, see e.g., \cite{bib:delay-neely}.} as in \cite{vrain_conf}. Furthermore, we denote with $c_t:\mathbb R^I\mapsto \mathbb R_+$ the vBS energy cost, which is considered to be convex and decreasing on $\bm y$. Our analysis below does not require any further assumptions on these utility and energy cost functions, while in Sec. \ref{sec:applications} we provide examples based on testbed measurements. 

The vBS aims to maximize the long-term latency fairness and minimize the average energy cost: 
\begin{align}\label{eq:latency_cost}
	G_\alpha(\{\bm y_t\}_t)={F_\alpha \left(\frac{1}{T}\sum_{t=1}^T \bm u_t(\bm{y}_t)\right) - \frac{1}{T}\sum_{t=1}^T c_t(\bm{y}_t)},
\end{align}
and to do so with a dynamic \texttt{minTB} policy $\{\bm y_t\}_t$ which ensures sublinear regret:
\begin{align}\label{eq:fairness-efficient-regret}
	\bm{\mathcal{R}}_T(G_\alpha) \doteq \sup\limits_{ \{\bm u_t, c_t\}_{t=1}^T }\left\{G_\alpha\big(\bm{y^\star}\big) - 	G_\alpha(\{\bm y_t\}_t)\right\}
\end{align}
where $G_\alpha\left(\bm{y^\star}\right)$ is the best performance (fairness and cost) that can be achieved if at $t=0$ the utilities and costs for the entire $\c T$ were known. This metric differs from the fairness-only criterion of the previous section due to the requirement for cost reduction and the constraints' geometry.

\subsection{Algorithm \& Regret Bounds}

The algorithm for this problem is based on the following modified proxy function:
\begin{align}\label{eq:proxy-function-cost}
	\Psi_{t}^c(\bm \theta_t, \bm y_t) \doteq (-F_\alpha)^\star(\bm\theta_t) - \bm\theta_t^\top\bm u_t(\bm y_t)-c_t(\bm y_t).
\end{align}
The analysis is based on the observation that the addition of the cost function $c_t(\cdot)$, which is independent of the dual variables, does not affect the algebraic operations on the proxy function. The primal OFTRL update is:
\begin{align}
\bm y_{t+1}=\arg\min_{\bm x\in \c Y}\left\{ r_{1:t}(\bm y) - \bm y^\top \big( \bm s_{1:t}+ \bm{\p s}_{t+1}\big) \right\}, \ \ \text{with} \ \ r_{1:t}(\bm y)=\frac{\eta\|\bm y\|^2}{2}\sqrt{\sum_{\tau=1}^{t}\|\bm s_\tau - \bm{\p s}_\tau\|_{2}^2 } \label{primal-update-3}
\end{align}
where $\bm s_t\!=\!\nabla_{\bm y} \Psi_{t}^c(\bm \theta_{t}, \bm y_{t})$ is the gradient of the proxy function w.r.t. the primal variables in slot $t$, and includes both the utility and the cost function differential (a linear operation), and $\bm{\p s}_{t+1}$ is the respective utility and cost gradient prediction for $t\!+1$. Similarly, the dual update is:
\begin{align}
\bm \theta_{t+1}=\arg\min_{\bm \theta \in \Theta}\left\{q_{1:t}(\bm \theta)+\bm \theta^\top(\bm m_{1:t}+\bm{\wp m}_{t+1}) \right\},  \quad \text{with} \ \ \ q_{1:t}(\bm \theta)=\frac{\sigma\|\bm \theta\|^2}{2}\sqrt{\sum_{\tau=1}^{t}\|\bm m_\tau - \wp{\bm{ m}}_\tau\|_{2}^2 } \label{dual-update-3}
\end{align}
where  $\bm m_{t}\!=\!\nabla_{\bm \theta} \Psi_{t}^c(\bm \theta_{t}, \bm y_{t})$. The detailed steps of the method are outlined in Algorithm \ref{alg:main-alg-cost}, which follows the same template as Algorithm \ref{alg:main-alg-fair-balance}, sans the proxy function and the gradient definition (and its prediction) in the primal space. The regret of Algorithm \ref{alg:main-alg-cost} is summarized next.
\begin{theorem}\label{eq:proxy-cost-regret}
Algorithm~\ref{alg:main-alg-cost} attains regret:
\begin{align*}
    \bm{\mathcal{R}}_T(G_\alpha) \!\leq   \frac{4\sqrt 2{D}_{\mathcal Y}}{T}\sqrt{\!\sum_{t=1}^{T}\!\|\bm{s}_{t}\!-\Tilde{\bm{s}}_{t} \|_{2}^2} + \frac{4\sqrt 2{D}_{\Theta}}{T} \sqrt{\!\sum_{t=1}^{T}\|\bm{m}_{t} \!- \wp{\bm{m}}_{t}\|_{2}^2} + \frac{1}{T}\!\sumT\!\left(\bm \theta_t\! - \bar{\bm \theta}_T\right)^\top \bm u_t(\bm y^\star)
\end{align*}
where $D_\Theta = (\frac{1}{u_{min}^{\alpha}} - \frac{1}{u_{max}^{\alpha}})\sqrt{I}$, $\bar{\bm \theta}_T=(1/T)\sum_{t=1}^T \bm \theta_t$, and $D_{\mathcal Y}$ is the diameter of set $\mathcal Y$.
\end{theorem}

\begin{algorithm}[t]
	\begin{small}
		\caption{{Fair and Cost-efficient \texttt{minTB} Policy}} \label{alg:main-alg-cost}
		\begin{algorithmic}[1]
			\Require{Compact convex set $\mathcal{Y} \in \mathbb{R}^{\textit{I}}$, $\alpha\geq 0$, $[u_{min}, u_{max}]$}
			\State $\Theta = \left[-1/u_{min}^{\alpha}, -1/u_{max}^{\alpha}\right]^I$
        \Comment{\textit{\color{blue}{\footnotesize{Initialize the dual space}}}}   
			\State $\sigma=2\sqrt{2}/D_\Theta, \ \eta=2\sqrt{2}/D_{\mathcal Y}$ \Comment{{\textit{\color{blue}{\footnotesize{Initialize the regul. parameters}}}}}   
			\State $\bm{y}_1 \in \mathcal{Y}, \bm{\theta}_1 \in \Theta$\Comment{\textit{\color{blue}{\footnotesize{Initialize the primal and dual vars}}}}      
			\For{$t=1 \textbf{ to } T$}
			\State Observe $\bm u_t(\bm y_t)$, $c_t(\bm y_t)$, $\nabla_{\bm y}\bm u_t(\bm y_t)$, $\nabla_{\bm y}c_t(\bm y_t)$\Comment{\textit{\color{blue}{\footnotesize{Incur reward and loss}}}}   
            \State Compute gradients $\bm s_{t}$, $\bm m_t$.
			\State Obtain gradient predictions ${\bm{\p s}}_{t+1}$ and $\wp{\bm{m}}_{t+1}$
			\State Compute $\bm y_{t+1}$ using~\eqref{primal-update-3} and $\bm \theta_{t+1}$ using~\eqref{dual-update-3}.
			\EndFor
		\end{algorithmic}
	\end{small}
\end{algorithm}

\textbf{Discussion}. The regret bound in the above Theorem verifies that the proposed OFTRL framework can deliver, also for this scenario, the desirable performance.  We see that the first two regret terms shrink proportionally to the prediction errors and in any case do not exceed $\mathcal O(\sqrt T)$. On the other hand, the residual last term captures the perturbation of the dual variables from their respective horizon-long average value, modulated by the optimal utility vector and depends on the adversary strategy, cf. {discussion of Theorem~\ref{the:regret-main} and }\cite{bib:tareq_fairness}. The execution of Algorithm \ref{alg:main-alg-cost} is lightweight as one can readily devise closed-form updates similar to those presented in Sec. \ref{sec:fairness_du_cu}, and, as such, suitable for the near-RT RIC. {Finally, it is worth stressing that one can extend the above model by scalarizing the two criteria, i.e., weighting the two metrics so as to reflect the operational priorities w.r.t. fairness of performance for the users versus the energy cost of the vBS. This scalarization serves also the purpose of unifying the units of measure. We elaborate further on this aspect in Sec. \ref{sec:applications}.}

\section{Performance Evaluation}\label{sec:applications}

We evaluate the proposed algorithms in a range of scenarios under realistic conditions.
First, we use a simulator to assess the regret and performance - cost trade-off in these problems. The simulator uses traffic traces obtained from a real-world operational network and employs utility and cost functions that are built using measurements. 
Secondly, we implement the algorithms in an O-RAN-compliant experimental platform that follows the design principles in \cite{experim_platform}. Thus, we measure the actual energy consumption and the processing latency of different baseband processors: two HAs and a pool of CPU cores. The platform uses two Nvidia GPU V100 as HAs and implements the O-RAN Acceleration Abstraction Layer (AAL) using Intel DPDK BBDev\footnote{https://doc.dpdk.org/guides/prog guide/bbdev.html} according to specifications \cite{oran-aal}.   {The AAL abstracts the O-Cloud computing resources as Logical Processing Units (LPUs).}
  {Note that the HAs consist of PCI boards which, although being faster in processing the workloads, they incur additional latency to transfer data from the software controller to the HA through a PCI bus, \cite{mbakoyiannis2018energy}.  This latency is accounted for in our experimental setup, as it is part of the GPU processing time.} For the CPU, we use an Intel Xeon Gold 6240R CPU with 32 cores, where 16 of them are assigned to signal processing tasks.

We generate the user traffic following the pattern from traces collected from a real BS using \cite{falcon}. Based on this, we generate the TBs, modulate them according to 5G specifications, add noise based on the SNR of the traces, and finally inject them into the system. The platform processes the incoming signals using the open-access software library Intel FlexRAN \cite{intel-flexran}. We measure the energy consumption using the drivers of each PU, i.e., \texttt{RAPL} and \texttt{nvidia-smi} for the CPU and GPU, respectively.  {Fig.~\ref{fig:testbed} presents a schematic of our experimental platform.}  {Finally, we note that in O-RAN architecture~\cite{bib:andres-magazine}, the non-RT and near-RT RICs operate closed-loops at, respectively, >1~second and 10-100~millisecond timescales. These timescales indicate how often the controller shall enforce a new policy~\cite{bib:edge-ric-sigcomm23}. To comply with such requirements, the application of the policy needs to be performed within a time window smaller than the timescale of the RICs. We confirm that all our algorithms require a negligible amount of time to execute (<10~ms), rendering them suitable to operate in the O-RAN RICs.}
 
\begin{figure*}[t!]
\centering
\includegraphics[width=0.8\textwidth]{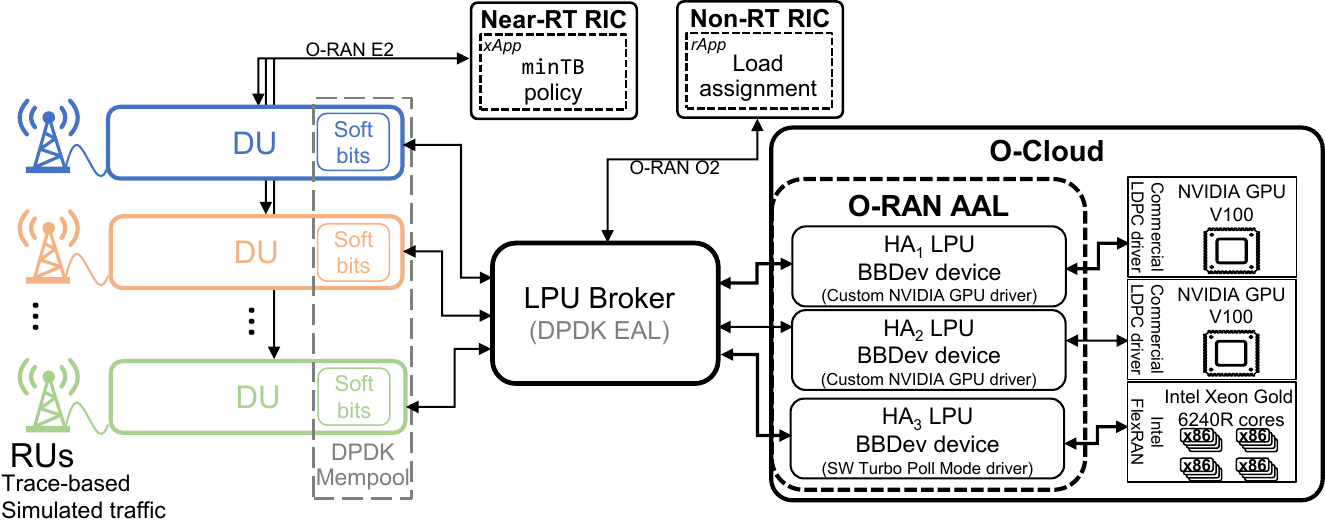}
\vspace{-0.5mm}
\caption{{{Schematic of the experimental platform, including the RICs and interfaces of the use cases.}}} 
\label{fig:testbed}
\end{figure*}

\subsection{Load Assignment Control Policy}\label{sec:applications-fair-balanced}

This section evaluates the vBSs' load assignment policy, which can be implemented as an rApp with a non-RT RIC at the SMO framework, and refers to a timescale of {1 second}. 

\subsubsection{Experimental Motivation}
In Fig. \ref{fig:assignment_motivation_01} we delve into the traffic trace (see also Fig. \ref{fig:intro_01}), to observe the high variability of the allocated radio resources and network conditions (evidenced from MCS) in a single cell. This highlights the importance of RIC control policies to be adaptive, a need that becomes even more crucial in small and/or mobile cells. Secondly, Fig. \ref{fig:assignment_motivation_02} presents the processing time and energy cost when a CPU server processes one TB, for different TB sizes and SNRs. Comparing these results with those in Fig.~\ref{fig:motivation12}, we find that CPU spends less energy per TB compared to a HA, especially for low SNRs, but this cost increases substantially with the TB size (amount of data). These findings highlight the potential benefits of an intelligent load assignment policy.

\begin{figure}[t!]
\centering
\minipage{0.48\columnwidth}
\includegraphics[width=\columnwidth]{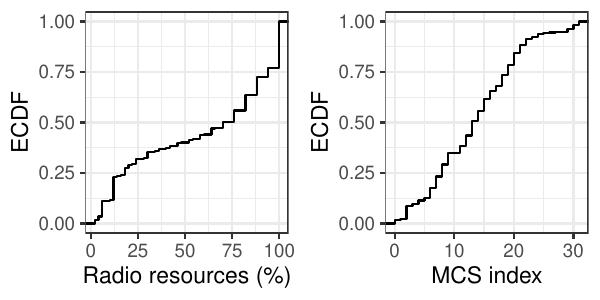}
\vspace{-4mm}
\caption{Cell load dynamics collected from an operational RAN in Frankfurt, Germany, May 2023.}
\label{fig:assignment_motivation_01}
 \endminipage{}
 \hfill
 \minipage{0.48\columnwidth}
\includegraphics[width=\columnwidth]{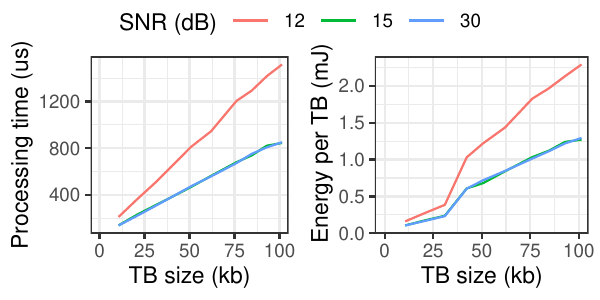}
\vspace{-5mm}
\caption{{Processing time (left) and energy consumption (right) when processing one TB with different sizes measured on an Intel Xeon CPU core.}}
\label{fig:assignment_motivation_02}
\endminipage{}
\vspace{-2mm}
\end{figure}

\subsubsection{Simulation Study}
We consider a simple model where the vBS utility\footnote{$u_{ijt}(\bm x_t)$ is concave on $\bm x_t$, see Appendix.} increases linearly with its load that is decoded at the assigned server, as long as the server is not overloaded, and it decreases rapidly when the server is assigned load that exceeds its capacity. In particular, the utility each vBS $i\in\mathcal I$ receives when sending $x_{ijt}\lambda_{it}$ load to a server $j\in\mathcal J$, is:
\begin{align}
    u_{ijt}(\bm x_t) = x_{ijt}\lambda_{it}\cdot \min\left\{1,\ 1 - \frac{1}{C_{jt}}\left( \sum_{k\in\mathcal I} \frac{x_{kjt}\lambda_{kt}}{n_{kt}}\big( \zeta_{kt}^jn_{kt}+o_{kt}^j \big)-C_{jt}\right)\right\}  \notag
\end{align}
where $\lambda_{kt}$ are the bytes sent by vBS $k\in\mathcal J$ during $t$ and $n_{kt}$ the average TB size of the flow (across all users). Parameters $\zeta_{kt}^j$ and $o_{kt}^j$ model the slope and intercept for the processing time of server $j$, for the (average) SNR of vBS $k$ during $t$; and we note that $\zeta_{kt}^j \approx 0$ for HA-based servers\footnote{{The values of these parameters can be non-zero (but still very small) for certain TB value ranges. The algorithm and analysis are readily applicable to those cases, as well.}}. These parameters are obtained by fitting measurements as those in Fig. \ref{fig:assignment_motivation_02}(left). Essentially the parenthesis term assess the portion of time that exceeds the server capacity, which we use to calculate how much vBS data are not decoded.  Note that we use a more coarse-grained estimation for the number of expected TBs here than the respective expression in Sec. \ref{sec:applications-minTB}, due to the aggregation over longer time periods {(1 sec instead of 100 msecs)} and over multiple base stations.

For the cost function, we study the general case where the monetary energy cost can be different for each server, and we define the respective price vector $\bm p_t=(p_{jt}\!>\!0, j\in\mathcal J)$ (cost/J). Based on our experiments, we define a different (average) energy saving function for each server type: 
\begin{align}
h_{jt}(\bm x_{jt})= {\varphi_h} p_{jt}\sum_{i\in\mathcal I}
\frac{(1-x_{ijt})\lambda_{it}}{n_{it}}\cdot \left(\delta_{it}^j n_{it} + \gamma_{it}^j\right), \quad j\in\mathcal J.
\end{align}
Parameters $\delta_{it}^j$ and $\gamma_{it}^j$ are the slope and intercept of the energy consumption profiles in Fig. \ref{fig:assignment_motivation_02}(right) and Fig. \ref{fig:motivation12}(a-right), { and $\varphi_h$ is the normalization parameter and it is set by the operator to prioritize throughput or energy}. For HA servers, the energy depends on the number of TBs and their average SNR ($\delta_{it}^j\approx 0$); while for legacy CPU servers, it also depends on the TB size. Recall that we define $h_{jt}$ as a cost reduction (energy savings) function, hence it is calculated w.r.t. the maximum possible cost for each server, i.e., when it serves all demand. 

We consider the following two scenarios where we simulate a stationary and a non-stationary environment by setting $\lambda_{it}, n_{it}, p_{jt}, C_{jt}$, and average SNR $s_{it}$, $\forall i\in \c I, j\in \c J, t\in\c T$, as: 
\begin{itemize}
\item \emph{Scenario 1 (Synthetic, Stationary)}. $I=5$ vBSs and $J=4$ servers, and the parameters are drawn randomly from uniform distributions: $\lambda_{it}\!\sim \mathcal U[4\cdot10^6, 6\cdot10^6), n_{it}\!\sim\mathcal U[4\cdot10^4, 6\cdot10^4), s_{it}\!\sim\mathcal U[10, 30)$, $p_{jt}\!\sim \mathcal U[10, 15)$, and $C_{jt}\!\sim \mathcal U[0, 10)$. We set $\alpha\!=\beta\!=\varphi_h\!=1$, unless stated otherwise.

\item \emph{Scenario 2 (Synthetic, Non-stationary)}. 
$\bm C_t$ follows a periodic pattern, while $\bm \lambda_t, \bm n_t, \bm s_t$, and $\bm p_t$ have vanishing perturbations. We draw the mean values for $\bm C_t$, $\bm \lambda_t, \bm n_t, \bm s_t$, and $\bm p_t$ from $\mathcal U[0, 10)$, $\mathcal U[4\cdot10^6, 6\cdot10^6), \mathcal U[4\cdot10^4, 6\cdot10^4), \mathcal U[10, 30)$, and $\mathcal U[10, 15)$, respectively; and we perturb $\bm C_t$ with a sine wave of period $\sqrt{T}$, $\bm \lambda_t, \bm n_t, \bm s_t$ with vanishing Gaussian noise scaled with $t^{-1}$, and $\bm p_t$ with vanishing Gaussian noise scaled with $0.1t^{-1}$. We set $\alpha\!=\beta\!=\varphi_h\!=1$. {$\bm C_t$ has sublinear partitioned severity and the other parameters have sublinear budgeted severity.}
\end{itemize}

\begin{figure*}[t]
%\begin{minipage}[b]{0.95\textwidth}
    \centering
    %\begin{subfigure}[b]{0.38\textwidth}\
    \begin{subfigure}[b]{0.28\textwidth}    
        \centering
        \includegraphics[width=\textwidth, page=10]{eval_figs.pdf}
        \caption{}
        \label{fig:regret-62}
    \end{subfigure}    
    %\begin{subfigure}[b]{0.6\textwidth} 
    \hfill
    \begin{subfigure}[b]{0.44\textwidth} 
        \centering
        \includegraphics[width=\textwidth, page=9]{eval_figs.pdf}
        \caption{}
        \label{fig:en-62}
    \end{subfigure}
    \hfill
    %\begin{subfigure}[b]{0.38\textwidth}
    \begin{subfigure}[b]{0.24\textwidth}    
        \centering
        \includegraphics[width=\textwidth]{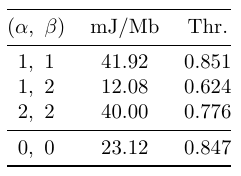}
        \vspace{0.8mm}
        \caption{}
        \label{fig:table-ab}
    \end{subfigure}     
    \vspace{-2mm}
    \caption{{\textbf{(a):} Fairness regret of Algorithm \ref{alg:main-alg-fair-balance} for different number of slots $T$, averaged over 5 runs. \textbf{(b):} Energy dispersion amongst PUs, average throughput per vBS, and energy consumption per bit under the Horizon-Fair (Algorithm \ref{alg:main-alg-fair-balance}, $\alpha=1, \beta=2$), Slot-Fair ($\alpha=1, \beta=2$) and Utilitarian Algorithm ($\alpha\!=\beta=0$). \textbf{(c):} PU energy consumption ($\rm mJ/Mbit$) \& Throughput (ratio of decoded TBs), for various $\alpha$ and $\beta$ values}.\label{fig:62}}
\end{figure*}
We first estimate the intercept and slope parameters $\zeta_{it}^j, o_{it}^j, \delta_{it}^j$, and $\gamma_{it}^j$ using the measurements obtained from the testbed and linear regression for each SNR value; and then use Algorithm~\ref{alg:main-alg-fair-balance}.
The horizon-fair regret is shown in Fig.~\ref{fig:62}(a), aggregated over 5 independent runs. Aligned with the theoretical analysis, the algorithm achieves sublinear (in fact, negative) fairness regret in both the stationary and non-stationary scenarios. Indeed, we observe the convergence in these experiments is particularly fast, as it requires only a few tens of slots to reach the performance of the benchmark.
\subsubsection{Experimental Evaluation} 
Next, we evaluate the algorithm on an O-RAN compliant platform \cite{experim_platform}. We consider a setting with 5 identical vBSs with 100 users each. As explained above, the traffic generation and SNR patterns are based on traffic traces collected from real BSs, {we scale the energy saving function by setting $\varphi_h=\max\{\lambda_{it}\}/\max\{p_{jt}\}$ so that both the utility ($ u_{it}(\bm x_t)$) and energy saving ($h_{jt}(x_{jt})$) functions are scaled between $0$ and $\max\{\lambda_{it}\}$}. We measure the PUs energy and normalized throughput (ratio of successfully decoded TBs) from the experimental platform at TTI granularity and aggregate the measurements to produce decisions at the non-RT timescale.
In order to emulate heterogeneous HAs, we half the speed of the second GPU (using the Nvidia drivers) and artificially double its energy cost in our measurements. We also consider two identical CPUs.

For comparison, we implement two new algorithms, namely a slot-fair algorithm, in line with suggestions in \cite{bib:slot-fairness-sinclair, bib:slot-fairness-sadegh, bib:slot-fairness-jalota}; and an algorithm that maximizes the aggregate system utility (utilitarian), without catering for any type of fairness. 
Namely, the objective of the slot-fair algorithm is to maximize the fairness in each slot, whereas the objective of the utilitarian algorithm is to maximize the sum of HAs' energy savings and vBSs' utilities. To simplify the comparison, we use the non-optimistic versions of our algorithm.

The Table in Fig.~\ref{fig:62}(c) summarizes the trade-off between the average throughput per vBS and the PUs' energy consumption when we impose the fairness criteria. As expected, the utilitarian algorithm (i.e., $\alpha\!=\beta\!=0$) outperforms Algorithm~\ref{alg:main-alg-fair-balance} in regards to total throughput and energy. However, Fig.~\ref{fig:62}(b) clearly shows that the utilitarian solution directs most demand to GPU1 and the CPUs, and does not employ GPU2 which is intentionally designed to be slower and more energy-consuming in this scenario.
That is, the utilitarian solution allows the maximization of energy savings and throughput by not using the worse GPU, since there are no fairness requirements. In contrast, the fair algorithms direct a significant portion of the vBSs demand to GPU2, increasing the energy consumption. The horizon fair algorithm is more fair than the slot fair algorithm with respect to energy dispersion amongst HAs, ending up almost equalizing the energy consumption of both GPUs. {Fig.~\ref{fig:62}(c) also indicates that modifying $\alpha$ and $\beta$ parameters has an unintentional effect on the prioritization of different objectives. Due to the exponential nature of $\alpha-$fairness function, increasing $\alpha$ prioritizes the utilities more, and results in more throughput; while increasing $\beta$ prioritizes the energy savings more, which results in reduced energy consumption.} {The decision of $\varphi_h$ should be made attentively to prevent any side effect when modifying the fairness parameters.}

\begin{figure*}[t!]
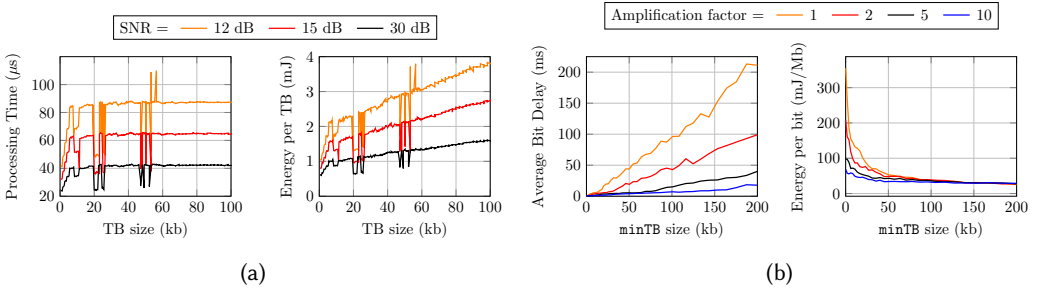

\centering
\begin{subfigure}[t]{0.49\textwidth}
\centering
\includegraphics[width=0.99\textwidth, page=12]{eval_figs.pdf}
\caption{}
\end{subfigure}%
\hfill 
\begin{subfigure}[t]{0.49\textwidth} 
\centering
\includegraphics[width=0.99\textwidth, page=13]{eval_figs.pdf}
\caption{}
\end{subfigure}
\vspace{-1mm}
\caption{\small{\textbf{(a):} Experimental measurement of computing time and energy consumption per TB, for different TB sizes. \textbf{(b):} Energy consumption per bit and average bit delay for a GPU HA as a function of the TB size threshold (\texttt{minTB}) and for different amplification factors applied to the traffic traces.}} 
\label{fig:motivation12}
\end{figure*}

\subsection{\texttt{minTB} Control Policy}\label{sec:applications-minTB}
Next, we evaluate the near-real-time compute control policy \texttt{minTB} that can be applied to each vBS independently. As in the previous section, we provide experimental motivation for the problem, run simulations with traces, and implement the solution at an O-RAN testbed.

\subsubsection{Experimental Motivation} 
Fig.~\ref{fig:motivation12}(a) presents the processing time and energy consumption of a GPU server (a common HA) for different TBs and SNRs. The processing time is practically independent of the TB size, an advantage stemming from the GPU's parallelization capability. Similarly, the energy consumption increases only slightly with the TB size. For example, with 15 dB SNR, the energy cost for 20 and 100 kb TBs ($5\times$ increase) is 1.7 and 2.8 mJ respectively ($0.5\times$ increase)\footnote{This small increase arises for very large load increments that require engaging additional processing elements of the GPU.}. Nevertheless, users often transmit TBs of small length, see \cite{bib:foukas-sigcom21} and our traces, thus inducing unnecessary energy costs. The \texttt{minTB} policy can tackle this issue. Indeed, in the experiments presented in Fig.~\ref{fig:motivation12}(b), we see how the \texttt{minTB} value affects the vBS energy and the delay for users with different loads. For example, with a minimum TB size of 25 kb, the energy consumption drops up to $4\times$ and the delay increases up to 20 msec, compared to when not using any threshold, which is currently the default implementation of software-defined base stations.

The \texttt{minTB} policy can be implemented as an xApp in the Near-RT RIC, which operates in slots of $\sim$100 ms. The users send data at each TTI (every 1 ms) and can provide feedback about their buffer status at such fine granularity. The algorithm selects the \texttt{minTB} value for each user and communicates (via the O1 interface) this rule to vBS, which is enforced by the radio scheduler at every TTI. The policy is updated every slot (100 TTIs), based on the users' feedback and experienced delay (calculated by the RIC), and the reported vBS energy consumption during the previous slots.

\subsubsection{Simulation Study}
The algorithms require a model for the utility and cost functions which we build using experimental results. Let us denote with $b_{it}$ the number of \textit{data generation events} of user $i$ during each slot $t\in \mathcal T$, and with $\rho_{it}$ the average number of bits generated at each such event, and define $\bm b_t=(b_{it}, i\in\mathcal I)$, $\bm \rho_t=(\rho_{it}, i\in\mathcal I)$. We assume that these values follow a Poisson distribution during each slot, but can change arbitrarily across different slots. Hence, we employ the following approximation for the HA energy cost: 
\begin{align}\label{eq:cost-61}
	c_t(\bm y)= \varphi\sum_{i\in\mathcal{I}} \beta(s_{it}) \pi_{it} = \varphi\sum_{i\in\mathcal{I}} \beta(s_{it}) b_{it} \frac{\rho_{it}}{y_i}\left( 1 - e^{-\frac{y_i}{\rho_{it}}} \right),
\end{align}
where $\pi_{it}$ is the expected number of TBs user $i$ will generate in slot $t$, $\beta(\cdot)$ is the mapping from SNR to a cost coefficient (as SNR affects the energy cost); $s_{it}$ is the average SNR of user $i$ in $t$ (calculated at the end of the slot); and $\varphi$ a normalization parameter that can prioritize cost over fairness, if necessary.
For the utility function, we use the percentage of time the user's buffer is empty \cite{vrain_conf}. This metric acts as a proxy for the delay. As the network is not in saturation most of the time, we assume that each user empties its buffer as soon as its data exceeds the TB threshold. Based on that, we derive the following approximation:
\begin{equation}\label{eq:util-61}
u_{it}(\bm y)= \mathbf{Pr}(B_{it} = 0) =\frac{\rho_{it}}{y_i}\left( 1 - e^{-\frac{y_i}{\rho_{it}}} \right),
\end{equation}
where $B_{it}$ is the number of bits in the buffer of user $i$ in time slot $t$, and $\mathbf{Pr}(B_{it}=0)$ the probability of empty buffer. Users experience more delay and the HA energy incurred by the user decreases as the value of the TB threshold increases. The approximations~\eqref{eq:cost-61} and~\eqref{eq:util-61} captures this dependency and both the cost and utility functions are decreasing functions. We validate these functions using real data gathered from the O-RAN platform (see Appendix).

We consider the following three scenarios for the simulations:
\begin{itemize}{\leftmargin=0.05cm}
\item \emph{Scenario 1 (Synthetic, Stationary)}. The vBS serves $I\!=\!10$ users, and the parameters are uniformly random as $b_{it}\sim\mathcal U(10,40)$, $\rho_{it}\sim \mathcal U(5\cdot10^4, 10^5)$, $s_{it}\sim\mathcal U(20, 30)$, $\forall i\in\mathcal I, t\in\mathcal T$.

\item \emph{Scenario 2 (Traced-driven)}. We use the above traces from a vBS obtained with \cite{falcon}, and generate the values for $b_{it}$, $\rho_{it}$, and $s_{it}$ for 5 users whose data parameters (i.e., $\rho_{it}$) are scaled by 1, 2, 4, 6 and 8.  

\item  \emph{Scenario 3 (Synthetic, non-stationary)}. We consider $I=5$ users with data generation and SNR values that follow an adversarial \emph{ping-pong} pattern which, further, is different for each user:
\begin{align*}
    b_{it} &= 
    \begin{cases}
        10 \ \ \text{ if } t < 2^{i-1}\pmod{2^i} \\
        40 \ \ \text{ otherwise }
    \end{cases},
    & s_{it} &= 
    \begin{cases}
        20 \ \ \text{ if } t < 2^{4-i}\pmod{2^{5-i}} \\
        30 \ \ \text{ otherwise }
    \end{cases}.
\end{align*}
{We restrict the adversary to have sublinear budgeted severity and set} $\bm \rho_t$ as $\rho_{it} = \Bar{\rho}_{i} + 10^4 n_{it}/t$ where $\Bar{\rho}_{i}\sim\mathcal U(5\cdot10^4, 10^5)$ and $n_{it} \sim \mathcal{N}(0, 1)$. 
\end{itemize}
The fairness parameter is set to $\alpha=1$, unless stated otherwise. Finally, we consider two types of predictions: good and moderate predictions. We obtain the prediction of the gradients in step 8 of Algorithm~\ref{alg:main-alg-fair-balance} by first calculating the actual gradient in the next slot and then adding a Gaussian noise, scaled with the gradient and accuracy coefficient. For the good predictions, the accuracy coefficient is $0.001$, whereas for moderate predictions, we set the coefficient to $0.3$. For instance, the good prediction of ${\bm{\p g}}_{t+1}$ is calculated as ${\bm{\p g}}_{t+1} = \bm{g}_{t+1} + 0.001\bm n_{t+1}\circ\bm{g}_{t+1}$ where $\bm n_{t+1}$ is Gaussian noise.

Fig.~\ref{fig:61}(a) plots the regret $\bm{\mathcal{R}}_T(G_\alpha)$ for the above scenarios and prediction models. In Scenario 1 (left), Algorithm \ref{alg:main-alg-cost} converges independently of the quality of predictions. In Scenario 2 (center), the algorithm with good predictions achieves lower regret and converges faster, as expected. However, due to high variations in the utility values of these traces (unrestricted adversary), the residual (last) term in Theorem~\ref{eq:proxy-cost-regret} is not eliminated. Finally, Fig.~\ref{fig:61}(a-right) shows the results for the restricted adversarial scenario where all algorithms achieve zero regret for $T\!>\!1000$.  

\begin{figure}[t!]
\centering
\begin{subfigure}[t]{0.59\textwidth}
\centering
\includegraphics[width=\textwidth, page=1]{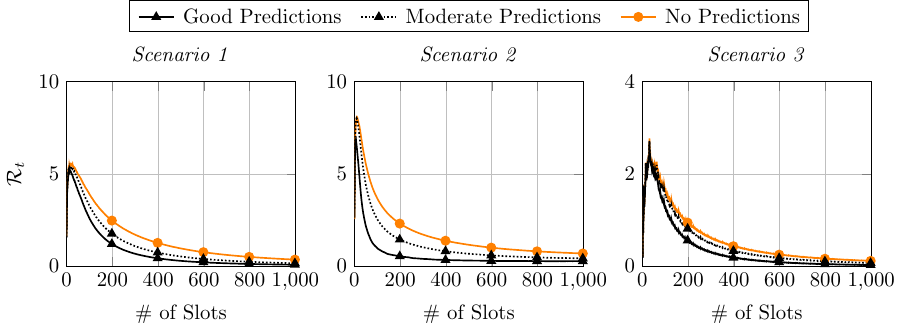}
\caption{}
\label{fig:regret-61}
\end{subfigure}%
\hfill 
\begin{subfigure}[t]{0.20\textwidth} 
\centering
\includegraphics[width=\textwidth, page=2]{eval_figs.pdf}
\caption{}
\label{fig:c-61}
\end{subfigure}
\begin{subfigure}[t]{0.20\textwidth} 
\centering
\includegraphics[width=\textwidth, page=6]{eval_figs.pdf}
\caption{}
\label{fig:c-62}
\end{subfigure}
\vspace{-2mm}
\caption{\small{\textbf{(a):} Regret $\bm{\mathcal{R}}_T(G_\alpha)$ of Algorithm \ref{alg:main-alg-cost} for different number of slots $T$, averaged over 10 runs. \textbf{(b):} Energy savings and user delay (performance) for different values of $\varphi$, and $\alpha=1.$ \textbf{(c) :} Delay of each user with respect to $\alpha$ and HA's energy consumption per bit (dashed lines), for $\varphi=50$.}}\label{fig:61}
\end{figure}

\subsubsection{Experimental Evaluation} 

Next, we implement and evaluate Algorithm \ref{alg:main-alg-cost} in a testbed. The \texttt{minTB} policy $\bm y_t$ is derived using the cost and utility models \eqref{eq:cost-61}-\eqref{eq:util-61}, and then implemented in the platform where we measure the actual energy consumption and (average) delay for each user.

First, we evaluate the effect of $\varphi$ which balances the importance of energy cost and user utility. In~Fig.~\ref{fig:61}(b), we compare the energy consumption of the \texttt{minTB} policies with the default policy where no threshold is applied to TBs, i.e., $\bm y = \bm 0$ (current default in such vBS). We calculate the energy saving of the \texttt{minTB} policy with respect to the default policy and plot the average measured delays. As $\varphi$ increases, the energy savings improve alongside an increase in average delay. We see, for instance, that we can save a remarkable amount of $67\%$ energy compared to the case no TB threshold is used, at the expense of $\sim$15msec additional average delay for the users; and we can save up to $72\%$ energy without incurring more than $\sim$38msec delay, by tuning the control parameter $\varphi$ accordingly. 
Next, we showcase how this delay is dispersed across the users. Fig. \ref{fig:61}(c) plots the average (over time) delay per user in Scenario 2 with 5 users for $\alpha=0$ (aggregate delay minimization), $\alpha=1$ (proportional fairness), and $\alpha\rightarrow \infty$ (max-min fairness). Since the data of each user might induce different energy costs due to their SNR and/or volume, the RIC will naturally apply a different \texttt{minTB} policy per user, hence inducing a different delay for each of them. The value of $\alpha$ affects these decisions directly. Indeed, we see that the delay dispersion is more fair when we set $\alpha=1$ and $\alpha\rightarrow \infty$; while the latter creates $\!>\!3\times$ more energy consumption.  
\section{Conclusions}

O-RAN, and similar virtualized RAN architectures, promise unprecedented performance and versatility for next generation of mobile networks, yet their energy costs are likely to constitute a prohibitive deployment factor. Motivated by this, we propose a radio-control policy and a compute-control (assignment) policy which cater for the energy consumption of vBSs and their O-Cloud processing units. The policies balance the user-perceived performance (throughput and transmission delay) with the network's energy costs, and importantly, disperse them \emph{fairly} across the users and the servers (respectively) throughout the entire operation of the system. The decision engine of the policies utilizes online learning algorithms (optimistic FTRL) that are tailored for the problem at hand, and as such is robust to a wide range of (unpredictable) parameter perturbations. We prove and demonstrate the optimality of these algorithms using a range of scenarios, both with simulations and testbed experiments, and measure energy savings (per vBS) up to $72\%$ when the users can tolerate $\sim 38$msec additional delay, on average.

\section*{Acknowledgments}

We would like to thank the anonymous reviewers and Igor Kadota (shepherd) for their valuable feedback that helped us improve this work. This work has been supported by the European Commission through Grant No. SNS-JU-101097083 (BeGREEN), 101139270 (ORIGAMI), and 101017109 (DAEMON) and CERCA Programme.

\bibliographystyle{ACM-Reference-Format}
\bibliography{ref}


\begin{thebibliography}{92}

%%% the Web address.

\ifx \showCODEN    \undefined \def \showCODEN     #1{\unskip}     \fi
\ifx \showDOI      \undefined \def \showDOI       #1{#1}\fi
\ifx \showISBNx    \undefined \def \showISBNx     #1{\unskip}     \fi
\ifx \showISBNxiii \undefined \def \showISBNxiii  #1{\unskip}     \fi
\ifx \showISSN     \undefined \def \showISSN      #1{\unskip}     \fi
\ifx \showLCCN     \undefined \def \showLCCN      #1{\unskip}     \fi
\ifx \shownote     \undefined \def \shownote      #1{#1}          \fi
\ifx \showarticletitle \undefined \def \showarticletitle #1{#1}   \fi
\ifx \showURL      \undefined \def \showURL       {\relax}        \fi
% invisible to TeX
\providecommand\bibfield[2]{#2}
\providecommand\bibinfo[2]{#2}
\providecommand\natexlab[1]{#1}
\providecommand\showeprint[2][]{arXiv:#2}

\bibitem[Agrawal and Devanur(2014)]%
        {bib:agrawal2014EC}
\bibfield{author}{\bibinfo{person}{Shipra Agrawal} {and}
  \bibinfo{person}{Nikhil Devanur}.} \bibinfo{year}{2014}\natexlab{}.
\newblock \showarticletitle{Bandits with Concave Rewards and Convex Knapsacks}.
  In \bibinfo{booktitle}{\emph{Proceedings of ACM EC}}.
  \bibinfo{pages}{989--1006}.
\newblock


\bibitem[Alcaraz et~al\mbox{.}(2020)]%
        {alcaraz2020online}
\bibfield{author}{\bibinfo{person}{Juan~J Alcaraz}, \bibinfo{person}{Jose~A
  Ayala-Romero}, \bibinfo{person}{Javier Vales-Alonso}, {and}
  \bibinfo{person}{Fernando Losilla-L{\'o}pez}.}
  \bibinfo{year}{2020}\natexlab{}.
\newblock \showarticletitle{Online Reinforcement Learning for Adaptive
  Interference Coordination}.
\newblock \bibinfo{journal}{\emph{Transactions on Emerging Telecommunications
  Technologies}} \bibinfo{volume}{31}, \bibinfo{number}{10}
  (\bibinfo{year}{2020}), \bibinfo{pages}{e4087}.
\newblock


\bibitem[ALLIANCE(2021)]%
        {oran-aal}
\bibfield{author}{\bibinfo{person}{O-RAN ALLIANCE}.}
  \bibinfo{year}{2021}\natexlab{}.
\newblock \bibinfo{title}{{O-RAN Acceleration Abstraction Layer General Aspects
  and Principles. O-RAN.WG6.AAL-GAnP-v01.00}}.
\newblock
\newblock


\bibitem[Alqerm and Shihada(2018)]%
        {alqerm2018sophisticated}
\bibfield{author}{\bibinfo{person}{Ismail Alqerm} {and} \bibinfo{person}{Basem
  Shihada}.} \bibinfo{year}{2018}\natexlab{}.
\newblock \showarticletitle{Sophisticated Online Learning Scheme for Green
  Resource Allocation in 5G Heterogeneous Cloud Radio Access Networks}.
\newblock \bibinfo{journal}{\emph{IEEE Transactions on Mobile Computing}}
  \bibinfo{volume}{17}, \bibinfo{number}{10} (\bibinfo{year}{2018}),
  \bibinfo{pages}{2423--2437}.
\newblock


\bibitem[Altman et~al\mbox{.}(2008)]%
        {bib:altman2008}
\bibfield{author}{\bibinfo{person}{Eitan Altman}, \bibinfo{person}{Konstantin
  Avrachenkov}, {and} \bibinfo{person}{Andrey Garnaev}.}
  \bibinfo{year}{2008}\natexlab{}.
\newblock \showarticletitle{Generalized $\alpha$-fair Resource Allocation in
  Wireless Networks}. In \bibinfo{booktitle}{\emph{Proceedings of IEEE CDC}}.
  \bibinfo{pages}{2414--2419}.
\newblock


\bibitem[Altman et~al\mbox{.}(2012)]%
        {bib:altman2012multiscale}
\bibfield{author}{\bibinfo{person}{Eitan Altman}, \bibinfo{person}{Konstantin
  Avrachenkov}, {and} \bibinfo{person}{Sreenath Ramanath}.}
  \bibinfo{year}{2012}\natexlab{}.
\newblock \showarticletitle{Multiscale Fairness and its Application to Resource
  Allocation in Wireless Networks}.
\newblock \bibinfo{journal}{\emph{Computer Communications}}
  \bibinfo{volume}{35}, \bibinfo{number}{7} (\bibinfo{year}{2012}),
  \bibinfo{pages}{820--828}.
\newblock


\bibitem[{Analysys Mason}(2023a)]%
        {mason-tco}
\bibfield{author}{\bibinfo{person}{{Analysys Mason}}.}
  \bibinfo{year}{2023}\natexlab{a}.
\newblock \bibinfo{title}{{Key TCO Considerations for Economically Viable Open
  RAN }}.
\newblock \bibinfo{howpublished}{Strategy report}.
\newblock


\bibitem[{Analysys Mason}(2023b)]%
        {mason-forecast}
\bibfield{author}{\bibinfo{person}{{Analysys Mason}}.}
  \bibinfo{year}{2023}\natexlab{b}.
\newblock \bibinfo{title}{{Open RAN: Translating the Hype into Revenue.}}
\newblock \bibinfo{howpublished}{Webinar}.
\newblock


\bibitem[Anderson et~al\mbox{.}(2023)]%
        {bib:llp}
\bibfield{author}{\bibinfo{person}{Daren Anderson}, \bibinfo{person}{George
  Iosifidis}, {and} \bibinfo{person}{Douglas Leith}.}
  \bibinfo{year}{2023}\natexlab{}.
\newblock \showarticletitle{Lazy Lagrangians for Optimistic Learning with
  Budget Constraints}.
\newblock \bibinfo{journal}{\emph{IEEE/ACM Transactions on Networking}}
  \bibinfo{volume}{31}, \bibinfo{number}{5} (\bibinfo{year}{2023}),
  \bibinfo{pages}{1935--1949}.
\newblock


\bibitem[Association(2020)]%
        {bib:energy-gsma}
\bibfield{author}{\bibinfo{person}{GSMA Association}.}
  \bibinfo{year}{2020}\natexlab{}.
\newblock \showarticletitle{5G Energy Efficiencies: Green is the New Black}.
\newblock \bibinfo{journal}{\emph{White Paper}} (\bibinfo{year}{2020}).
\newblock


\bibitem[Auer et~al\mbox{.}(2011)]%
        {auer2011much}
\bibfield{author}{\bibinfo{person}{Gunther Auer}, \bibinfo{person}{Vito
  Giannini}, \bibinfo{person}{Claude Desset}, \bibinfo{person}{Istvan Godor},
  \bibinfo{person}{Per Skillermark}, \bibinfo{person}{Magnus Olsson},
  \bibinfo{person}{Muhammad~Ali Imran}, \bibinfo{person}{Dario Sabella},
  \bibinfo{person}{Manuel~J Gonzalez}, \bibinfo{person}{Oliver Blume},
  {et~al\mbox{.}}} \bibinfo{year}{2011}\natexlab{}.
\newblock \showarticletitle{How Much Energy is Needed to Run a Wireless
  Network?}
\newblock \bibinfo{journal}{\emph{IEEE wireless communications}}
  \bibinfo{volume}{18}, \bibinfo{number}{5} (\bibinfo{year}{2011}),
  \bibinfo{pages}{40--49}.
\newblock


\bibitem[Auer et~al\mbox{.}(2002)]%
        {bib:auer-adaptive2002}
\bibfield{author}{\bibinfo{person}{Peter Auer}, \bibinfo{person}{Nicolo
  Cesa-Bianchi}, {and} \bibinfo{person}{Claudio Gentile}.}
  \bibinfo{year}{2002}\natexlab{}.
\newblock \showarticletitle{Adaptive and Self-confident Online Learning
  Algorithms}.
\newblock \bibinfo{journal}{\emph{J. Comput. System Sci.}}
  \bibinfo{volume}{64} (\bibinfo{year}{2002}), \bibinfo{pages}{48--75}.
\newblock


\bibitem[Ayala-Romero et~al\mbox{.}(2019a)]%
        {ayala2019online}
\bibfield{author}{\bibinfo{person}{Jose~A Ayala-Romero},
  \bibinfo{person}{Juan~J Alcaraz}, \bibinfo{person}{Andrea Zanella}, {and}
  \bibinfo{person}{Michele Zorzi}.} \bibinfo{year}{2019}\natexlab{a}.
\newblock \showarticletitle{Online Learning for Energy Saving and Interference
  Coordination in Hetnets}.
\newblock \bibinfo{journal}{\emph{IEEE Journal on Selected Areas in
  Communications}} \bibinfo{volume}{37}, \bibinfo{number}{6}
  (\bibinfo{year}{2019}), \bibinfo{pages}{1374--1388}.
\newblock


\bibitem[Ayala-Romero et~al\mbox{.}(2021a)]%
        {bib:ayala-edgebol}
\bibfield{author}{\bibinfo{person}{Jose~A Ayala-Romero},
  \bibinfo{person}{Andres Garcia-Saavedra}, \bibinfo{person}{Xavier
  Costa-Perez}, {and} \bibinfo{person}{George Iosifidis}.}
  \bibinfo{year}{2021}\natexlab{a}.
\newblock \showarticletitle{EdgeBOL: Automating Energy-Savings for Mobile Edge
  AI}. In \bibinfo{booktitle}{\emph{Proceedings of ACM CoNEXT}}.
  \bibinfo{pages}{397--410}.
\newblock


\bibitem[Ayala-Romero et~al\mbox{.}(2021b)]%
        {ayala2021bayesian}
\bibfield{author}{\bibinfo{person}{Jose~A Ayala-Romero},
  \bibinfo{person}{Andres Garcia-Saavedra}, \bibinfo{person}{Xavier
  Costa-Perez}, {and} \bibinfo{person}{George Iosifidis}.}
  \bibinfo{year}{2021}\natexlab{b}.
\newblock \showarticletitle{Orchestrating Energy-Efficient vRANs: Bayesian
  Learning and Experimental Results}.
\newblock \bibinfo{journal}{\emph{IEEE Transactions on Mobile Computing}}
  \bibinfo{volume}{22}, \bibinfo{number}{5} (\bibinfo{year}{2021}),
  \bibinfo{pages}{2910--2924}.
\newblock


\bibitem[Ayala-Romero et~al\mbox{.}(2019b)]%
        {vrain_conf}
\bibfield{author}{\bibinfo{person}{Jose~A Ayala-Romero},
  \bibinfo{person}{Andres Garcia-Saavedra}, \bibinfo{person}{Marco Gramaglia},
  \bibinfo{person}{Xavier Costa-Perez}, \bibinfo{person}{Albert Banchs}, {and}
  \bibinfo{person}{Juan~J Alcaraz}.} \bibinfo{year}{2019}\natexlab{b}.
\newblock \showarticletitle{vrAIn: A Deep Learning Approach Tailoring Computing
  and Radio Resources in Virtualized RANs}. In
  \bibinfo{booktitle}{\emph{Proceedings of MobiCom}}. \bibinfo{pages}{1--16}.
\newblock


\bibitem[Ayala-Romero et~al\mbox{.}(2022)]%
        {vrain_journal}
\bibfield{author}{\bibinfo{person}{Jose~A Ayala-Romero},
  \bibinfo{person}{Andres Garcia-Saavedra}, \bibinfo{person}{Marco Gramaglia},
  \bibinfo{person}{Xavier Costa-Pérez}, \bibinfo{person}{Albert Banchs}, {and}
  \bibinfo{person}{Juan~J Alcaraz}.} \bibinfo{year}{2022}\natexlab{}.
\newblock \showarticletitle{vrAIn: Deep Learning Based Orchestration for
  Computing and Radio Resources in vRANs}.
\newblock \bibinfo{journal}{\emph{IEEE Transactions on Mobile Computing}}
  \bibinfo{volume}{21}, \bibinfo{number}{7} (\bibinfo{year}{2022}),
  \bibinfo{pages}{2652--2670}.
\newblock


\bibitem[Ayala-Romero et~al\mbox{.}(2021c)]%
        {bib:vbs-experiments}
\bibfield{author}{\bibinfo{person}{Jose~A Ayala-Romero},
  \bibinfo{person}{Ihtisham Khalid}, \bibinfo{person}{Andres Garcia-Saavedra},
  \bibinfo{person}{Xavier Costa-Perez}, {and} \bibinfo{person}{George
  Iosifidis}.} \bibinfo{year}{2021}\natexlab{c}.
\newblock \showarticletitle{Experimental Evaluation of Power Consumption in
  Virtualized Base Stations}. In \bibinfo{booktitle}{\emph{Proceedings of IEEE
  ICC}}. \bibinfo{pages}{1--6}.
\newblock


\bibitem[Baek and Farias(2021)]%
        {bib:hori-fair-regr-baek}
\bibfield{author}{\bibinfo{person}{Jackie Baek} {and} \bibinfo{person}{Vivek
  Farias}.} \bibinfo{year}{2021}\natexlab{}.
\newblock \showarticletitle{Fair Exploration via Axiomatic Bargaining}.
\newblock \bibinfo{journal}{\emph{Proceedings of NeurIPS}},
  \bibinfo{pages}{22034--22045}.
\newblock


\bibitem[Beck(2017)]%
        {bib:beck-book}
\bibfield{author}{\bibinfo{person}{Amir Beck}.}
  \bibinfo{year}{2017}\natexlab{}.
\newblock \showarticletitle{First-Order Methods in Optimization}.
\newblock \bibinfo{journal}{\emph{MOS-SIAM Series on Optimization}}
  (\bibinfo{year}{2017}).
\newblock


\bibitem[Bega et~al\mbox{.}(2018)]%
        {bega2018cares}
\bibfield{author}{\bibinfo{person}{Dario Bega}, \bibinfo{person}{Albert
  Banchs}, \bibinfo{person}{Marco Gramaglia}, \bibinfo{person}{Xavier
  Costa-P{\'e}rez}, {and} \bibinfo{person}{Peter Rost}.}
  \bibinfo{year}{2018}\natexlab{}.
\newblock \showarticletitle{CARES: Computation-aware Scheduling in Virtualized
  Radio Access Networks}.
\newblock \bibinfo{journal}{\emph{IEEE Transactions on Wireless
  Communications}} \bibinfo{volume}{17}, \bibinfo{number}{12}
  (\bibinfo{year}{2018}), \bibinfo{pages}{7993--8006}.
\newblock


\bibitem[Bega et~al\mbox{.}(2019)]%
        {bega-NNSlicing-JSAC2020}
\bibfield{author}{\bibinfo{person}{Dario Bega}, \bibinfo{person}{Marco
  Gramaglia}, \bibinfo{person}{Marco Fiore}, \bibinfo{person}{Albert Banchs},
  {and} \bibinfo{person}{Xavier Costa-Perez}.} \bibinfo{year}{2019}\natexlab{}.
\newblock \showarticletitle{DeepCog: Optimizing Resource Provisioning in
  Network Slicing with AI-based Capacity Forecasting}.
\newblock \bibinfo{journal}{\emph{IEEE Journal on Selected Areas in
  Communications}} \bibinfo{volume}{38}, \bibinfo{number}{2}
  (\bibinfo{year}{2019}), \bibinfo{pages}{361--376}.
\newblock


\bibitem[Benade et~al\mbox{.}(2018)]%
        {bib:hori-fair-regr-benade}
\bibfield{author}{\bibinfo{person}{Gerdus Benade}, \bibinfo{person}{Aleksandr~M
  Kazachkov}, \bibinfo{person}{Ariel~D Procaccia}, {and}
  \bibinfo{person}{Christos-Alexandros Psomas}.}
  \bibinfo{year}{2018}\natexlab{}.
\newblock \showarticletitle{How to Make Envy Vanish Over Time}. In
  \bibinfo{booktitle}{\emph{Proceedings of ACM EC}}. \bibinfo{pages}{593--610}.
\newblock


\bibitem[Bertsekas(2016)]%
        {bib:bertsekas-nonlinear}
\bibfield{author}{\bibinfo{person}{Dimitri~P Bertsekas}.}
  \bibinfo{year}{2016}\natexlab{}.
\newblock \bibinfo{booktitle}{\emph{Nonlinear Programming, 3rd Edition}}.
\newblock \bibinfo{publisher}{Athena Scientific}.
\newblock


\bibitem[Bertsimas et~al\mbox{.}(2011)]%
        {bib:bertsimas2011price}
\bibfield{author}{\bibinfo{person}{Dimitris Bertsimas},
  \bibinfo{person}{Vivek~F Farias}, {and} \bibinfo{person}{Nikolaos
  Trichakis}.} \bibinfo{year}{2011}\natexlab{}.
\newblock \showarticletitle{The price of fairness}.
\newblock \bibinfo{journal}{\emph{Operations research}} \bibinfo{volume}{59},
  \bibinfo{number}{1} (\bibinfo{year}{2011}), \bibinfo{pages}{17--31}.
\newblock


\bibitem[Blankenship et~al\mbox{.}(2021)]%
        {Blankenship2021}
\bibfield{author}{\bibinfo{person}{Yufei Blankenship}, \bibinfo{person}{Dennis
  Hui}, {and} \bibinfo{person}{Mattias Andersson}.}
  \bibinfo{year}{2021}\natexlab{}.
\newblock \bibinfo{booktitle}{\emph{Channel Coding in NR}}.
\newblock \bibinfo{publisher}{Springer International Publishing},
  \bibinfo{address}{Cham}, \bibinfo{pages}{303--332}.
\newblock
\urldef\tempurl%
\url{https://doi.org/10.1007/978-3-030-58197-8_10}
\showDOI{\tempurl}


\bibitem[Bonald and Roberts(2015)]%
        {bib:bonald-sigmetrics15}
\bibfield{author}{\bibinfo{person}{Thomas Bonald} {and}
  \bibinfo{person}{James~W Roberts}.} \bibinfo{year}{2015}\natexlab{}.
\newblock \showarticletitle{Multi-Resource Fairness: Objectives, Algorithms and
  Performance}. In \bibinfo{booktitle}{\emph{Proceedings of ACM Sigmetrics}}.
  \bibinfo{pages}{31--42}.
\newblock


\bibitem[Bonati et~al\mbox{.}(2021)]%
        {bib:melodia-commag-21}
\bibfield{author}{\bibinfo{person}{Leonardi Bonati}, \bibinfo{person}{Salvatore
  D'Oro}, \bibinfo{person}{Michele Polese}, \bibinfo{person}{Stefano Basagni},
  {and} \bibinfo{person}{Tommaso Melodia}.} \bibinfo{year}{2021}\natexlab{}.
\newblock \showarticletitle{Intelligence and Learning in O-RAN for Data-Driven
  NextG Cellular Networks}.
\newblock \bibinfo{journal}{\emph{IEEE Commun. Mag.}} \bibinfo{volume}{59},
  \bibinfo{number}{10} (\bibinfo{year}{2021}), \bibinfo{pages}{21--27}.
\newblock


\bibitem[Buchbinder and Naor(2013)]%
        {bib:fair-online-naor13}
\bibfield{author}{\bibinfo{person}{Niv Buchbinder} {and}
  \bibinfo{person}{Joseph Naor}.} \bibinfo{year}{2013}\natexlab{}.
\newblock \showarticletitle{Fair Online Load Balancing}.
\newblock \bibinfo{journal}{\emph{Journal of Scheduling}}  \bibinfo{volume}{16}
  (\bibinfo{year}{2013}), \bibinfo{pages}{117--127}.
\newblock


\bibitem[Cayci et~al\mbox{.}(2020)]%
        {bib:hori-fair-regr-cayci}
\bibfield{author}{\bibinfo{person}{Semih Cayci}, \bibinfo{person}{Swati Gupta},
  {and} \bibinfo{person}{Atilla Eryilmaz}.} \bibinfo{year}{2020}\natexlab{}.
\newblock \showarticletitle{Group-fair Online Allocation in Continuous Time}.
\newblock \bibinfo{journal}{\emph{Proceedings of NeurIPS}}
  \bibinfo{volume}{33}, \bibinfo{pages}{13750--13761}.
\newblock


\bibitem[Darabi et~al\mbox{.}(2022)]%
        {bib:gpus-darabi-sigm22}
\bibfield{author}{\bibinfo{person}{Sina Darabi}, \bibinfo{person}{Negin
  Mahani}, \bibinfo{person}{Hazhir Bakhishi}, \bibinfo{person}{Ehsan
  Yousefzadeh-Asl-Miandoab}, \bibinfo{person}{Mohammad Sadrosadati}, {and}
  \bibinfo{person}{Hamid Sarbazi-Azad}.} \bibinfo{year}{2022}\natexlab{}.
\newblock \showarticletitle{NURA: A Framework for Supporting Non-Uniform
  Resource Accesses in GPUs}.
\newblock \bibinfo{journal}{\emph{Proceedings of the ACM on Measurement and
  Analysis of Computing Systems}} \bibinfo{volume}{6}, \bibinfo{number}{1}
  (\bibinfo{year}{2022}), \bibinfo{pages}{16:1--16:27}.
\newblock


\bibitem[{Dell}(2022)]%
        {bib:dell-2022}
\bibfield{author}{\bibinfo{person}{{Dell}}.} \bibinfo{year}{2022}\natexlab{}.
\newblock \showarticletitle{Dell Open RAN Accelerator Card}. In
  \bibinfo{booktitle}{\emph{{Solution Brief}}}.
\newblock


\bibitem[D'Oro et~al\mbox{.}(2022)]%
        {bib:melodia-commag-dApps}
\bibfield{author}{\bibinfo{person}{Salvatore D'Oro}, \bibinfo{person}{Michele
  Polese}, \bibinfo{person}{Leonardo Bonati}, \bibinfo{person}{Hai Cheng},
  {and} \bibinfo{person}{Tommaso Melodia}.} \bibinfo{year}{2022}\natexlab{}.
\newblock \showarticletitle{dApps: Distributed Applications for Real-Time
  Inference and Control in O-RAN}.
\newblock \bibinfo{journal}{\emph{IEEE Commun. Mag.}} \bibinfo{volume}{60},
  \bibinfo{number}{11} (\bibinfo{year}{2022}), \bibinfo{pages}{52--58}.
\newblock


\bibitem[Ericsson(2023)]%
        {bib:energy-erricson}
\bibfield{author}{\bibinfo{person}{Ericsson}.} \bibinfo{year}{2023}\natexlab{}.
\newblock \bibinfo{title}{Ericsson Unwraps New Energy Efficiency Solutions
  Designed for Open RAN Architecture}.
\newblock
  \bibinfo{howpublished}{\url{https://www.ericsson.com/en/news/2023/4/energy-efficient-ericsson-rapps-for-open-ran-architecture}}.
\newblock


\bibitem[Falcao et~al\mbox{.}(2010)]%
        {falcao2010massively}
\bibfield{author}{\bibinfo{person}{Gabriel Falcao}, \bibinfo{person}{Leonel
  Sousa}, {and} \bibinfo{person}{Vitor Silva}.}
  \bibinfo{year}{2010}\natexlab{}.
\newblock \showarticletitle{Massively LDPC Decoding on Multicore
  Architectures}.
\newblock \bibinfo{journal}{\emph{IEEE Transactions on Parallel and Distributed
  Systems}} \bibinfo{volume}{22}, \bibinfo{number}{2} (\bibinfo{year}{2010}),
  \bibinfo{pages}{309--322}.
\newblock


\bibitem[Falkenberg and Wietfeld(2019)]%
        {falcon}
\bibfield{author}{\bibinfo{person}{Robert Falkenberg} {and}
  \bibinfo{person}{Christian Wietfeld}.} \bibinfo{year}{2019}\natexlab{}.
\newblock \showarticletitle{FALCON: An Accurate Real-time Monitor for
  Client-based Mobile Network Data Analytics}. In
  \bibinfo{booktitle}{\emph{Proceedings of IEEE GLOBECOM}}. IEEE,
  \bibinfo{pages}{1--7}.
\newblock


\bibitem[Fossati et~al\mbox{.}(2020)]%
        {fossati2020multi}
\bibfield{author}{\bibinfo{person}{Francesca Fossati}, \bibinfo{person}{Stefano
  Moretti}, \bibinfo{person}{Patrice Perny}, {and} \bibinfo{person}{Stefano
  Secci}.} \bibinfo{year}{2020}\natexlab{}.
\newblock \showarticletitle{Multi-resource Allocation for Network Slicing}.
\newblock \bibinfo{journal}{\emph{IEEE/ACM Transactions on Networking}}
  \bibinfo{volume}{28}, \bibinfo{number}{3} (\bibinfo{year}{2020}),
  \bibinfo{pages}{1311--1324}.
\newblock


\bibitem[Foukas and Radunovic(2021)]%
        {bib:foukas-sigcom21}
\bibfield{author}{\bibinfo{person}{Xenofon Foukas} {and}
  \bibinfo{person}{Bozidar Radunovic}.} \bibinfo{year}{2021}\natexlab{}.
\newblock \showarticletitle{Concordia: Teaching the 5G vRAN to Share Compute}.
  In \bibinfo{booktitle}{\emph{Proceedings of ACM SIGCOMM}}.
  \bibinfo{pages}{580--596}.
\newblock


\bibitem[Galanopoulos et~al\mbox{.}(2020)]%
        {galanopoulos2020bayesian}
\bibfield{author}{\bibinfo{person}{Apostolos Galanopoulos},
  \bibinfo{person}{Jose~A Ayala-Romero}, \bibinfo{person}{George Iosifidis},
  {and} \bibinfo{person}{Douglas Leith}.} \bibinfo{year}{2020}\natexlab{}.
\newblock \showarticletitle{Bayesian Online Learning for MEC Object Recognition
  Systems}. In \bibinfo{booktitle}{\emph{Proceedings of IEEE GLOBECOM}}. IEEE.
\newblock


\bibitem[Garcia-Aviles et~al\mbox{.}(2021)]%
        {bib:nuberu}
\bibfield{author}{\bibinfo{person}{Gines Garcia-Aviles},
  \bibinfo{person}{Andres Garcia-Saavedra}, \bibinfo{person}{Marco Gramaglia},
  \bibinfo{person}{Xavier Costa-Perez}, \bibinfo{person}{Pablo Serrano}, {and}
  \bibinfo{person}{Albert Banchs}.} \bibinfo{year}{2021}\natexlab{}.
\newblock \showarticletitle{Nuberu: Reliable RAN Virtualization in Shared
  Platforms}. In \bibinfo{booktitle}{\emph{Proceedings of ACM MobiCom}}.
  \bibinfo{pages}{749--761}.
\newblock


\bibitem[Garcia-Saavedra et~al\mbox{.}(2018)]%
        {bib:fluidran}
\bibfield{author}{\bibinfo{person}{Andres Garcia-Saavedra},
  \bibinfo{person}{Xavier Costa-Perez}, \bibinfo{person}{Douglas Leith}, {and}
  \bibinfo{person}{George Iosifidis}.} \bibinfo{year}{2018}\natexlab{}.
\newblock \showarticletitle{Fluidran: Optimized VRAN/MEC Orchestration}. In
  \bibinfo{booktitle}{\emph{Proceedings of IEEE INFOCOM}}.
  \bibinfo{pages}{2366--2374}.
\newblock


\bibitem[Garcia-Saavedra and Costa-Pérez(2021)]%
        {bib:andres-magazine}
\bibfield{author}{\bibinfo{person}{Andres Garcia-Saavedra} {and}
  \bibinfo{person}{Xavier Costa-Pérez}.} \bibinfo{year}{2021}\natexlab{}.
\newblock \showarticletitle{O-RAN: Disrupting the Virtualized RAN Ecosystem}.
\newblock \bibinfo{journal}{\emph{IEEE Communications Standards Magazine}}
  \bibinfo{volume}{5}, \bibinfo{number}{4} (\bibinfo{year}{2021}),
  \bibinfo{pages}{96--103}.
\newblock


\bibitem[Georgiadis et~al\mbox{.}(2006)]%
        {bib:tassiulas-monograph}
\bibfield{author}{\bibinfo{person}{Leonidas Georgiadis},
  \bibinfo{person}{Michael~J Neely}, {and} \bibinfo{person}{Leandros
  Tassiulas}.} \bibinfo{year}{2006}\natexlab{}.
\newblock \showarticletitle{Resource Allocation and Cross-Layer Control in
  Wireless Networks}.
\newblock \bibinfo{journal}{\emph{Foundations and Trends in Networking}}
  \bibinfo{volume}{1}, \bibinfo{number}{1} (\bibinfo{year}{2006}).
\newblock


\bibitem[Gupta and Kamble(2021)]%
        {bib:hori-fair-regr-gupta}
\bibfield{author}{\bibinfo{person}{Swati Gupta} {and} \bibinfo{person}{Vijay
  Kamble}.} \bibinfo{year}{2021}\natexlab{}.
\newblock \showarticletitle{Individual Fairness in Hindsight}.
\newblock \bibinfo{journal}{\emph{The Journal of Machine Learning Research}}
  \bibinfo{volume}{22}, \bibinfo{number}{1} (\bibinfo{year}{2021}),
  \bibinfo{pages}{6386--6420}.
\newblock


\bibitem[Halabian(2019)]%
        {halabian2019distributed}
\bibfield{author}{\bibinfo{person}{Hassan Halabian}.}
  \bibinfo{year}{2019}\natexlab{}.
\newblock \showarticletitle{Distributed Resource Allocation Optimization in 5G
  Virtualized networks}.
\newblock \bibinfo{journal}{\emph{IEEE Journal on Selected Areas in
  Communications}} \bibinfo{volume}{37}, \bibinfo{number}{3}
  (\bibinfo{year}{2019}), \bibinfo{pages}{627--642}.
\newblock


\bibitem[Hazan(2016)]%
        {bib:hazan-monograph}
\bibfield{author}{\bibinfo{person}{Elad Hazan}.}
  \bibinfo{year}{2016}\natexlab{}.
\newblock \showarticletitle{Introduction to Online Convex Optimization}.
\newblock \bibinfo{journal}{\emph{Foundations and Trends in Optimization}}
  \bibinfo{volume}{2}, \bibinfo{number}{3-4} (\bibinfo{year}{2016}),
  \bibinfo{pages}{157--325}.
\newblock


\bibitem[Hyndman and Athanasopoulos(2018)]%
        {bib:hyndman2018forecasting}
\bibfield{author}{\bibinfo{person}{Rob~J Hyndman} {and} \bibinfo{person}{George
  Athanasopoulos}.} \bibinfo{year}{2018}\natexlab{}.
\newblock \bibinfo{booktitle}{\emph{Forecasting: Principles and Practice}}.
\newblock \bibinfo{publisher}{OTexts}.
\newblock


\bibitem[{Intel}(2019)]%
        {intel-flexran}
\bibfield{author}{\bibinfo{person}{{Intel}}.} \bibinfo{year}{2019}\natexlab{}.
\newblock \showarticletitle{FlexRAN LTE and 5G NR FEC Software Development Kit
  Modules}.
\newblock  (\bibinfo{year}{2019}).
\newblock


\bibitem[{Intel}(2021)]%
        {bib:intel-2021}
\bibfield{author}{\bibinfo{person}{{Intel}}.} \bibinfo{year}{2021}\natexlab{}.
\newblock \showarticletitle{Virtual RAN (vRAN) with Hardware Acceleration}. In
  \bibinfo{booktitle}{\emph{{White Paper}}}.
\newblock


\bibitem[Jalota and Ye(2022)]%
        {bib:slot-fairness-jalota}
\bibfield{author}{\bibinfo{person}{Devansh Jalota} {and} \bibinfo{person}{Yinyu
  Ye}.} \bibinfo{year}{2022}\natexlab{}.
\newblock \showarticletitle{Online Learning in Fisher Markets with Unknown
  Agent Preferences}.
\newblock \bibinfo{journal}{\emph{arXiv preprint arXiv:2205.00825}}
  (\bibinfo{year}{2022}).
\newblock


\bibitem[Kalntis and Iosifidis(2022)]%
        {bib:mike-globecom}
\bibfield{author}{\bibinfo{person}{Michail Kalntis} {and}
  \bibinfo{person}{George Iosifidis}.} \bibinfo{year}{2022}\natexlab{}.
\newblock \showarticletitle{Energy-Aware Scheduling of Virtualized Base
  Stations in O-RAN with Online Learning}. In
  \bibinfo{booktitle}{\emph{Proceedings of IEEE GLOBECOM}}.
  \bibinfo{pages}{6048--6054}.
\newblock


\bibitem[Kelly et~al\mbox{.}(1998)]%
        {bib:kelly98}
\bibfield{author}{\bibinfo{person}{Frank Kelly}, \bibinfo{person}{Aman~K
  Maulloo}, {and} \bibinfo{person}{David Kim~Hong Tan}.}
  \bibinfo{year}{1998}\natexlab{}.
\newblock \showarticletitle{Rate Control for Communication Networks: Shadow
  Prices, Proportional Fairness, and Stability}.
\newblock \bibinfo{journal}{\emph{J. Oper. Res. Soc.}} \bibinfo{volume}{49},
  \bibinfo{number}{3} (\bibinfo{year}{1998}), \bibinfo{pages}{237--252}.
\newblock


\bibitem[Ko et~al\mbox{.}(2023)]%
        {bib:edge-ric-sigcomm23}
\bibfield{author}{\bibinfo{person}{Woo-Hyun Ko}, \bibinfo{person}{Ushasi
  Ghosh}, \bibinfo{person}{Ujwal Dinesha}, \bibinfo{person}{Raini Wu},
  \bibinfo{person}{Srinivas Shakkottai}, {and} \bibinfo{person}{Dinesh
  Bharadia}.} \bibinfo{year}{2023}\natexlab{}.
\newblock \showarticletitle{Demo: EdgeRIC: Delivering Realtime RAN
  Intelligence}. In \bibinfo{booktitle}{\emph{Proceedings of ACM SIGCOMM}}.
  \bibinfo{pages}{1162--1164}.
\newblock


\bibitem[Li et~al\mbox{.}(2022)]%
        {bib:gpu-tiwari2022}
\bibfield{author}{\bibinfo{person}{Baolin Li}, \bibinfo{person}{Tirthak Patel},
  \bibinfo{person}{Siddharth Samsi}, \bibinfo{person}{Vijay Gadepally}, {and}
  \bibinfo{person}{Devesh Tiwari}.} \bibinfo{year}{2022}\natexlab{}.
\newblock \showarticletitle{MISO: Exploiting Multi-instance GPU Capability on
  Multi-tenant GPU Clusters}.
\newblock \bibinfo{journal}{\emph{Proceedings of ACM SoCC}}
  (\bibinfo{year}{2022}), \bibinfo{pages}{173--189}.
\newblock


\bibitem[Liao et~al\mbox{.}(2022)]%
        {bib:hori-fair-regr-liao}
\bibfield{author}{\bibinfo{person}{Luofeng Liao}, \bibinfo{person}{Yuan Gao},
  {and} \bibinfo{person}{Christian Kroer}.} \bibinfo{year}{2022}\natexlab{}.
\newblock \showarticletitle{Nonstationary Dual Averaging and Online Fair
  Allocation}.
\newblock \bibinfo{journal}{\emph{Proc. of NeurIPS}},
  \bibinfo{pages}{37159--37172}.
\newblock


\bibitem[Limited(2021)]%
        {bib:energy-china-mobile}
\bibfield{author}{\bibinfo{person}{China~Mobile Limited}.}
  \bibinfo{year}{2021}\natexlab{}.
\newblock \showarticletitle{2021 Sustainability Report.}
\newblock \bibinfo{journal}{\emph{White Paper}} (\bibinfo{year}{2021}).
\newblock


\bibitem[Mbakoyiannis et~al\mbox{.}(2018)]%
        {mbakoyiannis2018energy}
\bibfield{author}{\bibinfo{person}{Dimitrios Mbakoyiannis},
  \bibinfo{person}{Othon Tomoutzoglou}, {and} \bibinfo{person}{George
  Kornaros}.} \bibinfo{year}{2018}\natexlab{}.
\newblock \showarticletitle{Energy-performance considerations for data
  offloading to FPGA-based accelerators over PCIe}.
\newblock \bibinfo{journal}{\emph{ACM Transactions on Architecture and Code
  Optimization (TACO)}} \bibinfo{volume}{15}, \bibinfo{number}{1}
  (\bibinfo{year}{2018}), \bibinfo{pages}{1--24}.
\newblock


\bibitem[McMahan(2017)]%
        {bib:mcmahan2017survey}
\bibfield{author}{\bibinfo{person}{H~Brendan McMahan}.}
  \bibinfo{year}{2017}\natexlab{}.
\newblock \showarticletitle{A Survey of Algorithms and Analysis for Adaptive
  Online Learning}.
\newblock \bibinfo{journal}{\emph{The Journal of Machine Learning Research}}
  \bibinfo{volume}{18}, \bibinfo{number}{1} (\bibinfo{year}{2017}),
  \bibinfo{pages}{3117--3166}.
\newblock


\bibitem[Mehmeti and Kellerer(2022)]%
        {bib:kelleler}
\bibfield{author}{\bibinfo{person}{Fidan Mehmeti} {and}
  \bibinfo{person}{Wolfgang Kellerer}.} \bibinfo{year}{2022}\natexlab{}.
\newblock \showarticletitle{Max-min Fair Resource Allocation in SD-RAN}. In
  \bibinfo{booktitle}{\emph{Proceedings of ACM Q2SWinet}}.
  \bibinfo{pages}{27--35}.
\newblock


\bibitem[Mehmeti and La~Porta(2022)]%
        {bib:laporta}
\bibfield{author}{\bibinfo{person}{Fidan Mehmeti} {and} \bibinfo{person}{Thomas
  La~Porta}.} \bibinfo{year}{2022}\natexlab{}.
\newblock \showarticletitle{Reducing the Cost of Consistency: Performance
  Improvements in Next Generation Cellular Networks With Optimal Resource
  Reallocation}.
\newblock \bibinfo{journal}{\emph{IEEE Transactions on Mobile Computing}}
  \bibinfo{volume}{21}, \bibinfo{number}{7} (\bibinfo{year}{2022}),
  \bibinfo{pages}{2546--2565}.
\newblock


\bibitem[Mhaisen et~al\mbox{.}(2022a)]%
        {bib:naram-caching}
\bibfield{author}{\bibinfo{person}{Naram Mhaisen}, \bibinfo{person}{George
  Iosifidis}, {and} \bibinfo{person}{Douglas Leith}.}
  \bibinfo{year}{2022}\natexlab{a}.
\newblock \showarticletitle{Online Caching with Optimistic Learning}. In
  \bibinfo{booktitle}{\emph{Proceedings of IFIP Networking}}.
\newblock


\bibitem[Mhaisen et~al\mbox{.}(2022b)]%
        {bib:naram_sigm}
\bibfield{author}{\bibinfo{person}{Naram Mhaisen}, \bibinfo{person}{Abhishek
  Sinha}, \bibinfo{person}{Georgios Paschos}, {and} \bibinfo{person}{George
  Iosifidis}.} \bibinfo{year}{2022}\natexlab{b}.
\newblock \showarticletitle{Optimistic No-regret Algorithms for Discrete
  Caching}.
\newblock \bibinfo{journal}{\emph{Proceedings of the ACM on Measurement and
  Analysis of Computing Systems}} \bibinfo{volume}{6}, \bibinfo{number}{3}
  (\bibinfo{year}{2022}), \bibinfo{pages}{1--28}.
\newblock


\bibitem[Mo and Walrand(2000)]%
        {bib:walrand-fairness}
\bibfield{author}{\bibinfo{person}{Jeonghoon Mo} {and} \bibinfo{person}{Jean
  Walrand}.} \bibinfo{year}{2000}\natexlab{}.
\newblock \showarticletitle{Fair End-to-end Window-based Congestion Control}.
\newblock \bibinfo{journal}{\emph{IEEE/ACM Transactions on Networking}}
  \bibinfo{volume}{8}, \bibinfo{number}{5} (\bibinfo{year}{2000}),
  \bibinfo{pages}{556--567}.
\newblock


\bibitem[Modina et~al\mbox{.}(2022)]%
        {modina2022multi}
\bibfield{author}{\bibinfo{person}{Naresh Modina}, \bibinfo{person}{Mandar
  Datar}, \bibinfo{person}{Rachid El-Azouzi}, {and} \bibinfo{person}{Francesco
  de Pellegrini}.} \bibinfo{year}{2022}\natexlab{}.
\newblock \showarticletitle{Multi Resource Allocation for Network Slices with
  Multi-Level fairness}. In \bibinfo{booktitle}{\emph{ICC 2022-IEEE
  International Conference on Communications}}. IEEE,
  \bibinfo{pages}{4872--4877}.
\newblock


\bibitem[Moharir et~al\mbox{.}(2015)]%
        {bib:loadbalanc-shakkottai15}
\bibfield{author}{\bibinfo{person}{Sharayu Moharir}, \bibinfo{person}{Sujay
  Sanghavi}, {and} \bibinfo{person}{Sanjay Shakkottai}.}
  \bibinfo{year}{2015}\natexlab{}.
\newblock \showarticletitle{Online Load Balancing Under Graph Constraints}.
\newblock \bibinfo{journal}{\emph{IEEE/ACM Transactions on Networking}}
  \bibinfo{volume}{24}, \bibinfo{number}{3} (\bibinfo{year}{2015}),
  \bibinfo{pages}{1690--1703}.
\newblock


\bibitem[Mohri and Yang(2016)]%
        {bib:mohri2016accelerating}
\bibfield{author}{\bibinfo{person}{Mehryar Mohri} {and} \bibinfo{person}{Scott
  Yang}.} \bibinfo{year}{2016}\natexlab{}.
\newblock \showarticletitle{Accelerating Online Convex Optimization via
  Adaptive Prediction}. In \bibinfo{booktitle}{\emph{Proceedings of AISTATS}}.
  \bibinfo{pages}{848--856}.
\newblock


\bibitem[Mondal and Ruffini(2023)]%
        {bib:fair-vran-ruffini}
\bibfield{author}{\bibinfo{person}{Sourav Mondal} {and} \bibinfo{person}{Marco
  Ruffini}.} \bibinfo{year}{2023}\natexlab{}.
\newblock \showarticletitle{Fairness Guaranteed and Auction-Based x-Haul and
  Cloud Resource Allocation in Multi-Tenant O-RANs}.
\newblock \bibinfo{journal}{\emph{IEEE Transactions on Communications}}
  \bibinfo{volume}{71}, \bibinfo{number}{6} (\bibinfo{year}{2023}),
  \bibinfo{pages}{3452--3468}.
\newblock


\bibitem[Nace and Pioro(2008)]%
        {bib:fairness-nace}
\bibfield{author}{\bibinfo{person}{Dritan Nace} {and} \bibinfo{person}{Michal
  Pioro}.} \bibinfo{year}{2008}\natexlab{}.
\newblock \showarticletitle{Max-min Fairness and its Applications to Routing
  and Load-balancing in Communication Networks: a Tutorial}.
\newblock \bibinfo{journal}{\emph{IEEE Communications Surveys \& Tutorials}}
  \bibinfo{volume}{10}, \bibinfo{number}{4} (\bibinfo{year}{2008}),
  \bibinfo{pages}{5--17}.
\newblock


\bibitem[Neely(2010)]%
        {bib:neely-monograph}
\bibfield{author}{\bibinfo{person}{Michael~J Neely}.}
  \bibinfo{year}{2010}\natexlab{}.
\newblock \showarticletitle{Stochastic Network Optimization with Application to
  Communication and Queueing Systems}.
\newblock \bibinfo{journal}{\emph{Synthesis Lectures on Communication
  Networks}} (\bibinfo{year}{2010}).
\newblock


\bibitem[Neely(2013)]%
        {bib:delay-neely}
\bibfield{author}{\bibinfo{person}{Michael~J Neely}.}
  \bibinfo{year}{2013}\natexlab{}.
\newblock \showarticletitle{Delay-Based Network Utility Maximization}.
\newblock \bibinfo{journal}{\emph{IEEE/ACM Transactions on Networking}}
  \bibinfo{volume}{21}, \bibinfo{number}{1} (\bibinfo{year}{2013}),
  \bibinfo{pages}{41--54}.
\newblock


\bibitem[Nguyen et~al\mbox{.}(2019)]%
        {nguyen2019market}
\bibfield{author}{\bibinfo{person}{Duong~Tung Nguyen},
  \bibinfo{person}{Long~Bao Le}, {and} \bibinfo{person}{Vijay~K Bhargava}.}
  \bibinfo{year}{2019}\natexlab{}.
\newblock \showarticletitle{A Market-based Framework for Multi-resource
  Allocation in Fog Computing}.
\newblock \bibinfo{journal}{\emph{IEEE/ACM Transactions on Networking}}
  \bibinfo{volume}{27}, \bibinfo{number}{3} (\bibinfo{year}{2019}),
  \bibinfo{pages}{1151--1164}.
\newblock


\bibitem[Orabona(2019)]%
        {bib:orabona-tutorial}
\bibfield{author}{\bibinfo{person}{Francesco Orabona}.}
  \bibinfo{year}{2019}\natexlab{}.
\newblock \showarticletitle{A Modern Introduction to Online Learning}.
\newblock \bibinfo{journal}{\emph{arXiv preprint arXiv:1912.13213}}
  (\bibinfo{year}{2019}).
\newblock


\bibitem[Polese et~al\mbox{.}(2023)]%
        {bib:melodia-tutorial-23}
\bibfield{author}{\bibinfo{person}{Michele Polese}, \bibinfo{person}{Leonardo
  Bonati}, \bibinfo{person}{Salvatore D'Oro}, \bibinfo{person}{Stefano
  Basagni}, {and} \bibinfo{person}{Tommaso Melodia}.}
  \bibinfo{year}{2023}\natexlab{}.
\newblock \showarticletitle{Understanding O-RAN: Architecture, Interfaces,
  Algorithms, Security, and Research Challenges}.
\newblock \bibinfo{journal}{\emph{IEEE Commun. Surv. Tutorials}}
  \bibinfo{volume}{25}, \bibinfo{number}{2} (\bibinfo{year}{2023}),
  \bibinfo{pages}{1376--1411}.
\newblock


\bibitem[Raca et~al\mbox{.}(2020)]%
        {cormac-NN-ComMag20}
\bibfield{author}{\bibinfo{person}{Darijo Raca}, \bibinfo{person}{Ahmed~H
  Zahran}, \bibinfo{person}{Cormac~J Sreenan}, \bibinfo{person}{Rakesh~K
  Sinha}, \bibinfo{person}{Emir Halepovic}, \bibinfo{person}{Rittwik Jana},
  {and} \bibinfo{person}{Vijay Gopalakrishnan}.}
  \bibinfo{year}{2020}\natexlab{}.
\newblock \showarticletitle{On Leveraging Machine and Deep Learning for
  Throughput Prediction in Cellular Networks: Design, Performance, and
  Challenges}.
\newblock \bibinfo{journal}{\emph{IEEE Communications Magazine}}
  \bibinfo{volume}{58}, \bibinfo{number}{3} (\bibinfo{year}{2020}),
  \bibinfo{pages}{11--17}.
\newblock


\bibitem[Radunovic and Le-Boudec(2007)]%
        {bib:fairness-leboudec}
\bibfield{author}{\bibinfo{person}{Bozidar Radunovic} {and}
  \bibinfo{person}{Jean-Yves Le-Boudec}.} \bibinfo{year}{2007}\natexlab{}.
\newblock \showarticletitle{A Unified Framework for Max-Min and Min-Max
  Fairness With Applications}.
\newblock \bibinfo{journal}{\emph{IEEE/ACM Transactions on Networking}}
  \bibinfo{volume}{15}, \bibinfo{number}{5} (\bibinfo{year}{2007}),
  \bibinfo{pages}{1073--1083}.
\newblock


\bibitem[Rakhlin and Sridharan(2013)]%
        {bib:rakhlin-nips13}
\bibfield{author}{\bibinfo{person}{Alexander Rakhlin} {and}
  \bibinfo{person}{Karthik Sridharan}.} \bibinfo{year}{2013}\natexlab{}.
\newblock \showarticletitle{Optimization, Learning, and Games with Predictable
  Sequences}. In \bibinfo{booktitle}{\emph{Proceedings of NIPS}}.
  \bibinfo{pages}{848--856}.
\newblock


\bibitem[Rost et~al\mbox{.}(2015)]%
        {rost-globecom15}
\bibfield{author}{\bibinfo{person}{Peter Rost}, \bibinfo{person}{Andreas
  Maeder}, \bibinfo{person}{Matthew~C Valenti}, {and}
  \bibinfo{person}{Salvatore Talarico}.} \bibinfo{year}{2015}\natexlab{}.
\newblock \showarticletitle{Computationally-aware Sum-rate Optimal Scheduling
  for Centralized Radio Access Networks}. In \bibinfo{booktitle}{\emph{Proc. of
  IEEE GLOBECOM}}. \bibinfo{pages}{1--6}.
\newblock


\bibitem[Salvat et~al\mbox{.}(2023)]%
        {experim_platform}
\bibfield{author}{\bibinfo{person}{J~Xavier Salvat}, \bibinfo{person}{Jose~A
  Ayala-Romero}, \bibinfo{person}{Lanfranco Zanzi}, \bibinfo{person}{Andres
  Garcia-Saavedra}, {and} \bibinfo{person}{Xavier Costa-Perez}.}
  \bibinfo{year}{2023}\natexlab{}.
\newblock \showarticletitle{Open Radio Access Networks (O-RAN) Experimentation
  Platform: Design and Datasets}.
\newblock \bibinfo{journal}{\emph{IEEE Communications Magazine}}
  (\bibinfo{year}{2023}).
\newblock


\bibitem[Shalev-Shwartz(2012)]%
        {bib:shai-monograph}
\bibfield{author}{\bibinfo{person}{Shai Shalev-Shwartz}.}
  \bibinfo{year}{2012}\natexlab{}.
\newblock \showarticletitle{Online Learning and Online Convex Optimization}.
\newblock \bibinfo{journal}{\emph{Foundations and Trends in Machine Learning}}
  \bibinfo{volume}{4}, \bibinfo{number}{2} (\bibinfo{year}{2012}),
  \bibinfo{pages}{107--194}.
\newblock


\bibitem[Shalev-Shwartz and Singer(2007)]%
        {bib:ftrl-shalev}
\bibfield{author}{\bibinfo{person}{Shai Shalev-Shwartz} {and}
  \bibinfo{person}{Yoram Singer}.} \bibinfo{year}{2007}\natexlab{}.
\newblock \showarticletitle{A Primal-dual Perspective of Online Learning
  Algorithms}.
\newblock \bibinfo{journal}{\emph{Machine Learning}} \bibinfo{volume}{69},
  \bibinfo{number}{2-3} (\bibinfo{year}{2007}), \bibinfo{pages}{115--142}.
\newblock


\bibitem[Si~Salem et~al\mbox{.}(2022)]%
        {bib:tareq_fairness}
\bibfield{author}{\bibinfo{person}{Tareq Si~Salem}, \bibinfo{person}{Georgios
  Iosifidis}, {and} \bibinfo{person}{Giovanni Neglia}.}
  \bibinfo{year}{2022}\natexlab{}.
\newblock \showarticletitle{Enabling Long-term Fairness in Dynamic Resource
  Allocation}.
\newblock \bibinfo{journal}{\emph{Proceedings of the ACM on Measurement and
  Analysis of Computing Systems}} \bibinfo{volume}{6}, \bibinfo{number}{3}
  (\bibinfo{year}{2022}), \bibinfo{pages}{1--36}.
\newblock


\bibitem[Sinclair et~al\mbox{.}(2022)]%
        {bib:slot-fairness-sinclair}
\bibfield{author}{\bibinfo{person}{Sean~R Sinclair},
  \bibinfo{person}{Siddhartha Banerjee}, {and} \bibinfo{person}{Christina
  Lee~Yu}.} \bibinfo{year}{2022}\natexlab{}.
\newblock \showarticletitle{Sequential Fair Allocation: Achieving the Optimal
  Envy-Efficiency Tradeoff Curve}.
\newblock \bibinfo{journal}{\emph{Operations Research}} \bibinfo{volume}{71},
  \bibinfo{number}{5} (\bibinfo{year}{2022}).
\newblock


\bibitem[Talebi and Proutiere(2018)]%
        {bib:slot-fairness-sadegh}
\bibfield{author}{\bibinfo{person}{Mohammad~Sadegh Talebi} {and}
  \bibinfo{person}{Alexandre Proutiere}.} \bibinfo{year}{2018}\natexlab{}.
\newblock \showarticletitle{Learning Proportionally Fair Allocations with Low
  Regret}.
\newblock \bibinfo{journal}{\emph{Proceedings of the ACM on Measurement and
  Analysis of Computing Systems}} \bibinfo{volume}{2}, \bibinfo{number}{2}
  (\bibinfo{year}{2018}), \bibinfo{pages}{1--31}.
\newblock


\bibitem[Telefonica(2022)]%
        {bib:energy-telefonica}
\bibfield{author}{\bibinfo{person}{Telefonica}.}
  \bibinfo{year}{2022}\natexlab{}.
\newblock \bibinfo{title}{Open RAN Technical Priorities Release 2}.
\newblock
  \bibinfo{howpublished}{\url{https://www.telefonica.com/en/communication-room/reports/open-ran-technical-priorities-release-2/}}.
\newblock


\bibitem[Tripathi et~al\mbox{.}(2023)]%
        {bib:fair-vran-andres}
\bibfield{author}{\bibinfo{person}{Sharda Tripathi}, \bibinfo{person}{Corrado
  Puligheddu}, \bibinfo{person}{Somreeta Pramanik}, \bibinfo{person}{Andres
  Garcia-Saavedra}, {and} \bibinfo{person}{Carla~Fabiana Chiasserini}.}
  \bibinfo{year}{2023}\natexlab{}.
\newblock \showarticletitle{Fair and Scalable Orchestration of Network and
  Compute Resources for Virtual Edge Services}.
\newblock \bibinfo{journal}{\emph{IEEE Transactions on Mobile Computing}}
  \bibinfo{volume}{early access} (\bibinfo{year}{2023}),
  \bibinfo{pages}{1--17}.
\newblock


\bibitem[Wang et~al\mbox{.}(2021)]%
        {wang-compaware-WirNet19}
\bibfield{author}{\bibinfo{person}{Ke Wang}, \bibinfo{person}{XiaoYi Yu},
  \bibinfo{person}{WenLiang Lin}, \bibinfo{person}{ZhongLiang Deng}, {and}
  \bibinfo{person}{Xin Liu}.} \bibinfo{year}{2021}\natexlab{}.
\newblock \showarticletitle{Computing-aware Scheduling in Mobile Edge Computing
  System}.
\newblock \bibinfo{journal}{\emph{Wireless Networks}}  \bibinfo{volume}{27}
  (\bibinfo{year}{2021}), \bibinfo{pages}{4229--4245}.
\newblock


\bibitem[Wang et~al\mbox{.}(2014)]%
        {bib:wang-dominant-fairness}
\bibfield{author}{\bibinfo{person}{Wei Wang}, \bibinfo{person}{Baochun Li},
  {and} \bibinfo{person}{Ben Liang}.} \bibinfo{year}{2014}\natexlab{}.
\newblock \showarticletitle{Dominant Resource Fairness in Cloud Computing
  Systems with Heterogeneous Servers}. In \bibinfo{booktitle}{\emph{Proceedings
  of IEEE INFOCOM}}. \bibinfo{pages}{583--591}.
\newblock


\bibitem[Weng et~al\mbox{.}(2020)]%
        {bib:loadbalanc-srikant21}
\bibfield{author}{\bibinfo{person}{Wentao Weng}, \bibinfo{person}{Xingyu Zhou},
  {and} \bibinfo{person}{Rayadurgam Srikant}.} \bibinfo{year}{2020}\natexlab{}.
\newblock \showarticletitle{Optimal Load Balancing with Locality Constraints}.
\newblock \bibinfo{journal}{\emph{Proceedings of the ACM on Measurement and
  Analysis of Computing Systems}} \bibinfo{volume}{4}, \bibinfo{number}{3}
  (\bibinfo{year}{2020}), \bibinfo{pages}{45:1--45:37}.
\newblock


\bibitem[Xu et~al\mbox{.}(2017)]%
        {xu2017online}
\bibfield{author}{\bibinfo{person}{Jie Xu}, \bibinfo{person}{Lixing Chen},
  {and} \bibinfo{person}{Shaolei Ren}.} \bibinfo{year}{2017}\natexlab{}.
\newblock \showarticletitle{Online Learning for Offloading and Autoscaling in
  Energy Harvesting Mobile Edge Computing}.
\newblock \bibinfo{journal}{\emph{IEEE Transactions on Cognitive Communications
  and Networking}} \bibinfo{volume}{3}, \bibinfo{number}{3}
  (\bibinfo{year}{2017}), \bibinfo{pages}{361--373}.
\newblock


\bibitem[Zhang et~al\mbox{.}(2019)]%
        {patras-DNN-Tutorial2019}
\bibfield{author}{\bibinfo{person}{Chaoyun Zhang}, \bibinfo{person}{Paul
  Patras}, {and} \bibinfo{person}{Hamed Haddadi}.}
  \bibinfo{year}{2019}\natexlab{}.
\newblock \showarticletitle{Deep Learning in Mobile and Wireless Networking: A
  Survey}.
\newblock \bibinfo{journal}{\emph{IEEE Commun. Surv. Tutor.}}
  \bibinfo{volume}{21}, \bibinfo{number}{3} (\bibinfo{year}{2019}),
  \bibinfo{pages}{2224--2287}.
\newblock


\bibitem[Zhao et~al\mbox{.}(2019)]%
        {niyato-RL-TWC2019}
\bibfield{author}{\bibinfo{person}{Nan Zhao}, \bibinfo{person}{Ying-Chang
  Liang}, \bibinfo{person}{Dusit Niyato}, \bibinfo{person}{Yiyang Pei},
  \bibinfo{person}{Minghu Wu}, {and} \bibinfo{person}{Yunhao Jiang}.}
  \bibinfo{year}{2019}\natexlab{}.
\newblock \showarticletitle{Deep Reinforcement Learning for User Association
  and Resource Allocation in Heterogeneous Cellular Networks}.
\newblock \bibinfo{journal}{\emph{IEEE Transactions on Wireless
  Communications}} \bibinfo{volume}{18}, \bibinfo{number}{11}
  (\bibinfo{year}{2019}), \bibinfo{pages}{5141--5152}.
\newblock


\bibitem[Zinkevich(2003)]%
        {bib:zinkevich}
\bibfield{author}{\bibinfo{person}{Martin Zinkevich}.}
  \bibinfo{year}{2003}\natexlab{}.
\newblock \showarticletitle{Online Convex Programming and Generalized
  Infinitesimal Gradient Ascent}. In \bibinfo{booktitle}{\emph{Proceedings of
  ICML}}. \bibinfo{pages}{928--936}.
\newblock


\end{thebibliography}

\section{Appendix} \label{sec:appendix}
This section provides the remaining proofs for the results presented in the previous sections, as well as some additional evaluation results for the interested reader. Please note that we abuse slightly the notation by redefining and reusing some symbols, in order to keep the presentation streamlined. 

\subsection{Proof of Lemma \ref{the:quadratic-regret}}
This lemma applies to the dual update for $\bm \theta_{t+1}$ in \eqref{dual-update-2a} that uses the regularizers \eqref{eq:dual-regularizer}. Applying~\cite[Theorem 2]{bib:mohri2016accelerating}, we can write:
\begin{align}
 \bm{\mathcal{R}}_T^{\theta}-q_{1:T-1}(\bm{\theta}^\star) &\leq \sumT \|\bm{\kappa}_t-\Tilde{\bm{\kappa}}_t\|_{(t-1),*}^2= \sumT \frac{\|\bm{\kappa}_t-\Tilde{\bm{\kappa}}_t\|_{2}^2}{\sigma_{1:t-1}} = \sumT \frac{(1/\sigma) \|\bm{\kappa}_t-\Tilde{\bm{\kappa}}_t\|_{2}^2}{ \sqrt{ \sum_{\tau=1}^{t-1}\|\bm{\kappa}_\tau - \Tilde{\bm{\kappa}}_\tau\|_{2}^2}}. \label{eqlem0-1}
\end{align}
Similarly, from the proof of the same Theorem, we extract the inequality:
\begin{align}
	&\bm{\mathcal{R}}_T^{\theta}-q_{1:T-1}(\bm{\theta}^\star) \leq \sumT \big (\bm{\kappa}_t - \Tilde{\bm{\kappa}}_t \big)^\top \big (\bm{\theta}_t - \bm{\vartheta}_t \big)\leq \sumT \|\bm{\kappa}_t - \Tilde{\bm{\kappa}}_t \|_{(t-1),*} \|\bm{\theta}_t - \bm{\vartheta}_t\|_{(t-1)} \label{eq-lem41proof-1} \\
	&\text{where} \quad \bm{\vartheta}_{t}=\arg\min_{\bm{\theta}\in{\Theta}} \left\{ q_{1:t-1}(\bm{\theta}) + \bm{\theta}^\top \sum_{\tau=1}^{t}\bm{\kappa}_{\tau}	\right\} 
\end{align}
is the prescient action that is selected with knowledge of next-round cost $\bm{\kappa}_{t}$ (instead of using predictions). Recalling the properties of the selected regularizer, we rewrite \eqref{eq-lem41proof-1} as:
\begin{align}
\bm{\mathcal{R}}_T^{\theta}-q_{1:T-1}(\bm{\theta}^\star) &\leq  \sumT \|\bm{\kappa}_t - \Tilde{\bm{\kappa}}_t \|_{2} \|\bm{\theta}_t - \bm{\vartheta}_t\|_{2} \leq \textit{D}_{\Theta}\sumT \|\bm{\kappa}_t - \Tilde{\bm{\kappa}}_t \|_{2} \label{eqlem0-2}
\end{align}
where we used the fact that $\Theta$ has a bounded diameter $D_\Theta$. Combining \eqref{eqlem0-1} and \eqref{eqlem0-2}, we can follow the rationale in \cite[Sec. 7.6]{bib:orabona-tutorial}, and write:
\begin{align*}
	&\bm{\mathcal{R}}_T^{\theta}-q_{1:T-1}(\bm{\theta}^\star)  \leq \min\left\{ \textit{D}_{\Theta}\sumT \|\bm{\kappa}_t - \Tilde{\bm{\kappa}}_t \|_{2}, \sumT \frac{(1/\sigma) \|\bm{\kappa}_t-\Tilde{\bm{\kappa}}_t\|_{2}^2}{ \sqrt{ \sum_{\tau=1}^{t-1}\|\bm{\kappa}_\tau - \Tilde{\bm{\kappa}}_\tau\|_{2}^2}} \right\} \notag \\ 
	&=\sumT\min\left\{ {D}_{\Theta} \|\bm{\kappa}_t - \Tilde{\bm{\kappa}}_t \|_{2}, \frac{(1/\sigma) \|\bm{\kappa}_t-\Tilde{\bm{\kappa}}_t\|_{2}^2}{ \sqrt{ \sum_{\tau=1}^{t-1}\|\bm{\kappa}_\tau - \Tilde{\bm{\kappa}}_\tau\|_{2}^2}} \right\}\\
	&=\sumT\sqrt{\min\left\{ \textit{D}_{\Theta}^2 \|\bm{\kappa}_t - \Tilde{\bm{\kappa}}_t \|_{2}^2, \frac{(1/\sigma)^2 \|\bm{\kappa}_t-\Tilde{\bm{\kappa}}_t\|_{2}^4}{  \sum_{\tau=1}^{t-1}\|\bm{\kappa}_\tau - \Tilde{\bm{\kappa}}_\tau\|_{2}^2} \right\}\!}\stackrel{(\alpha)}\leq\! \sumT\sqrt{\frac{2}{\frac{1}{\textit{D}_{\Theta}^2 \|\bm{\kappa}_t - \Tilde{\bm{\kappa}}_t \|_{2}^2} + \frac{\sum_{\tau=1}^{t-1}\|\bm{\kappa}_\tau - \Tilde{\bm{\kappa}}_\tau\|_{2}^2}{(1/\sigma)^2 \|\bm{\kappa}_t-\Tilde{\bm{\kappa}}_t\|_{2}^4}}\!}\\
	&= \sqrt{2}\sumT \sqrt{ \frac{{(1/\sigma)^2 \textit{D}_{\Theta}^2 \|\bm{\kappa}_t-\Tilde{\bm{\kappa}}_t\|_{2}^4}}{(1/\sigma)^2 \|\bm{\kappa}_t-\Tilde{\bm{\kappa}}_t\|_{2}^2 + \textit{D}_{\Theta}^2 \sum_{\tau=1}^{t-1}\|\bm{\kappa}_\tau - \Tilde{\bm{\kappa}}_\tau\|_{2}^2} } \!=\!\sumT \! \frac{(1/\sigma)\textit{D}_{\Theta} \sqrt{2}\|\bm{\kappa}_t-\Tilde{\bm{\kappa}}_t\|_{2}^2}{\sqrt{(1/\sigma)^2 \|\bm{\kappa}_t-\Tilde{\bm{\kappa}}_t\|_{2}^2 + \textit{D}_{\Theta}^2 \sum_{\tau=1}^{t-1}\|\bm{\kappa}_\tau - \Tilde{\bm{\kappa}}_\tau\|_{2}^2} }\!\!\\
	&\stackrel{(\beta)}\leq\! \sumT\! \frac{(1/\sigma)\textit{D}_{\Theta} \sqrt{2}\|\bm{\kappa}_t-\Tilde{\bm{\kappa}}_t\|_{2}^2}{\sqrt{(1/\sigma)^2 \sum_{\tau=1}^{t}\|\bm{\kappa}_\tau - \Tilde{\bm{\kappa}}_\tau\|_{2}^2} }
	\!=\! \sumT\! \frac{\textit{D}_{\Theta}\sqrt{2}\|\bm{\kappa}_t-\Tilde{\bm{\kappa}}_t\|_{2}^2}{\sqrt{ \sum_{\tau=1}^{t}\|\bm{\kappa}_\tau - \Tilde{\bm{\kappa}}_\tau\|_{2}^2} }\stackrel{(\gamma)}\leq 2 \textit{D}_{\Theta}\sqrt{2}  {\sqrt{ \sum_{t=1}^{T}\|\bm{\kappa}_t - \Tilde{\bm{\kappa}}_t\|_{2}^2} }
\end{align*}
where $(\alpha)$ uses that the minimum between two numbers is less than their harmonic mean; $(\beta)$ assumes that $(1/\sigma) \leq \textit{D}_{\Theta}$, which is satisfied by the proposed value for $\sigma$ (see below); and $(\gamma)$ applies an identify from \cite[Lemma 3.5]{bib:auer-adaptive2002}. To conclude, it suffices to observe that $q_{1:T-1}(\bm \theta^\star)$ can be upper bounded due to boundedness of $\Theta$ as follows:
\begin{align*}
q_{1:T-1}(\bm \theta^\star)\leq\sigma\textit{D}_{\Theta}^2\sqrt{\sum_{t=1}^{T}\|\bm{\kappa}_t - \Tilde{\bm{\kappa}}_t\|_{2}^2},
\end{align*}
and the value of parameter $\sigma$ that minimizes the constant factor above is $\sigma=2\sqrt{2}/D_\Theta$.

\subsection{Proof of Lemma \ref{the:entropic-regret}}

We start by characterizing the strong convexity of the entropic regularizer $r_{1:t}$ that we use in the primal update. We note that this is not the typical entropic regularizer used in FTRL (or Mirror Descent) algorithms, cf. \cite{bib:mcmahan2017survey}. Here, the regularizing parameter does not have a constant term (this allows us to get $\mathcal O(1)$ for perfect predictions), and the constraint set is a set of simplices, i.e., a multi-simplex, instead of a single simplex binding all variables.

\begin{lemma}\label{lem:reg-conv-entropic}
Consider the convex set $\mathcal{X} = \{x_{ij} \geq 0 \ : \ \sum_{j\in\mathcal{J}} x_{ij} = 1, \ \ \forall i \in \mathcal{I}\}$, and the nonnegative convex function $r_{1:t} : \mathcal{X} \mapsto \mathbb{R}_+$ defined in \eqref{eq:primal-regularizer} as:
\[
r_{1:t}(\bm x)=\frac{\eta_{1:t}}{2}\left( I\log J + \sum_{i\in \c I}\sum_{j\in \c J} x_{ij}\log x_{ij}\right), \quad \text{where} \quad \eta_{1:t}=\eta\sqrt{\sum_{\tau=1}^{t}\|\bm g_\tau+\bm w_\tau - \bm{\p g}_\tau- \bm{\p w}_\tau\|_{\infty}^2 }
\]
Then, function $r_{1:t}(\bm x)$ is 1-strongly convex with respect to the norm $\|\bm{x}\|_{(t)} = \|\bm{x}\|_1 \sqrt{\frac{\eta_{1:t}}{I}}$. 
\end{lemma}
\begin{proof}
Let us define $\bm{x}_i=(x_{ij}, j\in \mathcal{J})$ and $\bm y_i=(y_{ij}, j\in \mathcal{J})$, and the unit simplex $\mathcal{X}_i = \{x_{ij}\!\geq\!0 :  \sum_{j\in\mathcal{J}} x_{ij} = 1\}$. Then, from the standard analysis of the entropic regularizer it holds that the (simpler) reguarlizer defined as
\begin{align}
	\hat{r}_i(\bm x_i)=\log{J} + \sum_{j\in\mathcal{J}} x_{ij}\log x_{ij}
\end{align}
is $1$-strongly convex w.r.t. the $\ell_1$ norm over $\mathcal{X}_i$, and therefore it holds: 
\begin{align}
	\hat{r}_i(\bm y_i)\geq \hat{r}_i(\bm x_i)+\nabla \hat{r}_i(\bm x_i)^\top(\bm y_i-\bm x_i)+\frac{1}{2}\|\bm y_i-\bm x_i\|_1^2, \quad \forall i \in \mathcal{I},
\end{align}
Hence, we can write:
\begin{align}
\frac{r_{1:t}(\bm y)}{\eta_{1:t}/2}	= \sum_{i\in\mathcal{I}} \hat{r}_i(\bm y_i) &\geq \sum_{i\in\mathcal{I}}\hat{r}_i(\bm x_i) +  \frac{\nabla r_{1:t}(\bm x)}{\eta_{1:t}/2}^\top(\bm y-\bm x)+\sum_{i\in\mathcal{I}} \frac{1}{2}\|\bm y_i-\bm x_i\|_1^2 \notag \\
&\stackrel{(\alpha)}\geq \sum_{i\in\mathcal{I}} \hat{r}_i(\bm x_i) +\frac{\nabla r_{1:t}(\bm x)}{\eta_{1:t}/2}^\top(\bm y-\bm x)+ \frac{1}{2{I}}\|\bm y-\bm x\|_1^2 \notag \\
&=\frac{r_{1:t}(\bm x)}{\eta_{1:t}/2} + \frac{\nabla r_{1:t}(\bm x)}{\eta_{1:t}/2}^\top(\bm y-\bm x)+ \frac{1}{2\textit{I}}\|\bm y-\bm x\|_1^2 \Rightarrow \\
{r_{1:t}(\bm y)}&\geq {r_{1:t}(\bm x)} + {\nabla r_{1:t}(\bm x)}^\top(\bm y-\bm x)+ \frac{\eta_{1:t}}{I}\|\bm y-\bm x\|_1^2
\end{align}
where in $(\alpha)$ we used the inequality $(a_1+a_2+\ldots+ a_n)^2\leq n(a_1^2 + a_2^2+\ldots +a_n^2)$. 
\end{proof}
With this result at hand, we can proceed to prove Lemma \ref{the:entropic-regret} following a similar approach as in the proof of Lemma \ref{the:quadratic-regret}.  From~\cite[Theorem 2]{bib:mohri2016accelerating}, we can write:
\begin{align}
\bm{\mathcal{R}}_T-r_{1:T-1}(\bm{x}^\star) &\leq \sumT \|\bm{g}_t+\bm{w}_t-\Tilde{\bm{g}}_t- \bm{\p w}_t\|_{(t-1),*}^2= \sumT \frac{I}{\eta_{1:t-1}}\|\bm{g}_t+\bm{w}_t-\Tilde{\bm{g}}_t- \bm{\p w}_t\|_{\infty}^2 \notag \\
&= \sumT \frac{I \|\bm{g}_t+\bm{w}_t-\Tilde{\bm{g}}_t- \bm{\p w}_t\|_{\infty}^2}{\eta \sqrt{ \sum_{k=1}^{t-1}\|\bm{g}_k+\bm{w}_k - \Tilde{\bm{g}}_k- \bm{\p w}_k\|_{\infty}^2}}. \label{eqlem1-1}
\end{align}
Similarly, from the proof of the same Theorem, we can write:
\begin{align*}
&\bm{\mathcal{R}}_T-r_{1:T-1}(\bm{x}^\star)\! \leq\! \sumT \big (\bm{g}_t+\bm{w}_t \!- \Tilde{\bm{g}}_t - \bm{\p w}_t\big)^\top \big (\bm{x}_t - \bm{\chi}_t \big)\!\leq \!\sumT \|\bm{g}_t+\bm{w}_t - \Tilde{\bm{g}}_t \!- \bm{\p w}_t\|_{(t-1),*} \|\bm{x}_t - \bm{\chi}_t\|_{(t-1)} \\
&\text{where} \ \quad \bm{\chi}_{t+1}=\arg\min_{\bm{x}\in\mathcal{X}} \left\{ r_{1:t}(\bm{x}) + \bm{x}^\top \big(\bm{g}_{1:t+1}+\bm{w}_{1:t+1}\big)	\right\}
\end{align*}
is the prescient action that is selected with knowledge of next-round cost $\bm{g}_{t+1}$. Recalling the properties of the entropic regularizer, we have:
\begin{align}
\bm{\mathcal{R}}_T\!-r_{1:T-1}(\bm{x}^\star)&\leq \! \sumT\! \|\bm{g}_t+\bm{w}_t \!- \Tilde{\bm{g}}_t \!-\! \bm{\p w}_t \|_{(t-1),*} \|\bm{x}_t \!- \bm{\chi}_t\|_{(t-1)} \!=\! \sumT \!\|\bm{g}_t+\bm{w}_t \!- \Tilde{\bm{g}}_t \!- \bm{\p w}_t\|_{\infty} \|\bm{x}_t \!-\! \bm{\chi}_t\|_{1} \!\Rightarrow	\notag  \\
\bm{\mathcal{R}}_T\!-r_{1:T-1}(\bm{x}^\star)&\leq \sumT \|\bm{g}_t+\bm{w}_t - \Tilde{\bm{g}}_t- \bm{\p w}_t \|_{\infty}\big(\|\bm{x}_t\|_1 +\|\bm{\chi}_t\|_1 \big)\!\leq \!2\textit{I}\sumT \|\bm{g}_t+\bm{w}_t \!-\! \Tilde{\bm{g}}_t \!-\! \bm{\p w}_t \|_{\infty}, \label{eqlem1-2}
\end{align}
where in the last step we used the fact that $\bm{x}_t, \bm{\chi}_t \in \mathcal{X}$, and they have non-negative elements. Combining \eqref{eqlem1-1} and \eqref{eqlem1-2}, we can follow the rationale in \cite[Sec. 7.6]{bib:orabona-tutorial}, and write:
\begin{align*}
&\bm{\mathcal{R}}_T-r_{1:T-1}(\bm x^\star) \leq \min\left\{ 2\textit{I}\sumT\|\bm{g}_t+\bm{w}_t - \Tilde{\bm{g}}_t- \bm{\p w}_t \|_{\infty}, \sumT \frac{I \|\bm{g}_t+\bm{w}_t-\Tilde{\bm{g}}_t- \bm{\p w}_t\|_{\infty}^2}{\eta \sqrt{ \sum_{k=1}^{t-1}\|\bm{g}_k+\bm{w}_k - \Tilde{\bm{g}}_k- \bm{\p w}_k\|_{\infty}^2}} \right\}\\
&= \sumT\min\left\{ 2I\|\bm{g}_t+\bm{w}_t - \Tilde{\bm{g}}_t- \bm{\p w}_t \|_{\infty},  \frac{I \|\bm{g}_t+\bm{w}_t-\Tilde{\bm{g}}_t- \bm{\p w}_t\|_{\infty}^2}{\eta \sqrt{ \sum_{k=1}^{t-1}\|\bm{g}_k+\bm{w}_k - \Tilde{\bm{g}}_k- \bm{\p w}_k\|_{\infty}^2}} \right\} \\
&= 2I\sumT\min\left\{ \|\bm{g}_t+\bm{w}_t - \Tilde{\bm{g}}_t - \bm{\p w}_t\|_{\infty}, \frac{\|\bm{g}_t+\bm{w}_t-\Tilde{\bm{g}}_t- \bm{\p w}_t\|_{\infty}^2}{ 2\eta \sqrt{ \sum_{k=1}^{t-1}\|\bm{g}_k+\bm{w}_k - \Tilde{\bm{g}}_k- \bm{\p w}_k\|_{\infty}^2}} \right\}\\
&= 2I\sumT\sqrt{\min\left\{ \|\bm{g}_t+\bm{w}_t - \Tilde{\bm{g}}_t- \bm{\p w}_t \|_{\infty}^2, \frac{ \|\bm{g}_t+\bm{w}_t-\Tilde{\bm{g}}_t- \bm{\p w}_t\|_{\infty}^4}{ 4\eta^2 \sum_{k=1}^{t-1}\|\bm{g}_k+\bm{w}_k - \Tilde{\bm{g}}_k- \bm{\p w}_k\|_{\infty}^2} \right\}\!}\\
&\stackrel{(\alpha)}\leq\! 2I\sumT\sqrt{\frac{2}{\frac{1}{\|\bm{g}_t +\bm{w}_t- \Tilde{\bm{g}}_t- \bm{\p w}_t \|_{\infty}^2} + \frac{4 \eta^2\sum_{k=1}^{t-1}\|\bm{g}_k +\bm{w}_k- \Tilde{\bm{g}}_k- \bm{\p w}_k\|_{\infty}^2}{\|\bm{g}_t+\bm{w}_t-\Tilde{\bm{g}}_t- \bm{\p w}_t\|_{\infty}^4}}\!}\\
&= 2\sqrt{2}I\sumT \sqrt{ \frac{{\|\bm{g}_t+\bm{w}_t-\Tilde{\bm{g}}_t- \bm{\p w}_t\|_{\infty}^4}}{\|\bm{g}_t+\bm{w}_t-\Tilde{\bm{g}}_t- \bm{\p w}_t\|_{\infty}^2 + 4 \eta^2\sum_{k=1}^{t-1}\|\bm{g}_k+\bm{w}_k - \Tilde{\bm{g}}_k- \bm{\p w}_k\|_{\infty}^2} }\\
&\stackrel{(\beta)}\leq 2\sqrt{2}I\sumT \sqrt{ \frac{{\|\bm{g}_t+\bm{w}_t-\Tilde{\bm{g}}_t- \bm{\p w}_t\|_{\infty}^4}}{4 \eta^2\sum_{k=1}^{t}\|\bm{g}_k+\bm{w}_k - \Tilde{\bm{g}}_k- \bm{\p w}_k\|_{\infty}^2} } \\
&=\frac{\sqrt{2} I}{\eta}\sumT \! \frac{\|\bm{g}_t+\bm{w}_t-\Tilde{\bm{g}}_t- \bm{\p w}_t\|_{\infty}^2}{\sqrt{\sum_{k=1}^{t}\|\bm{g}_k+\bm{w}_k - \Tilde{\bm{g}}_k- \bm{\p w}_k\|_{\infty}^2}} \stackrel{(\gamma)}\leq \frac{\sqrt{2} I}{\eta} {\sqrt{ \sum_{t=1}^{T}\|\bm{g}_t+\bm{w}_t - \Tilde{\bm{g}}_t- \bm{\p w}_t\|_{\infty}^2} }
\end{align*}
where $(\alpha)$ uses that the minimum between two numbers is less than their harmonic mean; $(\beta)$ assumes that $\eta^2\leq 1/4$ or $\eta\leq 1/2$; and in $(\gamma)$ we applied the identity \cite[Lemma 3.5]{bib:auer-adaptive2002}. To conclude, it suffices to observe that $r_{1:T-1}(\bm x^\star)$ can be upper bounded as follows:
\begin{align*}
r_{1:T-1}(\bm x^\star)&=\frac{\eta_{1:T-1}}{2}\left(I\log J +\sum_{i\in\c I}\sum_{j\in\c J}x_{ij}\log x_{ij}\right) \leq \frac{\eta_{1:T-1}}{2}I\log J\\
&=\frac{\eta I\log J}{2}{ \sqrt{\sum_{t=1}^{T-1}\|\bm{g}_t+\bm{w}_t-\Tilde{\bm{g}}_t - \bm{\p w}_t\|_{\infty}^2} }\leq \frac{\eta I\log J}{2}{ \sqrt{\sum_{t=1}^{T}\|\bm{g}_t+\bm{w}_t-\Tilde{\bm{g}}_t- \bm{\p w}_t \|_{\infty}^2} }
\end{align*}

\subsection{Proof of Proposition \ref{prop:closed-form1}}
Iteration \eqref{dual-update-2a} requires the solution of a convex optimization problem. Since we use non-proximal regularizers, we can provide a closed-form expression using the KKT conditions \cite[Chapter 4]{bib:bertsekas-nonlinear}. In detail, in order to calculate $\bm \theta_{t+1}$, we need to solve:
\begin{align*}
\min_{\bm x} \quad & \frac{\sigma_{1:t}}{2}\|\bm \theta\|^2 + \bm{\theta}^\top \left(\bm \kappa_{1:t}+\bm{\p \kappa_{t+1}} \right) \\
\text{s.t.} \quad&\theta_{i}\geq \theta_i^l, \quad \forall i\in\mathcal{I}, \\
&\theta_{i}\leq \theta_i^u, \quad\forall i\in\mathcal{I}.
\end{align*}
where $\theta_i^l$ and $\theta_i^u$ are the lowest and largest values the dual variables can attain (and depend on the maximum utility values). First, we define the vectors $\bm{\omega}_t \doteq \bm{\kappa}_{1:t} + \Tilde{\bm{\kappa}}_{t+1}\doteq (\omega_{it}, i\in\mathcal I)$, $\bm \theta^l=(\theta_i^l, i\in\c I)$, $\bm \theta^u=(\theta_i^u, i\in\c I)$; and introduce the non-negative dual variables $\bm \lambda$ and $\bm \mu$ to relax the respective constraints and define the Lagrangian:
\begin{align*}
	\mathcal{L}(\bm{\theta}, \bm{\lambda}, \bm{\mu}) = \frac{\sigma_{1:t}}{2}\|\bm{\theta}\|^2 + \bm{\theta}^\top \bm{\omega}_t + \bm{\lambda}^\top (\bm{\theta}^l-\bm \theta) + \bm{\mu}^\top (\bm{\theta}-\bm \theta^u).
\end{align*}
Applying the KKT conditions we can write for the optimal solution $\bm \theta^\star$, $\bm \lambda^\star$, and $\bm \mu^\star$:
\begin{enumerate}
\item Stationarity:
\begin{align*}
\nabla_{\bm \theta}\mathcal{L}(\bm{\theta}, \bm{\lambda}, \bm{v}) = 0 \ \ 
\Rightarrow \ \ \sigma_{1:t}\theta_i^\star + \omega_{it} - \lambda_{i}^\star + \mu_i^\star = 0, \quad \forall i\in\mathcal{I}.
\end{align*}
\item Complementary slackness: ${\lambda_i}^{\star} ({\theta_i^l - \theta_i^\star})= 0, \ \ \text{and} \ \ {\mu_i}^{\star}(\theta_i^\star - \theta_i^u)= 0, \ \ \ \forall i\in\mathcal I$.
\item Primal feasibility: $\theta_i^l \leq \theta_i^\star \leq \theta_i^u, \quad\forall i\in\mathcal{I}$.
\item Dual feasibility: $\lambda_i^\star \geq 0, \mu_i^\star \geq 0, \quad\forall i \in \mathcal{I}$.
\end{enumerate}

Using the above conditions and exploring the different cases for satisfying the complementary slackness conditions, we can see from the proposed expression in \eqref{eq:closed-form-quad}, that indeed $\theta_{i,t+1}^\star$ can admit the following values:
\begin{itemize}
\item $\theta_{i,t+1}^\star = \frac{-\omega_{it}}{\sigma_{1:t}} \Rightarrow$ All 4 conditions are satisfied by setting $\lambda_i^\star = 0, \mu_i^\star = 0$.
\item $\theta_{i,t+1}^\star = \theta_i^l \Rightarrow \theta_i^l \geq \frac{-\omega_{it}}{\sigma_{1:t}}$. Setting $\mu_i^\star = 0$, $\lambda_i^\star = \sigma_{1:t}\theta_i^l + \omega_{it} \geq 0$ satisfies all 4 conditions.
\item $\theta_{i,t+1}^\star = \theta_i^u \Rightarrow \theta_i^u \leq \frac{-\omega_{it}}{\sigma_{1:t}}$. Setting $\lambda_i^\star = 0$, $\mu_i^\star = -\sigma_{1:t}\theta_i^u - \omega_{it} \geq 0$ satisfies all 4 conditions.
\end{itemize}
\subsection{Proof of Proposition \ref{prop:closed-form2}}

The update \eqref{primal-update-2} involves solving the convex problem (dropping the time index of variables):
\begin{align*}
\min_{\bm x} \quad &r_{1:t}(\bm{x}) - \bm{x}^\top \parentheses{\bm{g}_{1:t}+\bm{w}_{1:t} + \Tilde{\bm{g}}_{t+1}+ \Tilde{\bm{w}}_{t+1} } \\
\text{s.t.} \quad&\sum_{j\in\mathcal{J}}x_{ij} = 1, \quad\forall i\in\mathcal{I}, \\
&x_{ij}\geq 0, \qquad \ \ \ \forall i\in\mathcal{I}, \forall j\in\mathcal{J}.
\end{align*}
First, we define $\bm{\omega}_t \doteq \bm{g}_{1:t} +\bm{w}_{1:t} +  \Tilde{\bm{g}}_{t+1}+  \Tilde{\bm{w}}_{t+1}\doteq (\omega_{ijt}, i\in\mathcal I, j\in\mathcal J)$ and introduce the dual variable vectors $\bm \lambda\in \mathbb R_+^{I\cdot J}$ and $\bm \mu \in \mathbb R^I$, to define the Lagrangian:
\begin{align*}
\mathcal{L}(\bm{x}, \bm{\lambda}, \bm{\mu}) = r_{1:t}(\bm{x}) - \bm{x}^\top \bm{\omega}_t - \bm{\lambda}^\top \bm{x} + \sum_{i\in\mathcal{I}}\mu_i\bigg(\sum_{j\in\mathcal{J}}x_{ij} - 1\bigg).
\end{align*}
The KKT conditions are:
\begin{enumerate}
\item Stationarity: $\nabla_{\bm x}\mathcal{L}(\bm{x}, \bm{\lambda}, \bm{\mu}) = \bm 0$, which yields the following:
\begin{align*}
&\text{if  } \ \  x_{ij}>0:\quad \frac{\eta_{1:t}}{2}\left(\log x_{ij}+1\right) - \omega_{ijt} - \lambda_{ij} + \mu_i = 0, \quad \forall i\in\mathcal{I}, \ \forall j\in\mathcal{J},\\
&\text{if  } \ \  x_{ij}=0:\quad - \omega_{ijt} - \lambda_{ij} + \mu_i = 0, \quad \qquad \qquad \quad \qquad \forall i\in\mathcal{I}, \ \forall j\in\mathcal{J}.
\end{align*}
\item Complementary slackness: $\lambda_{ij}x_{ij} = 0$, $\forall i\in\c I, j\in\c J$.
\item Primal feasibility: $x_{ij}\geq 0, \ \  \sum_{j\in\mathcal{J}}x_{ij} = 1, \ \ \forall i\in\mathcal{I}, \forall j\in\mathcal{J}$
\item Dual feasibility: $\lambda_{ij} \geq 0, \forall i\in\mathcal I, j\in\mathcal J$.
\end{enumerate}
Setting $\bm{\lambda} = \bm{0}$, and solving for $\mu_i$ in each equation, we obtain:
\[
\mu_i =-\frac{\eta_{1:t}}{2}\left(\log x_{ij}+1\right) + \omega_{ijt}, \quad \forall i\in\mathcal{I}, \ \forall j\in\mathcal{J},
\]
and replacing the proposed expression for $\bm{x}$ from~\eqref{eq:closed-form}, we get $\forall i \in \mathcal I$:
\[
\mu_i =-\frac{\eta_{1:t}}{2}\left( \frac{2\omega_{ijt}}{\eta_{1:t}} -\log\left(\sum_{j\in\mathcal J} \exp\left({\frac{2\omega_{ijt}}{\eta_{1:t}}}\right) \right) +1\right) + \omega_{ijt}=\frac{\eta_{1:t}}{2}\log\left(\sum_{j\in\mathcal J} \exp\left({\frac{2\omega_{ijt}}{\eta_{1:t}}}\right) \right) -\frac{\eta_{1:t}}{2},
\]
where notice that $\mu_i$ contains summations over all elements of $\mathcal J$, and hence its value does not depend on the variable derivative $j$-wise. Therefore, this solution satisfies all KKT conditions, since the primal variables and the $\lambda_{ij}$ variables are nonegative, and it holds:
\[
\sum_{j\in\mathcal J} \frac{\exp\parentheses{ 2\omega_{ijt}/\eta_{1:t}    }}{\sum_{j\in\mathcal{J}}\exp\parentheses{2\omega_{ijt}/\eta_{1:t}}}=1,  \ \ \forall i\in\mathcal I, \ j\in\mathcal J,
\]

\subsection{Proof of Theorem~\ref{eq:proxy-cost-regret}}\label{appendix:proof-theorem-2}
Using \cite[Lemma 2]{bib:tareq_fairness}, we can write:
\begin{align}
F_{\alpha}\left( \frac{1}{T}\sum_{t=1}^T\bm u_t(\bm y_t) \right) &= \min_{\bm \theta\in\Theta}\left\{ (-F_\alpha)^\star(\bm\theta) - \bm\theta\cdot\frac{1}{T}\sum_{t=1}^T\bm u_t(\bm y_t) \right\} = \min_{\bm \theta\in\Theta}\left\{ \frac{1}{T}\sum_{t=1}^T (-F_\alpha)^\star(\bm\theta) - \bm\theta^\top\bm u_t(\bm y_t) \right\}.\notag
\end{align}
Following the definition of $G_\alpha(\{\bm y_{t}\}_t)$ and combining it with the above result, we can write:
\begin{align}
G_\alpha(\{\bm y_{t}\}_t) &= \! F_\alpha\left(\frac{1}{T}\sum_{t\in\mathcal{T}}u_t(\bm y_t)\right)\! - \frac{1}{T}\sum_{t\in\mathcal{T}}c_t(\bm y_t)\!= \min_{\bm \theta\in\Theta}\left\{ \frac{1}{T}\left[\sum_{t=1}^T (-F_\alpha)^\star(\bm\theta) \!- \bm\theta^\top\bm u_t(\bm y_t)\right] \right\} \!- \frac{1}{T}\sum_{t=1}^T c_t(\bm y_t)\nonumber\\
&= \min_{\bm \theta\in\Theta}\left\{ \frac{1}{T}\left[\sum_{t=1}^T (-F_\alpha)^\star(\bm\theta) - \bm\theta^\top\bm u_t(\bm y_t) - c_t(\bm y_t)\right]\right\}= \min_{\bm \theta\in\Theta}\left\{ \frac{1}{T}\sum_{t=1}^T \Psi_{t}^c(\bm \theta, \bm y_t) \right\}.\label{eq:psi-sum-cost}
\end{align}
Denoting with $\bm{\mathcal{R}}_{T}^{y,c}$ the primal-space regret (in analogy with \eqref{primal-space-regret-xx}), we have:
\begin{align}
&\frac{1}{T}\sum_{t=1}^T \Psi_{t}^c(\bm \theta_t,\bm y_t) \!+ \frac{\bm{\mathcal{R}}_{T}^{y,c}}{T} \!=\! \frac{1}{T}\sum_{t=1}^T \Psi_{t}^c(\bm \theta_t,\bm y_\star) \!=\!\frac{1}{T}\sum_{t=1}^T (-F_\alpha)^\star(\bm\theta_t) \!- \frac{1}{T}\sum_{t=1}^T \bm \theta_t^\top\bm u_t(\bm{y}^\star) \!-\! \frac{1}{T}\sum_{t=1}^T c_t(\bm{y}^\star) \nonumber\\
&\stackrel{(\gamma_1)}\geq \! (-F_\alpha)^\star\left( \frac{1}{T}\sum_{t=1}^T \bm\theta_t \right) - \frac{1}{T}\sum_{t=1}^T \bm \theta_t^\top\bm u_t(\bm{y}^\star) - \frac{1}{T}\sum_{t=1}^T c_t(\bm{y}^\star) \nonumber\\
&= (-F_\alpha)^\star(\Bar{\bm \theta}) \!- \Bar{\bm \theta}\cdot\left( \frac{1}{T}\sum_{t=1}^T \bm u_t(\bm{y}^\star) \right) \!- \frac{1}{T}\sum_{t=1}^T (\bm \theta_t - \Bar{\bm \theta})^\top\cdot\bm u_t(\bm{y}^\star) - \frac{1}{T}\sum_{t=1}^T c_t(\bm{y}^\star) \nonumber\\
&\geq  \min_{\bm \theta\in\Theta}\left\{ \frac{1}{T}\sum_{t=1}^T (-F_\alpha)^\star(\bm\theta) - \bm\theta^\top\bm u_t(\bm{y}^\star) - c_t(\bm{y}^\star)\right\} - \frac{1}{T}\sum_{t=1}^T (\bm \theta_t - \Bar{\bm \theta})^\top\bm u_t(\bm{y}^\star) \nonumber\\
&=\min_{\bm \theta\in\Theta}\left\{ \frac{1}{T}\sum_{t=1}^T \Psi_{t}^c(\bm \theta_t, \bm{y}^\star) \right\} - \frac{1}{T}\sum_{t=1}^T (\bm \theta_t - \Bar{\bm \theta})^\top\bm u_t(\bm{y}^\star) \stackrel{(\gamma_2)}= G_\alpha(\bm{y}^\star) - \frac{1}{T}\sum_{t=1}^T (\bm \theta_t - \Bar{\bm \theta})^\top\bm u_t(\bm{y}^\star),\label{eq:proof-sigma}
\end{align}
where~$(\gamma_1)$ follows from Jensen's inequality and the convexity of $(-F_\alpha)^\star$, and in~($\gamma_2$) we used~\eqref{eq:psi-sum-cost}. Next, we define the dual-space regret $\bm{\mathcal{R}}_{T}^{\theta,c}$ (in analogy to \eqref{dual-space-regret-xx}) and relate it to function $G_\alpha$, namely:
\begin{align}
\bm{\mathcal{R}}_{T}^{\theta,c} &= \sum_{t=1}^T \Psi_{t}^c(\bm \theta_t, \bm y_t) - \sum_{t=1}^T \Psi_{t}^c(\bm \theta, \bm y_t) \quad \text{for every } \bm \theta \in \Theta \nonumber\\
&\stackrel{(\gamma_3)}= \sum_{t=1}^T \Psi_{\alpha, t}^c(\bm \theta_t, \bm y_t) - T G_\alpha([\bm y_t]). \label{eq:dual-regret-g}
\end{align}
where $(\gamma_3)$ follows from the definition of $G_\alpha(\{\bm y_t\}_t)$ as the minimizer of the averaged proxy function values w.r.t. $\bm \theta \in \Theta$, see \eqref{eq:psi-sum-cost}. Now, we can combine~\eqref{eq:proof-sigma} and~\eqref{eq:dual-regret-g}, and write:
\begin{align}
G_\alpha([\bm y_{t}]) + \frac{\bm{\mathcal{R}}_{T}^{\theta,c}}{T}&=\frac{1}{T} \sum_{t=1}^T \Psi_{ t}^c(\bm\theta_t, \bm y_t)\geq G_\alpha(\bm y^\star) - \frac{1}{T}\sum_{t=1}^T (\bm \theta_t - \Bar{\bm \theta})^\top\bm u_t(\bm{y}^\star) - \frac{\bm{\mathcal{R}}_{T}^{y,c}}{T}. \notag
\end{align}
Rearranging, we arrive at the main result of the theorem.

\subsection{Additional Experiments and Evaluation Results}
This section includes further results that could not be included in the main part of the paper due to lack of space. All the results presented in this section are obtained using the O-RAN compliant experimental platform presented in Sec.~\ref{sec:applications}.

\begin{figure*}[t]
 
\begin{subfigure}[t]{0.3\textwidth} 
\centering
\includegraphics[width=\textwidth, page=8]{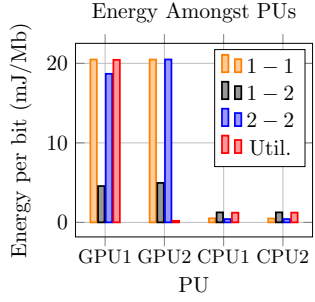}

\end{subfigure}
\vspace{-2mm}
\caption{\small{Energy dispersion amongst PUs for different $\alpha-\beta$ pairs. \emph{Util.} stands for utilitarian algorithm.}}
\label{fig:appendix2}
\end{figure*}
\begin{figure*}[t]
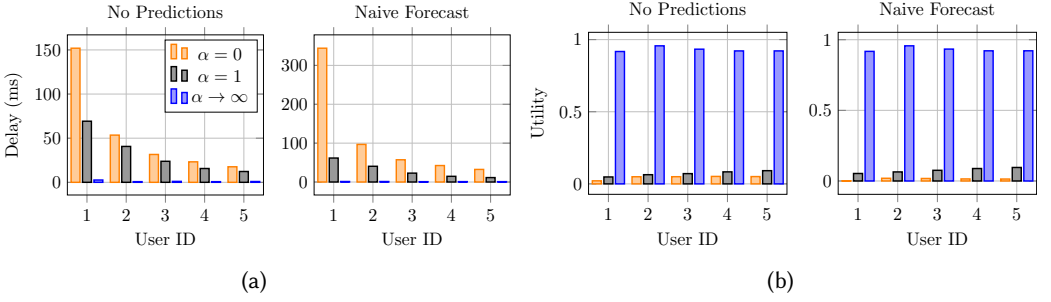

\centering
\begin{subfigure}[t]{0.49\textwidth}
\centering
\includegraphics[width=\textwidth, page=4]{eval_figs.pdf}
\caption{}
\end{subfigure}%
\hfill 
\begin{subfigure}[t]{0.49\textwidth} 
\centering
\includegraphics[width=\textwidth, page=5]{eval_figs.pdf}
\caption{}
\end{subfigure}
\vspace{-1mm}
\caption{\small{\textbf{(a):} Delay dispersion amongst users for different $\alpha$, non-optimistic FTRL (left) and OFTRL with Naive forecast (right). \textbf{(b):} Utility (probability of empty buffer) dispersion amongst users for different $\alpha$ values, non-optimistic FTRL (left) and OFTRL with Naive forecast (right).}}
\label{fig:appendix1}
\end{figure*}
We start with the assignment (compute control) policy. Fig.~\ref{fig:appendix2} shows the dispersion of throughput amongst vBSs and energy of PUs using the horizon fair, and utilitarian algorithms as in Sec.~\ref{sec:applications-fair-balanced}.
In Fig.~\ref{fig:appendix2}, the difference between $\alpha=1, \beta=1$ (orange) and $\alpha=1, \beta=2$ (black) indicates that as $\beta$ is increased, energy is distributed more fairly among servers with the horizon fair algorithm. Also note that the horizon fair algorithm disperses the energy fairly and uses both of the GPUs, whereas the utilitarian algorithm chooses to use only the faster and cheaper GPU to reduce its energy with an unfair use. 

Next, we provide additional results on the radio control policy (\texttt{minTB}). Fig.~\ref{fig:appendix1} shows the dispersion of actual measured delay, and percentage of the user buffer being empty, amongst different users when using the non-optimistic FTRL and OFTRL with Naive forecast~\cite{bib:hyndman2018forecasting} algorithms. 
We consider the configuration of Scenario 2, detailed in Sec.~\ref{sec:applications-minTB}.
We see that the delays and utilities of the users are dispersed more fairly as $\alpha$ increases.

\subsection{Derivation of Convex Utility and Cost Functions in Sec.~\ref{sec:applications-minTB}}

Our policy makes decisions every $100$ ms, and we need to approximate the probability of empty buffer and expected energy cost between decisions, depending on $\bm y_t, \bm b_t, \bm \rho_t,$ and $\bm s_t$.
To approximate $\bm u_t$, we assume that data generation of each user $b_{it}$ follows a Poisson distribution, where the times between data generations are exponentially distributed with the parameter $1/b_{it}$ and each data generation consists of $\rho_{it}$ number of bits. We stress, however, that this is a non-binding assumption (other models can be studied), and that we allow the parameters of the distribution to change arbitrarily (based on the adversary model) across the different slots.

We designed the system such that the user data is transmitted when the number of bits in the buffer or each user $i\in \c I$ exceeds the threshold $y_{it}$, where we denote the number of bits as $B_{it} = b_{it}\rho_{it}$. Here, for notational convenience we drop the subscripts $i$ and $t$ and derive a utility function for each user $i$ in each time slot $t$.
Additionally, we define a new time variable $\tau \in [0, 1)$ within the time slot $t$ and denote the number of bits in the user buffer at time $\tau$ as $B_\tau$.
Next, we calculate $\mathbf{Pr}\left( B_\tau > 0\text{ and } \tau<\frac{y}{b\rho}\right)$ as:
\begin{align}
    \mathbf{Pr}\left( B_\tau > 0\text{ and } \tau<\frac{y}{b\rho}\right) &= \mathbf{Pr}\left( \tau \geq \text{time of the first bit generation and } \tau < \frac{y}{b\rho}\right) \notag \\
    &= \mathbf{Pr}\left(\text{time of the first bit generation} \leq \tau < \frac{y}{b\rho}\right) =\int_{0}^{\frac{y}{b\rho}}\left( 1 - e^{-b\tau} \right)\,d\tau \notag
\end{align}
Now, we calculate the portion of time the user has non-empty buffer when $\tau\in[0, 1)$ as 
\begin{align}
    1 - u(y) &= \mathbf{Pr}\left( B_\tau > 0\text{ and } \tau<\frac{y}{b\rho}\right)\frac{b\rho}{y} = \frac{\rho e^{-\frac{y}{\rho}} + y - \rho}{y} \notag
\end{align}
Thus, we have $u(y) = \frac{\rho}{y}\left( 1 - e^{-\frac{y}{\rho}} \right)$. Note that when $y\rightarrow 0, u(y) = 1$ and for larger $y$ values $u(y) \approx \frac{\rho}{y}$.
We approximate the cost function similarly, since the probability of empty buffer is an indicator of the rate of data transmissions, i.e., the buffer is empty right after the transmission up until the first data generation after the last transmission. Therefore, we can approximate the number of data transmissions by $b\cdot u(y)$. We then multiply this with the cost multiplier due to SNR, and the cost scaling parameter to calculate hardware cost induced by user $i$ as $c(y) = \varphi \beta(s)b\frac{\rho}{y}\left( 1 - e^{-\frac{y}{\rho}}\right)$. We sum this cost function for all users to calculate the HA cost. We show that the functions are convex as $u''(y) \geq 0$ is always satisfied. 
\begin{figure}
    \centering
    \includegraphics[width=0.75\textwidth, page=11]{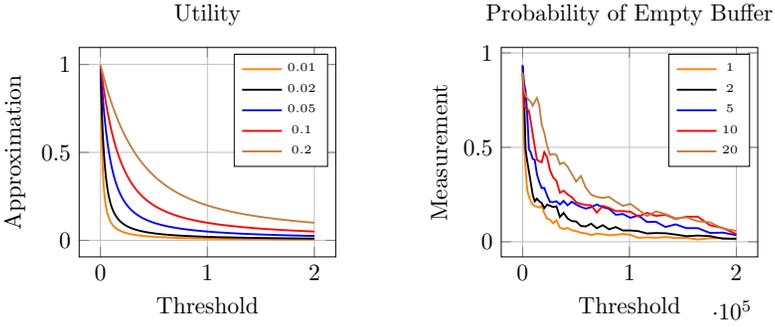}
    \caption{Comparison of $u(y)$ with the real measurement of probability of empty buffer. Legends are $\rho$ (left) and traffic multiplier (right).}
    \label{fig:utility-approximation}
\end{figure}

Fig.~\ref{fig:utility-approximation} demonstrates a comparison between our approximation function $u(y)$ and the real measurement of empty buffer probability gathered using our testbed.
Here, in the approximation we modify $\rho$, and in the real measurements we multiply the user traffics to increase the demand.

\subsection{Convexity of Functions in Sec.~\ref{sec:applications-fair-balanced}}
First, we prove that the utility function in Sec.~\ref{sec:applications-fair-balanced} is indeed concave.
\begin{align}
    u_{ij}(\bm x) = x_{ij}\lambda_{i}\cdot \min\left\{1,\ 1 - \frac{1}{C_{j}}\left( \sum_{k\in\mathcal I} \frac{x_{kj}\lambda_{k}}{n_{k}}\big( \zeta_{k}^jn_{k}+o_{k}^j \big)-C_{j}\right)\right\}  \notag
\end{align}
We do not use the time subscript $t$ for notational simplicity. 
Note that the piecewise minimum of two concave functions is also concave, and it is sufficient to prove that both functions inside min$\{\}$ after multiplied with $x_{ij}\lambda_i$ are concave. The LHS, $x_{ij}\lambda_i$ is a linear function, thus concave. Therefore, it is sufficient to show that: 
\begin{align}
    x_{ij}\lambda_i - \frac{x_{ij}\lambda_i}{C_{j}}\left( \sum_{k\in\mathcal I} \frac{x_{kj}\lambda_{k}}{n_{k}}\big( \zeta_{k}^jn_{k}+o_{k}^j \big)-C_{j}\right), \notag
\end{align}
is concave. We calculate the Hessian matrix $\bm H$ of 
\[{f_{ij}(\bm x) = x_{ij}\lambda_i - \frac{x_{ij}\lambda_i}{C_j}\left( \sum_{k\in\mathcal{I}}\frac{x_{kj}\lambda_k}{n_k}\left( \zeta_{k}^jn_{k}+o_{k}^j \right) - C_j\right)},
\]
as ${\bm H \doteq \nabla \bm g}$ where ${\bm g \doteq \nabla f_{ij}(\bm x)}$, $\bm H \in \mathbb{R}^{(I\cdot J)\times(I\cdot J)}$ and $\bm g \in \mathbb{R}^{(I\cdot J)}$. Denoting $\frac{\lambda_i}{n_i}\left( \zeta_i^jn_i+o_i^j \right) \doteq k_i$, we can write:
\begin{align}
    g_{ij} &= 2\lambda_i - x_{ij}\frac{2\lambda_i k_i}{C_j} - \frac{\lambda_i}{C_j}\sum_{i'\neq i}x_{i'j}k_{i'}, \notag \\
    g_{i'j} &= - \frac{x_{ij}\lambda_i k_{i'}}{C_j}, \notag \\
    g_{ij'} &= 0, \notag \\
    g_{i'j'} &= 0, \label{eq:grad}
\end{align}
for the values of $\bm g$ where $i'\neq i,j'\neq j$. The Hessian matrix $\bm H = \nabla \bm g$ has the following elements:
\begin{align}
    H_{ij, ij} &= - \frac{2\lambda_i k_i}{C_j}, & H_{ij, i'j} &= -\frac{\lambda_i k_{i'}}{C_j}, & H_{ij, ij'} &= 0, & H_{ij, i'j'} &= 0, \notag \\
    H_{i'j, ij} &= -\frac{\lambda_i k_{i'}}{C_j}, & H_{i'j, i'j} &= 0, & H_{i'j, ij'} &= 0, & H_{i'j, i'j'} &= 0, \notag \\
    H_{ij', ij} &= 0, & H_{ij', i'j} &= 0, & H_{ij', ij'} &= 0, & H_{ij', i'j'} &= 0, \notag \\
    H_{i'j', ij} &= 0, & H_{i'j', i'j} &= 0, & H_{i'j', ij'} &= 0, & H_{i'j', i'j'} &= 0. \label{eq:hessian}
\end{align}
Hence, $\bm H$ is a negative semi definite matrix, thus $f_{ij}(\bm x)$ is a concave function of $\bm x$.

The cost efficiency function $h_{j}(\bm x_j)$ is a linear function of $\bm x$, hence concave.

\end{document}